\documentclass{article}

\usepackage[
    bibliosources={refs.bib}
]{pomegranate}

\providecommand{\mathbbm}[1]{\mathds{#1}}
\usepackage{algorithm}
\usepackage{algpseudocode}
\usepackage{nicefrac}
\usepackage{amsmath}
\usepackage{titlesec}

\newcommand{\bE}{\mathbb{E}}

\newcommand{\Prb}{\mathbb{P}}

\newcommand{\cE}{\mathcal{E}}

\newcommand{\Tr}{\mathsf{Tr}}

\newcommand{\polylog}{\mathrm{polylog}}

\newcommand{\Otilde}{\widetilde{O}}
\newcommand{\estep}{\varepsilon_{\mathrm{step}}}

\newcommand{\dotp}[2]{\left\langle #1, #2 \right\rangle}

\newcommand{\Uniform}{\mathsf{Uniform}}

\newcommand{\KL}[2]{D_\mathsf{KL}\left(#1 \ || \ #2\right)}
\newcommand{\TV}[1]{d_\mathsf{TV}\left(#1\right)}

\newcommand{\op}{\mathsf{op}}

\newcommand{\bR}{\mathbb{R}}
\newcommand{\cN}{\mathcal{N}}
\newcommand{\SL}{\text{SL}}

\DeclareMathOperator{\atanh}{atanh}
\DeclareDelimiter{\sign}[sign]{\lparen}{\rparen}
\DeclareMathOperator{\sech}{sech}
\newcommand{\opnorm}[1]{\left\lVert #1 \right\rVert_{\mathrm{op}}}

\DeclarePairedDelimiterX{\ip}[2]{\langle}{\rangle}{#1, #2}
\DeclareDelimiter{\abs}{\lvert}{\rvert}

\providecommand{\R}{\mathbb{R}}
\providecommand{\E}{\mathbb{E}}
\providecommand{\Prb}{\mathbb{P}}

\providecommand{\op}{\mathrm{op}}
\providecommand{\opnorm}[1]{\norm{#1}_{\op}}

\providecommand{\AMP}{\mathrm{AMP}}
\providecommand{\NGD}{\mathrm{NGD}}
\providecommand{\TAP}{\mathrm{TAP}}
\providecommand{\ONS}{\mathrm{ONS}}

\providecommand{\sech}{\operatorname{sech}}
\providecommand{\arctanh}{\operatorname{arctanh}}
\titleformat{\section}
  {\normalfont\Large\bfseries}
  {\thesection}
  {1em}
  {}

\makeatletter
\renewcommand*\l@section[2]{%
  \ifnum \c@tocdepth >\z@
    \addpenalty\@secpenalty
    \addvspace{1.0em \@plus\p@}%
    \setlength\@tempdima{2.3em}%
    \begingroup
      \parindent \z@
      \rightskip \@pnumwidth
      \parfillskip -\@pnumwidth
      \leavevmode \bfseries
      \advance\leftskip\@tempdima
      \hskip -\leftskip
      #1\nobreak\hfil \nobreak\hb@xt@\@pnumwidth{\hss #2}\par
    \endgroup
  \fi}
\makeatother

\title{Parallel Sampling from the Ising $p$-Spin Model}

\author[1]{Nima Anari}
\author[1]{Aniket Das}
\author[1]{Alireza Haqi}
\affil[1]{Stanford University, \url{{anari,aniketd, ahaqi}@stanford.edu}}
\date{}

\begin{document}
\maketitle
\begin{abstract}
We study the parallel complexity of sampling from the high-temperature Ising mixed $p$-spin Gibbs measure, a canonical instance of a mean-field spin glass on the hypercube $\{\pm 1\}^n$. We propose two different algorithms for this problem, corresponding to two different regimes of accuracy. 

Our first algorithm is a parallel implementation of a Markov chain known as block dynamics, combined with an approximate rejection sampling step that uses an Ising model in a novel way as a proposal distribution to approximate the quadratic interaction terms of the $p$-spin Hamiltonian. For any $\varepsilon > 0$, this algorithm runs in $n^{\nicefrac{1}{3}}\polylog(\nicefrac{n}{\varepsilon})$ parallel time with $\operatorname{poly}(n, \log(\nicefrac{1}{\varepsilon}))$ work, and outputs a sample whose law is $\varepsilon$-close to the $p$-spin measure in total variation distance.

Our second algorithm uses Picard iterations to parallelize the Algorithmic Stochastic Localization (ASL) process of El Alaoui, Montanari, and Sellke (2025), and for any $\varepsilon > \varepsilon_n$, takes $\polylog(\nicefrac{n}{\varepsilon})$ parallel time and $\poly(\nicefrac{n}{\varepsilon})$ work to produce a sample that is $\varepsilon$-close to the $p$-spin measure in the normalized 2-Wasserstein metric. Here, $\varepsilon_n > 0$ is a threshold that goes to $0$ as $n \to \infty$. Our result constitutes a doubly exponential improvement in the $\varepsilon$ dependence of the runtime and an exponential improvement in the $\varepsilon$ dependence of the total work when compared to na\"ive ASL, whose runtime scales as $\exp(\poly(\nicefrac{1}{\varepsilon}))$.
\end{abstract}

\tableofcontents


\section{Introduction}
Efficient sampling from a probability distribution which is specified only up to an unknown normalizing constant is a central algorithmic problem across computer science, statistical physics, and machine learning. Generally, given a finite state space $\Omega$ and a Hamiltonian $H : \Omega \to \R$, the target distribution is a Gibbs measure of the following form:
\begin{align}
    \mu(x) = \frac{1}{Z} \exp(H(x)), \qquad Z = \sum_{x \in \Omega} \exp(H(x))
\end{align}
where the normalizing constant $Z$, also known as the partition function, is unknown and typically hard to compute. This setting includes classical problems in approximate counting \cite{jerrum1986random} and randomized algorithms such as approximating the volume of convex bodies \cite{dyer1991random} and approximating matrix permanents \cite{jerrum2004polynomial}, Markov-chain Monte Carlo methods in statistical physics \cite{parisi1981correlation}, and posterior sampling in high-dimensional statistical models. 

Our work studies the problem of sampling from the Gibbs measure of the Ising mixed $p$-spin model, which is a Gibbs measure on the hypercube $\{\pm 1\}^n$ whose Hamiltonian is a random polynomial with Gaussian coefficients.
\begin{align}
\label{eq:p-spin-model}
    H(x) = \sum_{p=2}^{P} \frac{\beta_p \sqrt{p!}}{n^{\tfrac{p-1}{2}}} \cdot \sum_{J \subseteq [n], |J|=p} G_J \ \prod_{k \in J}x_k + 
    \ \langle x, h \rangle, \qquad \mu(x) = \frac{1}{Z}\exp(H(x))
\end{align}
Here, $P \geq 2$ is an integer, the disorder $G_{J} \stackrel{i.i.d}{\sim} \cN(0,1)$ for every choice of $J$, and $h$ is an arbitrary external field independent of $G$. The $p$-spin Hamiltonian can be equivalently characterized as a Gaussian process on $\{\pm 1\}^n$ with $\E_G{H(x)} = \dotp{h}{x}$ and $\operatorname{Cov}[H(x)H(y)] = N \xi(\nicefrac{\dotp{x}{y}}{N})$, where $\xi(t) = \sum_{p=2}^{P} \beta^2_p t^p$ denotes the mixture function. The mixed $p$-spin model is a canonical instance of a mean-field spin glass, which plays a central role in several areas of probability, mathematical physics, and average-case complexity \cite{mezard1988spin, talagrand2010mean,panchenko2012sherrington, montanari2024friendly}. 
Notable special cases include the pure $p$-spin model where only one $\beta_p$ is nonzero, and the Sherrington-Kirkpatrick (SK) model where $P=2$.

We consider the problem of approximate sampling from the $p$-spin Gibbs measure in the quenched setting, where the disorder $G$ is sampled once and then fixed, and the algorithm is tasked with approximately sampling from $\mu$ given random access to the disorder coefficients (which can be stored with $\poly(n)$ memory). Given an error parameter $\varepsilon > 0$, the goal is to design an algorithm that runs in time $\poly(n, \nicefrac{1}{\varepsilon})$ for any realization of $G$, and produces an output distributed as $\hat{\mu}$ such that $\operatorname{dist}(\hat{\mu}, \mu) \leq \varepsilon$ for typical realizations of $G$. Here, $\operatorname{dist}$ denotes any appropriate metric between probability distributions.

Under this setting, \cite{adhikari2022spectral} and \cite{anari2024universality} analyze Glauber dynamics \cite{glauber1963time}, a canonical Markov chain for sampling from discrete distributions, which, at every step, resamples a randomly chosen coordinate according to its conditional marginal, while fixing the remaining coordinates. The authors identify a high-temperature regime for the $p$-spin model under which Glauber dynamics efficiently samples from $\mu$. Specifically, they show that when $ \sum_{p=2}^{P} \beta_p \sqrt{p^3 \ln(p)} = O(1)$, Glauber dynamics takes $O(n \ln(\nicefrac{n}{\varepsilon}))$ time to output a sample distributed as $\hat{\mu}$ satisfying $d_{\mathrm{TV}}(\hat{\mu}, \mu) \leq \varepsilon$ with high probability over the randomness of $G$. Since the error parameter $\varepsilon > 0$ can be made as small as desired, we refer to such a guarantee as a \emph{High Accuracy Guarantee}.

In another line of work, \cite{el2025sampling-p-spin} propose a different approximate sampler for the $p$-spin model with a zero external field (i.e., $h=0$ in \cref{eq:p-spin-model}) based on an approximate discretization of an It\^o diffusion process known as Eldan's stochastic localization \cite{eldan2013thin}. Their algorithm, named Algorithmic Stochastic Localization (ASL), offers a weaker sampling guarantee, albeit supporting a broader range of temperatures than that covered by the high-temperature condition defined above. Concretely, there exists an error threshold $\varepsilon_n = o_n(1)$ such that for any $\varepsilon > \varepsilon_n$, ASL runs in time $O(\poly(n)\exp(\poly(\nicefrac{1}{\varepsilon})))$ and with high probability over $G$, approximately samples from $\mu$ up to error $\varepsilon$ in the normalized 2-Wasserstein metric. In other words, the output $\hat{x}$ of their algorithm can be coupled with some $x^{\star} \sim \mu$ such that $n^{-1} \E{\|\hat{x} - x^{\star}\|} \leq \varepsilon^2$. We call such a guarantee a \emph{Low Accuracy Guarantee} for two key reasons. Firstly, for a fixed problem instance, the accuracy of ASL can only be certified above a threshold $\varepsilon > \varepsilon_n$. Secondly, their guarantees are presented in the normalized 2-Wasserstein metric, which is strictly weaker than the total variation metric for the range of $\varepsilon$ tolerated by their algorithm. 

A common drawback of both algorithms considered above is that they are inherently sequential. Indeed, Glauber dynamics sequentially updates one coordinate after another in a random order, while Algorithmic Stochastic Localization discretizes the trajectory of a continuous Markov process, which naturally evolves in a sequential fashion. Consequently, these algorithms fail to utilize the full potential of modern hardware, which is generally capable of massive parallelism, both in terms of parallel computation and parallel memory access. This motivates the central problem of our work: 
\begin{center}
    \emph{Can we design an efficient approximate sampling algorithm for the Ising $p$-spin model that can take advantage of parallel hardware? }
\end{center}

Specifically, we consider the problem of sampling from the $p$-spin Gibbs measure in the \Class{PRAM} model of computation, where the algorithm is allowed to query the disorder $G$ in multiple rounds, with each round comprising polynomially many synchronous queries that are executed in parallel. Given an error tolerance $\varepsilon > 0$, our goal is to design an algorithm that approximately samples from $\mu$ up to error $\varepsilon$ in some suitable metric for typical realizations of $G$, while also optimizing for two criteria: \emph{Parallel Runtime} or \emph{Round Complexity}, which denotes the number of rounds of parallel queries made to $G$, and \emph{Work}, which denotes the total number of queries made to $G$. Ideally, for every realization of $G$, we want the parallel runtime to scale at most sublinearly in $n$ and $\varepsilon^{-1}$, and the work to scale at most polynomially in $n$ and $\varepsilon^{-1}$. 

Compared to the complexity of sequential sampling, the theory of parallel sampling for mean-field spin glasses is far less developed. To our knowledge, the only prior parallel algorithm for sampling from the $p$-spin Gibbs measure is the work of \cite{lee2023parallelising}, which considers the $p$-spin model in the high-temperature regime and achieves a high-accuracy guarantee with a parallel runtime of $O(\sqrt{n} \cdot\polylog(\nicefrac{n}{\varepsilon}))$ and $O(\poly(\nicefrac{n}{\varepsilon}))$ work. Beyond this, for the special case of the high-temperature SK model (i.e., $P=2$ and $\beta_2 < \nicefrac{1}{4}$), \cite{chen2025efficient} show that the Restricted Gaussian Dynamics algorithm of \cite{chen2022localization} is a high-accuracy approximate parallel sampler with parallel runtime $O(\polylog(\nicefrac{n}{\varepsilon}))$ and $O(\poly(\nicefrac{n}{\varepsilon}))$ work. 

\subsection{Our Contributions}
In this work, we develop two different algorithms for parallel sampling from the Ising mixed $p$-spin Gibbs measure. The first is a high-accuracy algorithm which samples from the $p$-spin model up to $\varepsilon$ total variation error in $O(n^{\nicefrac{1}{3}} \polylog(\nicefrac{n}{\varepsilon}))$ parallel time and $O(\poly(n, \log(\nicefrac{1}{\varepsilon})))$ work. To our knowledge, this constitutes an improvement both in terms of parallel runtime and work compared to the previous best known result of \cite{lee2023parallelising}, which exhibits a parallel runtime of $O(\sqrt{n}\polylog(\nicefrac{n}{\varepsilon}))$ with $O(\poly(\nicefrac{n}{\varepsilon}))$ work.
\begin{theorem}[High-Accuracy Parallel $p$-spin Sampler,  Informal Version of Theorem \ref{thm:parallel-p-spin-sampler-main}] \label{thm-informal:p-spin-sampler}
Let $\mu$ denote the Ising $p$-spin Gibbs measure as defined in \cref{eq:p-spin-model}. Assume the model is in the high-temperature regime, i.e., $\sum_{p=2}^{P} \beta_p \sqrt{p^3 \ln(p)} \leq C$ for some $C=\Theta(1)$. Then, there exists a parallel sampling algorithm (\cref{alg:parallel-p-spin-sampler}) satisfying the following for any $\varepsilon > 0$.
\begin{itemize}
    \item The algorithm runs in parallel time $\Otilde(n^{\nicefrac{1}{3}} \ln(\nicefrac{1}{\varepsilon}))$ with $O(\poly(n, \log(\nicefrac{1}{\varepsilon})))$ work.  
    
    \item With probability at least $1-\exp(-cn)$ over the randomness of the disorder $G$, the algorithm outputs $\hat{x} \sim \hat{\mu}$ such that $d_{\mathrm{TV}}(\hat{\mu}, \mu) \leq \varepsilon$.
\end{itemize}
\end{theorem}
Our second algorithm obtains a low-accuracy approximate sampling guarantee in the normalized 2-Wasserstein metric, and enjoys a parallel runtime of $O(\polylog(\nicefrac{n}{\varepsilon}))$ with $O(\poly(\nicefrac{n}{\varepsilon}))$ work. This significantly improves upon the $O(\poly(n) \cdot\exp(\poly(\nicefrac{1}{\varepsilon})))$ runtime of \cite{el2025sampling-p-spin} and constitutes the first parallel algorithm for sampling from the $p$-spin model with $\polylog(\nicefrac{n}{\varepsilon})$ parallel runtime in any metric. 
\begin{theorem}[Low-Accuracy Parallel $p$-spin Sampler, Informal Version of Theorem \ref{thm:p-spin-model-picard-asl}] \label{thm:informal-picard-asl}
Let $\mu$ denote the Ising $p$-spin Gibbs measure as defined in \cref{eq:p-spin-model} with $h=0$. Assume the model is in the high-temperature regime, i.e., $\sum_{p=2}^{P} \beta_p \sqrt{p^3 \ln(p)} \leq C$ for some $C=\Theta(1)$. Then, there exists an $\varepsilon_n > 0$ such that $\varepsilon_n \to 0$ as $n \to \infty$, and a parallel sampling algorithm (\cref{alg:picard-asl}) which satisfies the following for any $\varepsilon > \varepsilon_n$:
\begin{itemize}
    \item The algorithm runs in parallel time $O(\polylog(\nicefrac{n}{\varepsilon}))$ with $O(\poly(\nicefrac{n}{\varepsilon}))$ work.
    
    \item With probability at least $1-o(1)$ over the randomness of the disorder $G$, the output $\hat{x}$ of the algorithm can be coupled to some $x^{\star} \sim \mu$ such that $n^{-1} \E{\|\hat{x}-x^{\star}\|^2} \leq \varepsilon^2$.
\end{itemize}
\end{theorem}
For the special case of the SK model, our low-accuracy algorithm satisfies the above guarantee for $\beta_2 \leq 0.375$. To the best of our knowledge, this constitutes the only parallel sampler for the SK model running in $O(\polylog(\nicefrac{n}{\varepsilon}))$ time with $O(\poly(\nicefrac{n}{\varepsilon}))$ work in this regime of inverse temperatures.
\subsection{Related Work}
\paragraph{Sampling from mean-field spin glasses.} For the Sherrington--Kirkpatrick model ($P=2$), classical mixing criteria such as Dobrushin's condition \cite{dobrushin1968description} imply $\Otilde(n^2)$ mixing time of Glauber dynamics only up to the vanishing window $\beta_2 = O(n^{-\nicefrac{1}{2}})$. This falls short of the dimension-free $\beta_2$ regime predicted in the statistical physics literature and leads to algorithmically unsatisfactory guarantees. The first major breakthrough was the work of \cite{eldan2022spectral} which used stochastic localization to prove that Glauber dynamics mixes in $\Otilde(n^2)$ time up to $\beta_2 < \nicefrac{1}{4}$. The mixing time was subsequently sharpened to $\Otilde(n)$ by \cite{anari2022entropic} and \cite{chen2022localization}, and the range of $\beta_2$ was improved to $\beta_2 < 0.295$ by \cite{anari2024trickle}. For the Ising $p$-spin model, \cite{adhikari2022spectral} proved an $\Otilde(n^2)$ mixing time for Glauber dynamics in the high-temperature regime $\sum_{p=2}^{P} \beta_p \sqrt{p^3 \ln(p)} = O(1)$. The mixing time was later improved to $\Otilde(n)$ by \cite{anari2024universality}. 

In a different direction, \cite{el2022sampling-sk} (see also \cite{celentano2024sudakov}) introduced Algorithmic Stochastic Localization (ASL) for the SK model with zero external field. For any $\varepsilon > \varepsilon_n$ for a fixed $\varepsilon_n = o_n(1)$, their algorithm runs in time $O(\poly(n)\exp(\poly(\nicefrac{1}{\varepsilon})))$ and obtains a low-accuracy guarantee: sampling up to $\varepsilon$-error in the normalized $2$-Wasserstein distance. Their algorithm works up to $\beta_2 < 1$, which is conjectured to be the sharp threshold up to which the SK model admits a polynomial-time approximate sampler. Their framework was later extended to the Ising $p$-spin model in \cite{el2025sampling-p-spin}, obtaining a similar low-accuracy guarantee but covering a broader range of temperatures than that of \cite{adhikari2022spectral} and \cite{anari2024universality}. 

There also exists a parallel line of work on the spherical $p$-spin model, where the Gibbs measure $\mu$ is supported on the sphere $\{ x \in \R^n \ | \ \|x\|^2 = n \}$, a setting that is often more analytically tractable than the hypercube. In the high-temperature regime, \cite{gheissari2019spectral} proved that a continuous-time diffusion process known as Langevin dynamics samples from the spherical $p$-spin model in total variation in $\Otilde(1)$ time. Beyond the high-temperature regime considered in \cite{gheissari2019spectral}, the recent breakthrough of \cite{huang2024sampling} designed an improved version of ASL for the spherical $p$-spin model with zero external field that samples up to $o(1)$ error in total variation in $O(\poly(n))$ time beyond the high-temperature regime considered by \cite{gheissari2019spectral}. Building upon their techniques, \cite{huang2025weak} proved that continuous-time annealed Langevin dynamics samples up to total variation error $\exp(-\Omega(n^{-\nicefrac{1}{5}}))$ in $O(\poly(n))$ time beyond the high-temperature regime. 

\paragraph{Parallel sampling.} Parallel sampling has been a longstanding topic of interest in theoretical computer science, and more recently, in machine learning, where diffusion-based generative models are often deployed on massively parallel hardware \cite{lai2025principles, shih2023parallel, hu2025diffusion}. In theoretical computer science, the problem dates back to at least \cite{mulmuley1987matching}, who developed an RNC algorithm (i.e., running in $O(\polylog(n))$ parallel time on a \Class{PRAM} with $O(\poly(n))$ work) for finding a perfect matching, and asked whether one can also sample a uniform random perfect matching in parallel; this remains a major open problem to date. 

Recent works study parallel sampling assuming access to global counting oracles. One such model is the weighted counting oracle, where one has access to the log Laplace transform $L(w)=\ln \E*_{\mu}{\exp(\dotp{w}{x})}$. A weaker model is the unweighted counting oracle, which returns conditional marginals $\Prb_{X\sim\mu}\!\left[X_i=a \,\middle|\, X_S=\omega_S\right]$ for $S\subseteq[n]$, $i\notin S$, $a\in\{\pm1\}$, and $\omega_S\in\{\pm1\}^S$. Such models are of interest for several combinatorial distributions such as spanning trees, Eulerian tours, planar perfect matchings, and determinantal point processes, where such oracles can be implemented on a \Class{PRAM} in $O(\polylog(n))$ time and $O(\poly(n))$ work \cite{csanky1975fast}. Assuming access to weighted counting oracles, \cite{anari2023parallel, anari2024fast} obtained RNC samplers for distributions satisfying $\|\nabla^2 L(w)\|_{\mathrm{op}}=O(1)$. For more general distributions that do not satisfy such a regularity condition, \cite{hu2025diffusion} developed an algorithm running in $\Otilde(n^{\nicefrac{2}{3}})$ parallel time with polynomial work. The parallel runtime was later improved to $\Otilde(\sqrt{n})$ by \cite{anari2025parallel}. In the weaker unweighted counting model, \cite{anari2024parallel} developed an $O(n^{\nicefrac{2}{3}})$ time parallel sampler with $O(\poly(n))$ work. The parallel runtime was later improved to $O(\sqrt{n})$ by \cite{anari2025parallel}.

Contrary to the above, our work falls under the more challenging local oracle setting, where one has to design a parallel sampler for a Gibbs measure given access to its Hamiltonian. This constitutes a far weaker form of access compared to the global counting oracles described above. In fact, the problem of computing a counting oracle of a Gibbs measure given its Hamiltonian is NP-hard in the worst case \cite{galanis2020complexity}, and highly nontrivial even for specialized Hamiltonians like the SK model. In this setting, \cite{feng2021distributed, liu2022simple} develop RNC samplers for Ising models satisfying a Dobrushin-type condition. For mean-field spin glass models, this only covers a vanishing window of temperatures, e.g., $\beta_2 = O(n^{-\nicefrac{1}{2}})$ for the SK model. More recently, \cite{chen2025efficient} show that the Restricted Gaussian Dynamics algorithm of \cite{chen2022localization} is an RNC sampler for the SK model. For general Ising $p$-spin models, much less is known. To our knowledge, the only parallel sampler in this setting prior to our work is \cite{lee2023parallelising}, which runs in $\Otilde(\sqrt{n})$ parallel time with polynomial work.

\subsection{Overview of Our Techniques}

\begin{Algorithm}
\caption{$s$-Glauber Dynamics with Parallel Rejection Sampling}
\label{alg:parallel-p-spin-sampler}

Sample $X^{(0)}$ from $\mu_0 \propto e^{\langle h, x\rangle}$ \;

\For{$t=0,\ldots,T-1$}{
    Sample $S \sim \textrm{Uniform} \binom{[n]}{s}$ \;
    
    $X_S^{(t+1)} \gets \textsc{ParallelRejectionSampler}(S, X_{\overline{S}}^{(t)}, \nicefrac{\varepsilon}{10T})\qquad$ \emph{(see \cref{alg:parallel-ising-rejection-sampler})} \;

    $X_{\overline{S}}^{(t+1)} \gets X_{\overline{S}}^{(t)}$ \;
}
\Return{$X^{(T)}$}
\end{Algorithm}

\begin{Algorithm}
\caption{Approximate Parallel Rejection Sampling with an Ising Proposal \textsc{(ParallelRejectionSampler)}}
\label{alg:parallel-ising-rejection-sampler}
\DontPrintSemicolon

\KwIn{$S$, pinning $\tau$ on $\overline{S}$, $\varepsilon_{\mathrm{step}}$}

Set $y = X_S$ and decompose the conditional Hamiltonian $H_{S,\tau}(y) = H(y, \tau)$ on $S$ as follows
\begin{align*}
H(y)
&=
\,\langle b_{S,\tau}, y\rangle
+
\frac{1}{2}\, y^\top A_{S,\tau} y
+
R_{S,\tau}(y).
\end{align*}

Set the rejection parameter $\bar c = \Theta(1)$, $L = \Theta(\bar{c} \ln(\nicefrac{1}{\varepsilon_{\mathrm{step}}}))$, $\varepsilon_{\mathrm{RGD}} = \Theta(\nicefrac{\varepsilon_{\mathrm{step}}}{L})$, $\eta = \Theta(1)$, and $\delta = \Theta(\varepsilon_{\mathrm{step}})$

$\hat{h} = \textsc{CenteringExternalField}(A_{S,\tau}, b_{S,\tau}, R_{S,\tau}, \eta, \delta) \qquad$ \emph{(see \cref{alg:centering-external-field})}

\For(\textbf{In Parallel}){$\ell = 1$ \KwTo $L$}{
    $U_\ell, V_\ell \stackrel{\mathrm{iid}}{\sim}
    \textsc{RGD}\!\left(A_{S,\tau},\, b_{S,\tau}+\hat h,\, \varepsilon_{\mathrm{RGD}}\right)$\;

    $W_\ell \sim \Uniform[0,1]$\;
    
    Accept $U_\ell$ if $ W_\ell \le \min\!\left\{1, \bar c^{-1} \exp(\psi(U_\ell) - \psi(V_\ell))\right\}$ where $\psi(y) = R_{S,\tau}(y) - \dotp{\hat{h}}{y}$ \;

}

\KwRet{the first accepted sample; declare \textsc{FAILURE} otherwise}
\end{Algorithm}

\paragraph{High Accuracy Sampler}
Based on the prior work of \cite{lee2023parallelising}, our high-accuracy sampler for the $p$-spin model, stated in \cref{alg:parallel-p-spin-sampler}, is a computationally efficient implementation of a Markov chain called $s$-Glauber dynamics (also known as $s$-block dynamics), whose kernel $P_s(X^{(t)}, \cdot)$ is defined as follows:
\begin{enumerate}
    \item Sample $S$ uniformly at random from all subsets of $[n]$ of size $s$.
    \item Sample $X^{(t+1)}$ from the posterior $\mu\parens*{ X^{(t+1)} \given  X^{(t+1)}_S=X^{(t)}_S}$
\end{enumerate}
The above Markov chain is reversible with respect to $\mu$ and, for $s=1$, it corresponds to Glauber dynamics, a canonical sampling algorithm for discrete spaces. When $\mu$ is the $p$-spin Gibbs measure at high temperature, $s$-Glauber dynamics is known to exhibit a mixing time of $\Otilde(\nicefrac{n}{s})$ \cite{anari2024universality, lee2023parallelising}. Although this mixing time improves upon increasing $s$, the computational complexity of implementing each step of $s$-Glauber dynamics worsens significantly as $s$ grows. Indeed, a na\"ive implementation of the posterior sampling step (i.e., Step 2 above) takes $O(2^s)$ time, which is computationally prohibitive, and completely negates the potential benefits of any mixing time improvements for $s > 1$.  

To circumvent this, our \cref{alg:parallel-p-spin-sampler} replaces the posterior sampling step in $s$-Glauber dynamics with an approximate rejection sampler (\cref{alg:parallel-ising-rejection-sampler}), whose proposal distribution is a carefully chosen Ising model. The idea of using rejection sampling in the posterior sampling step also appears in the work of \cite{lee2023parallelising}, which approximates the posterior sampling step with a product proposal. While this enables computational efficiency (as product measures can be sampled in $O(1)$ parallel time), their rejection sampling procedure fails to correctly approximate the posterior when $s=\omega(\sqrt{n})$, and obtains a parallel runtime of $\Otilde(n^{\nicefrac{1}{2}})$. In contrast, the choice of our Ising proposal, which serves as the key technical innovation underlying our parallel speedup, can be made to approximate the posterior to arbitrary accuracy as long as $s=O(n^{\nicefrac{2}{3}})$, leading to an improved parallel runtime of $\Otilde(n^{\nicefrac{1}{3}})$.

To motivate the construction of our Ising proposal, we observe that the conditional posterior distribution $\mu_{S,\tau}(y) \coloneq \mu(x_S=y \ | x_{\bar{S}} = \tau) \propto \exp(H_{S,\tau}(y))$ is also a $p$-spin Gibbs measure whose Hamiltonian, $H_{S,\tau}(y) = H(x)$ for $x_{S} = y$ and $x_{\bar{S}}=\tau$, is obtained by restricting $H$ to the coordinates in $S$ and pinning the coordinates in $\bar{S}$ according to $\tau$. $H_{S,\tau}$ is then a polynomial of the underlying disorder $G$, and admits the following decomposition:
\begin{align}
    H_{S,\tau}(y) = \dotp{b_{S,\tau}}{y} + \frac{1}{2}\dotp{y}{A_{S,\tau} y} + R_{S,\tau}(y)
\end{align}
where $b_{S,\tau}$, $A_{S,\tau}$, and $R_{S,\tau}$ depend on $G$. Motivated by this decomposition, our rejection sampler uses a proposal distribution $\pi_{S,\tau}$ such that $\ln \pi_{S,\tau}$ serves as a quadratic approximation to $H_{S,\tau}$, i.e., $\pi_{S,\tau}(x) \propto \exp\parens*{\tfrac{1}{2}\dotp{y}{A_{S,\tau} y} + \dotp{b_{S,\tau} + \hat{h}}{y}}$\footnote{As explained in \cref{sec:parallel-ising-rejection-sampler-analysis}, the nonzero external field $\hat{h}$ \emph{centers} our rejection sampling proposal, leading to improved accuracy.}. The performance and accuracy of our algorithm are governed by two key problems: 1.\ How easily can we sample from the proposal? 2.\ How well does our rejection sampler approximate the posterior? Below, we discuss how we address each of these questions in our proof of \cref{thm-informal:p-spin-sampler}. 

Analyzing the parallel complexity of our algorithm is a delicate problem due to the choice of our proposal distribution. Indeed, sampling from an Ising model is in itself a nontrivial problem, which can even be \Class{NP}-hard in the worst case \cite{sly2012computational, galanis2024sampling}. To resolve this, we use the concentration properties of the $p$-spin Hamiltonian to prove that sampling from $\pi_{S,\tau}$ is computationally tractable. In particular, the Hamiltonian of $\pi_{S,\tau}$ satisfies a spectral condition which ensures that a Markov chain known as Restricted Gaussian Dynamics (or RGD, see \cref{alg:rgd-ising-parallel}) approximately samples from $\pi_{S,\tau}$ in total variation, with $\Otilde(1)$ mixing time \cite{chen2022localization}. Moreover, each step of RGD can be implemented in $\Otilde(1)$ parallel time with $O(\poly(n))$ work \cite{chen2025efficient}. Thus, to ensure parallel efficiency, our rejection sampler uses approximate samples from the proposal distribution, drawn via $\Otilde(1)$ parallel calls to RGD, thereby implying an $\Otilde(1)$ parallel runtime for \cref{alg:parallel-ising-rejection-sampler}, which in turn leads to an $\Otilde(\nicefrac{n}{s})$ parallel runtime for \cref{alg:parallel-p-spin-sampler}. 

\begin{Algorithm}[h]
\caption{Restricted Gaussian Dynamics (RGD) \cite{chen2022localization, chen2025efficient}}
\label{alg:rgd-ising-parallel}
\DontPrintSemicolon

\KwIn{symmetric interaction matrix $J \in \bR^{m \times m}$, external field $h \in \bR^m$, tolerance $\varepsilon$}

\KwOut{Sample $X$ satisfying $\TV{\mathsf{Law}(X),\nu_{J,h}} \leq \epsilon$ where $\nu_{J,h}(x) \propto \exp(\tfrac{1}{2}\dotp{x}{Jx} + \dotp{h}{x})$}

Initialize $X^{(0)} \in \{\pm 1\}^m$ arbitrarily\;

Set $T = \Theta(\log(\nicefrac{m}{\epsilon}))$

\For{$t=1,\ldots,T$}{
    Sample $y \sim \mathcal{N}(X^{(t)}, J^{-1})$ \;

    Sample $X^{(t)}$ from the product measure $
        p(x) \propto \exp\!\left(\langle x, h + Jy\rangle\right)$
}

\KwRet{$X^{(T)}$}
\end{Algorithm}

Finally, the optimal choice of $s$ is determined by the approximation error of our rejection sampler. There are two main sources of this error, namely, the approximate sampling error of RGD (which can be controlled to arbitrary precision without sacrificing algorithmic efficiency), and the error due to mismatch between the proposal $\pi_{S,\tau}$ and the conditional posterior $\mu_{S,\tau}$, which is governed by the log-density $\psi(y) = \ln \tfrac{\d \mu_{S,\tau}}{\d \pi_{S,\tau}}(y)$. We show that this error can be controlled by proving exponential tail bounds on $\psi$ under $\pi_{S,\tau}$, which we establish via the concentration properties of the $p$-spin Hamiltonian and the isoperimetric properties of the proposal, and conclude that \cref{alg:parallel-ising-rejection-sampler} uniformly approximates the posterior sampling step of $s$-Glauber dynamics for $s \leq O(n^{\nicefrac{2}{3}})$, leading to an overall runtime of $\Otilde(n^{\nicefrac{1}{3}})$ for \cref{alg:parallel-p-spin-sampler}. 

\paragraph{Low Accuracy Sampler} Adapting the framework of \cite{el2022sampling-sk, el2025sampling-p-spin}, our low-accuracy parallel sampler, presented in \cref{alg:picard-asl}, is based on Eldan's Stochastic Localization (SL) \cite{eldan2013thin}, which we briefly introduce. For any probability measure $\mu$ supported on $\{\pm 1\}^n$, the associated SL process $(y_t)_{t \geq 0}$ is defined by the following stochastic differential equation on $\R^n$.
\begin{align}
\label{eq:sl-process}
    \d y_t = m(y_t) \d t + \d B_t, \ \ y_0=0, \qquad \ m(y)= \frac{\E*_{x \sim \mu}{x \cdot \exp(\dotp{y}{x})}}{\E*_{x \sim \mu}{\exp(\dotp{y}{x})}}
\end{align}
Here, $(B_t)_{t \geq 0}$ denotes the standard Brownian motion on $\R^n$, and $m(y)$ corresponds to the mean of the exponential tilt of the measure $\mu$, which we call the tilted mean. A key property of SL is that $t^{-1} y_t$ is marginally distributed as $x^{\star} + \zeta_t$ where $x^{\star} \sim \mu$ and $\zeta_t \sim \cN(0,t^{-1} I)$. Intuitively, $t^{-1}y_t$ is approximately distributed as $\mu$ for large enough $t$, and thus, a suitable discretization of \cref{eq:sl-process}, such as the Euler discretization below, can be used to approximately sample from $\mu$.
\begin{align}
\label{eq:sl-euler}
\widetilde{y}_{(k+1)h} = \widetilde{y}_{kh} + h m(\widetilde{y}_{kh}) + (B_{(k+1)h} - B_{kh}), \ 0 \leq k < K, \qquad y_0 = 0
\end{align}
Being the Euler discretization of a continuous Markov process, the discrete dynamical system specified by \cref{eq:sl-euler} naturally evolves in a sequential fashion, and thus, naively simulating $K$ steps requires $\Theta(K)$ calls to an oracle for the tilted mean $m$, even on a \Class{PRAM}. To circumvent this, we employ a parallelization technique based on the concept of Picard iteration, a classical technique for proving the existence and uniqueness of ODE solutions \cite{perko2013differential}, which was later used algorithmically by \cite{anari2024fast} to develop parallel algorithms for sampling from continuous log-concave densities in $\R^n$. The key idea behind Picard iteration is to view the trajectory defined by \cref{eq:sl-euler} as the solution to the following fixed point problem defined by an operator $\mathcal{T} : \R^{n \times K} \to \R^{n \times K}$ over discrete trajectories:
\begin{align}
\mathcal{T}(z) \coloneq z^{\prime}; \qquad z^{\prime}_{kh} = h \sum_{i=0}^{k-1} m(z_{ih}) \ + B_{kh} \ \forall \ k \in [K], \qquad z^{\prime}_0 = 0
\end{align}
By unrolling the recurrence, one can observe that the trajectory defined by \cref{eq:sl-euler} is a fixed point of the operator $\mathcal{T}$, i.e., $\mathcal{T}(\widetilde{y}) = \widetilde{y}$. Consequently, the discretized SL process can be approximated via the following fixed point iteration over discrete trajectories, which we call Picard iteration.
\begin{align}
\label{eq:picard-sl}
\widetilde{y}^{(r)}_{kh} = h \sum_{i=0}^{k-1}m(\widetilde{y}^{(r-1)}_{ih}) + B_{kh}, \ \forall \ k \in [K], \qquad \widetilde{y}^{(r)}_{0} = 0
\end{align}

\begin{Algorithm}[th]
\caption{Picard Algorithmic Stochastic Localization}
\label{alg:picard-asl}
\DontPrintSemicolon
\SetKwInOut{Input}{Input}
\SetKwInOut{Output}{Output}

\Input{Parallel depth $R$, steps $K$, step size $h$, parameters $(\eta, K_{\mathrm{AMP}}, K_{\mathrm{NGD}})$}

$B_0 \gets 0$\;

\For(\textbf{In Parallel}){$k \gets 0$ \KwTo $K-1$}{
    $B_{(k+1)h} \gets B_{kh} + \varepsilon_k$, where $\varepsilon_k \sim \mathcal{N}(0,hI)$\;
    
    $\hat{y}_{kh}^{(0)} \gets 0$\;
}

\For{$r \gets 1$ \KwTo $R$}{
    $\hat{y}_0^{(r)} \gets 0$\;

    \For(\textbf{In Parallel}){$k \gets 0$ \KwTo $K-1$}{
        $\hat{m}\bigl(\hat{y}_{kh}^{(r-1)}\bigr) \gets \textsc{TAP-AMP}\bigl(\hat{y}_{kh}^{(r-1)}, \eta, q_{\star}(kh),  K_{\mathrm{AMP}}, K_{\mathrm{NGD}} \bigr)$\; \qquad \emph{(see \cref{alg:tap-amp}) and \cref{eq:asymptotic-overlap}}
    }

    \For(\textbf{In Parallel}){$k \gets 1$ \KwTo $K$}{
        $\hat{y}_{kh}^{(r)} \gets h \sum_{j=0}^{k-1} \hat{m}\bigl(\hat{y}_{jh}^{(r-1)}\bigr) + B_{kh}$\;
    }
}

\Output{$\hat{x} \gets \operatorname{sign}\bigl(\hat{y}_{Kh}^{(R)}\bigr)$}
\end{Algorithm}

The key algorithmic benefit of this reformulation lies in its inherent parallelizability, albeit at the cost of introducing some mild approximation error. Indeed, \cref{eq:picard-sl} turns the problem of sequentially evaluating the recurrence in \cref{eq:sl-euler} into a highly parallelizable fixed point computation over the whole trajectory: given access to a parallel oracle for the tilted mean $m$, each Picard iteration can be implemented efficiently in parallel across all time points. In particular, at depth $r$, $(\widetilde{y}^{(r-1)}_{kh})_{k \in [K]}$ can be computed simultaneously in one parallel oracle call, after which the cumulative drift terms in \cref{eq:picard-sl} can be computed via parallel prefix sums in $O(\log(k))$ time and $O(k)$ work. However, to guarantee the effectiveness of this approach as a parallel sampler, one needs to quantify how well the Picard iterate $\widetilde{y}^{(R)}$ approximates its limiting trajectory $\widetilde{y}$. To this end, one can show (see \cref{thm:picard-convergence}) that Lipschitzness of the tilted mean (which can be guaranteed in the high-temperature regime) suffices to ensure that the Picard iterates converge to the discretized SL process exponentially fast. In particular, $R=\widetilde{\Theta}(Kh)$ iterations suffice to approximate the discrete SL trajectory in \cref{eq:sl-euler} to arbitrary accuracy. 

The key challenge in algorithmically implementing the Picard iteration in \cref{eq:picard-sl} arises from the approximate computation of the tilted mean $m$, which is highly nontrivial and is known to be \Class{NP}-hard even for worst-case Hamiltonians \cite{galanis2020complexity}. To circumvent this, we make use of the TAP-AMP algorithm of \cite{el2023algorithmic-spin} (\cref{alg:tap-amp}), which constructs an estimator $\hat{m}$ of the tilted mean for points along the continuous SL trajectory. This results in the following inexact fixed point iteration, which we call Picard Algorithmic Stochastic Localization (\cref{alg:picard-asl}).
\begin{align}
\label{eq:picard-asl}
\hat{y}^{(r)}_{kh} = h \sum_{i=0}^{k-1}\hat{m}(\hat{y}^{(r-1)}_{ih}) + B_{kh}, \ \forall \ k \in [K], \ \ \hat{y}^{(r)}_{0} = 0
\end{align}
While standard perturbation-based analyses of inexact fixed point iterations rely on uniformly controlling the deviation between the ideal and the inexact iterates, such a strategy is not applicable to \cref{eq:picard-asl} due to the absence of uniform guarantees on the estimation error of $\hat{m}$. In fact, the current best known analysis of the tilted mean estimator can certify only a very weak error bound of the form $\| \hat{m}(y_t) - m(y_t) \|^2 \leq n \varepsilon^2_n$ for points along the continuous SL trajectory. Here, $\varepsilon_n$ is a fixed error threshold which converges to $0$ as $n \to \infty$. This subtle distinction constitutes a major obstacle in our analysis since the Picard trajectories at low depth (i.e., low values of $r$) are expected to deviate significantly from the continuous SL trajectory. Consequently, the currently available guarantees are insufficient for controlling $\|\hat{m}(\hat{y}^{(r)}_{kh})-m(\hat{y}^{(r)}_{kh})\|$. 

To circumvent this obstacle, we depart from the standard perturbation-based analysis of inexact fixed point iteration, and instead analyze \cref{eq:picard-asl} directly. In particular, we prove that in the high-temperature regime, the tilted mean estimator $\hat{m}(y)$ itself is uniformly $O(1)$-Lipschitz. Hence, the inexact Picard iteration in \cref{eq:picard-asl} converges exponentially fast to the discrete process $\hat{y}_{(k+1)h} = \hat{y}_{kh} + h \hat{m}(y_{kh}) + (B_{(k+1)h} - B_{kh})$ in $R = \Otilde(Kh)$ iterations. We then prove a sign rounding stability guarantee for the limiting discrete trajectory $\hat{y}$, and show that for any $\varepsilon \geq \varepsilon_n$, the distribution of $\operatorname{sign}(\nicefrac{\hat{y}^{(R)}_{Kh}}{Kh})$ approximates the $p$-spin measure $\mu$ up to $\varepsilon$ error in normalized 2-Wasserstein distance whenever $Kh=O(\log(\nicefrac{n}{\varepsilon}))$ and $h=O(\poly(\nicefrac{1}{\varepsilon}))$.  Since the $K$ iterations are executed in parallel in each round, \cref{alg:picard-asl} exhibits a parallel runtime of $O(\polylog(\nicefrac{n}{\varepsilon}))$ with $O(\poly(\nicefrac{n}{\varepsilon}))$ work. Compared to the $O(\poly(n) \exp(\poly(\nicefrac{1}{\varepsilon})))$ runtime of Algorithmic Stochastic Localization, this constitutes a doubly exponential improvement in runtime and an exponential improvement in work in terms of $\varepsilon$.

\begin{Algorithm}[t]
\caption{AMP Algorithm for Tilted Mean Estimation (\textsc{TAP-AMP}) \quad \cite[Algorithm 1]{el2025sampling-p-spin}}
\label{alg:tap-amp}
\DontPrintSemicolon
\SetKwInOut{Input}{Input}
\SetKwInOut{Output}{Output}

\Input{Iterations $K_{\mathrm{AMP}}, K_{\mathrm{NGD}}$, step size $\eta$, overlap parameter $q$}

$\hat m^{-1} \gets 0$, $z^0 \gets 0$\;

\For{$k \gets 0$ \KwTo $K_{\mathrm{AMP}}-1$}{
    $\hat m^k \gets \tanh(z^k)$\;
    $\hat q_k \gets \tfrac{1}{n}\sum_{i=1}^n \tanh^2(z_{i,k})$\;
    $b_k \gets (1-\hat q_k)\,\xi''(\hat q_k)$\; 
    
    $z^{k+1} \gets \beta \nabla H(\hat m^k) + y - b_k \hat m^{k-1}$\;
}

$u^0 \gets z^{K_{\mathrm{AMP}}}$, $\hat m^{+,0} \gets \hat m^{K_{\mathrm{AMP}}}$\;

\For{$k \gets 0$ \KwTo $K_{\mathrm{NGD}}-1$}{
    $u^{k+1} \gets u^k - \eta \nabla \hat{\mathscr{F}}_{\mathrm{TAP}}(\hat m^{+,k} \ ;  y , q)$\;  \qquad \emph{(see \cref{eq:approx-tap})}
    
    $\hat m^{+,k+1} \gets \tanh(u^{k+1})$\;
}

\Output{$\hat m^{+,K_{\mathrm{NGD}}}$}
\end{Algorithm}

\section{Preliminaries}
\subsection{Notation}
 For differentiable functions $g : \R^d \to \R$, we use $\nabla g$ and $\nabla^2g$ to denote their Euclidean gradient and Hessian, respectively. For functions $f : \{ \pm 1 \}^m \to \R$, we define the discrete partial derivative, and the associated discrete gradient and discrete Hessian as follows:
 \begin{align}
 \label{eq:discrete-partial}
 \partial_i f(x) = \frac{1}{2} \cdot  [f(x_{\neg i} \ ; \ x_i = +1) - f(x_{\neg i} \ ; \ x_i = -1)], \qquad
 (\nabla f(x))_{i} = \partial_i f(x), \qquad
 (\nabla^2 f(x))_{ij} = \partial_i \partial_j f(x)
 \end{align}
 For any $x \in \R^m$ and multi-set $J \subseteq [n]$, we define $x^J = \prod_{k \in J}x_k$. For any two distributions $P$ and $Q$ on a state space $\Omega$, we denote their total variation distance as
 \begin{align}
     d_{\mathsf{TV}}(P,Q) = \frac{1}{2} \sum_{x \in \Omega}|P(x)-Q(x)|
 \end{align}
 When clear from the context, for random variables $x, y$ we use $d_{\mathsf{TV}}(x,y)$ to denote $d_{\mathsf{TV}}(\mathsf{Law}(x),\mathsf{Law}(y))$.

\begin{definition}[Normalized 2-Wasserstein distance]
\label{def:normalized-wasserstein-l2}
 For probability measures $\mu, \nu$ on $\R^d$ with a finite second moment, which shall always be equipped with the Euclidean metric unless stated otherwise, we define the $2$-Wasserstein distance $W_2(\mu, \nu)$ and its normalized variant $W_{2,d}(\mu,\nu)$ as follows:
 \begin{align}
     W_2(\mu,\nu) = \left(\inf \limits_{\mathcal{C} \in \Pi(\mu,\nu)} \E*_{(X,Y) \sim \mathcal{C}}{\|X-Y\|^2}\right)^{\nicefrac{1}{2}}, \qquad
     {W}_{2, d}(\mu,\nu) = \frac{1}{\sqrt{d}} \cdot W_2(\mu,\nu)
 \end{align}
 where $\Pi(\mu,\nu)$ denotes the set of all couplings of the measures $\mu$ and $\nu$. 
\end{definition}

Throughout, we use $H$ and $\mu$ to denote the Ising $p$-spin Hamiltonian and its associated Gibbs measure as defined in \cref{eq:p-spin-model}. We use $\xi(t) = \sum_{p=2}^{P} \beta_p^2 t^p$ to denote its mixture function. We also define the coefficients $\mathfrak{C}(\beta)$ and $\mathfrak{D}(\beta)$ as follows:
\begin{align}
\label{eq:temp-coeffs}
\mathfrak{C}(\beta) \coloneqq \sum_{p=2}^{P} \beta_p \sqrt{p^3 \ln(p)} \ ; \qquad \mathfrak{D}(\beta) \coloneqq \sum_{p=2}^{P} \beta_p \sqrt{2^{p} p^3 \ln(p)} 
\end{align}
\subsection{Functional Inequalities}
\begin{definition}[Poincar\'e and Log-Sobolev Inequalities]
\label{def:lsi-pi}
A measure $\pi$ on $\{\pm 1\}^m$ is said to satisfy a Poincar\'e inequality with constant $\lambda_{\mathrm{PI}}$ if for any $f : \{\pm 1\}^m \to \mathbb{R}$,
\begin{align}
\mathsf{Var}_{\pi}[f]
\leq
\lambda_{\mathrm{PI}} \, \mathbb{E}_{\pi}\!\left[\|\nabla f\|^2\right].
\end{align}

$\pi$ is said to satisfy a Log-Sobolev inequality with constant $\lambda_{\mathrm{LSI}}$ if for any $f : \{\pm 1\}^m \to \mathbb{R}$,
\begin{align}
\mathsf{Ent}_{\pi}[f^2]
\leq
\lambda_{\mathrm{LSI}} \, \mathbb{E}_{\pi}\!\left[\|\nabla f\|^2\right].
\end{align}

Furthermore, LSI implies PI with $\lambda_{\mathrm{PI}} \leq \lambda_{\mathrm{LSI}}.$
\end{definition}

\begin{definition}[Approximate Tensorization of Entropy]
\label{def:ate}
A measure $\pi$ on $\{\pm 1\}^m$ is said to satisfy $C$-approximate tensorization of entropy (or $C$-ATE) if for any $f : \{\pm 1\}^m \to \mathbb{R}$,
\begin{align}
\mathsf{Ent}_{\pi}[f^2]
\leq
C \sum_{i=1}^m \bE_{\pi}\!\left[\mathsf{Ent}_{\mu(x_i \mid x_{-i})}[f^2]\right].
\end{align}
\end{definition}
It is easy to show that any probability measure on $\{\pm 1\}$ satisfies an LSI with a constant of $\nicefrac{1}{2}$ \cite{bauerschmidt2019very}. Hence, any measure $\pi$ satisfying $C$-approximate tensorization of entropy also satisfies an LSI with $\lambda_{\mathrm{LSI}} = \nicefrac{C}{2}$.

The following Lipschitz concentration guarantee for measures satisfying an LSI is standard, and follows directly from the Herbst argument \cite[Ch. 3]{van2014probability}.

\begin{theorem}[Lipschitz Concentration under LSI]
\label{thm:herbst-conc}
Let $\pi$ be a measure on $\{\pm 1\}^m$ satisfying an LSI. Then, for any $f : \{\pm 1\}^m \to \mathbb{R}$, which satisfies $\max \limits_{x \in \{\pm 1\}^m} \|\nabla f(x)\| \leq G$, the following holds:
\begin{align}
\pi\!\left(\left|f - \bE_{\pi}[f]\right| \ge t\right)
\le
2 \exp\left(-\frac{t^2}{2 \lambda_{\mathrm{LSI}}G^2}\right)
\end{align}
\end{theorem}

The following Bernstein-type concentration bound for measures satisfying an ATE follows from Propositions 2.16 and 2.18 of \cite{sambale2019modified}.

\begin{theorem}[Bernstein-Type Concentration for ATE Measures]
\label{thm:sambale-sinulis}
Let $\pi$ be a measure on $\{\pm 1\}^m$ satisfying $C$-approximate tensorization of entropy. Then, for any $f : \{\pm 1\}^m \to \mathbb{R}$,
\begin{align}
\pi\!\left(\left|f - \bE_{\pi}[f]\right| \ge t\right)
\le
2 \exp\!\left(
-\frac{c_1}{C}
\min\left\{
\frac{t^2}{\bE_{\pi}\!\left[\|\nabla f\|^2\right]},
\frac{t}{\max\limits_{x \in \{\pm 1\}^m} \|\nabla^2 f(x)\|_{F}}
\right\}
\right).
\end{align}
\end{theorem}

The following theorem establishing ATE for high-temperature Ising models is implicit in \cite{anari2024universality} and also appears as Theorem 4.1 in \cite{lee2023parallelising}.

\begin{theorem}[ATE for High-Temperature Ising Models]
\label{thm:ising-ate}
Let $\nu_{J,h}(x) \propto \exp\!\left(\frac{1}{2}\langle x,Jx\rangle + \langle h,x\rangle\right)$ be an Ising model on $\{\pm 1\}^m$ with $\|J\|_{\op} < 1$. Then, $\nu_{J,h}$ satisfies $C$-ATE with $C = (1-\|J\|_{\op})^{-1}$.
Consequently, $\nu_{J,h}$ satisfies an LSI (and hence, PI), with $\lambda_{\mathrm{LSI}} = \tfrac{1}{2} \cdot (1-\|J\|_{\op})^{-1}$.
\end{theorem}

\subsection{Stochastic Localization}
Let $\mu$ denote any arbitrary measure on $\R^n$ with a finite second moment. Stochastic localization, introduced by \cite{eldan2013thin}, denotes the following diffusion process:
\begin{align}
\label{eq:sl-full}
\d y_t &= m(y_t,t) \d t + \d B_t, \qquad y_0=0  \nonumber \\
m(y,t) &= \E*_{x \sim \mu, \zeta \sim \cN(0,1)}{x \given tx + \sqrt{t}\zeta = y}
\end{align}
Stochastic localization exhibits several interesting properties that have found innumerable applications in areas such as convex geometry and the analysis of Markov chains \cite{lee2024eldan, eldan2022spectral, chen2022localization}. Among these, the following property, which appears in \cite{el2022information}, will be of particular importance to us.

\begin{theorem}[Equivalent Characterization of SL \cite{el2022information}]
\label{thm:sl-marginal}
Let $(y_t)_{t \geq 0}$ denote the stochastic localization process associated with the measure $\mu$, as defined in \cref{eq:sl-full}. Then, there exists an $x^{\star} \sim \mu$ and a standard Brownian motion $(W_t)_{t \geq 0}$ such that for any $t \geq 0$, $y_t = t x^{\star} + W_t$. Consequently, the marginal law of $t^{-1}y_t$ is equal to that of $x^{\star} + \zeta_t$ where $\zeta_t \sim \cN(0,t^{-1} I)$.
\end{theorem}
The following lemmas, which follow from \cref{thm:sl-marginal}, are vital for analyzing the rounding scheme used in \cref{alg:picard-asl}.
\begin{lemma}[Rounding SL for the hypercube]
\label{lem:sl-sign-rounding}
Let $\mu$ be a measure supported on $\{\pm 1\}^n$, and let $(y_t)_{t \ge 0}$ denote the associated SL process. For any $T > 0$, define $\hat x_T = \operatorname{sign}(T^{-1} y_T)$. Then,  $W_{2, n}^2(\mathsf{Law}(\hat x_T), \mu) \le 8 \exp(-\nicefrac{T}{2}).$
\end{lemma}

\begin{proof}
By \cref{thm:sl-marginal}, there exist $x^* \sim \mu$ and $z \sim \mathcal{N}(0,T^{-1}I)$ such that $T^{-1} y_T = x^* + z.$ Then,
\begin{align}
W_{2, n}^2\bigl(\mathsf{Law}(\hat x_T), \mu\bigr)
&\le  \frac{1}{n} \mathbb{E}\bigl[\|\hat x_T - x^*\|^2\bigr] \\
&=  \frac{1}{n}\sum_{i=1}^n \mathbb{E}\Bigl[\bigl(x_i^* - \operatorname{sign}(x_i^* + z_i)\bigr)^2\Bigr] \\
&= \frac{1}{n} \sum_{i=1}^n 4 \Prb\Bigl[x_i^* \neq \operatorname{sign}(x_i^* + z_i)\Bigr] \\
&\le \frac{1}{n} \sum_{i=1}^n 4 \Prb\bigl[|z_i| \ge 1\bigr] \\
&\le 8 e^{-T/2}.
\end{align}
\end{proof}

\begin{lemma}[Stability of SL Rounding]
\label{lem:stability-sl-rounding}
Let $\mu$ be a distribution on $\{\pm 1\}^n$, and let $(y_t)_{t \ge 0}$ denote the associated SL process. For any $T > 0$, let $\hat z_T$ be a random variable satisfying $W_2^2(\hat z_T, T^{-1} y_T) \le n \Delta^2.$ Then,
\begin{align}
W_{2, n}^2\bigl(\operatorname{sign}(\hat z_T), \operatorname{sign}(T^{-1} y_T)\bigr)
\le
8 e^{-T/8} + 16 \Delta^2.
\end{align}
\end{lemma}

\begin{proof}
By \cref{thm:sl-marginal}, there exist $x^* \sim \mu$ and $\zeta \sim \mathcal{N}(0,T^{-1}I)$ such that $T^{-1} y_T = x^* + \zeta$. 
Couple $\hat z_T$ and $T^{-1} y_T$ $W_2$-optimally such that $\bE\bigl[\|\hat z_T - T^{-1} y_T\|^2\bigr] \le n \Delta^2.$ Then,
\begin{align}
W_2^2\bigl(\operatorname{sign}(\hat z_T), \operatorname{sign}(T^{-1} y_T)\bigr)
&\le
\bE\bigl[\|\operatorname{sign}(\hat z_T) - \operatorname{sign}(T^{-1} y_T)\|^2\bigr] \nonumber \\
&=
4 \sum_{i=1}^n \Prb\bigl[\operatorname{sign}(\hat z_{T,i}) \neq \operatorname{sign}(T^{-1} y_{T,i})\bigr].
\label{eq:sl-rounding-coordinate}
\end{align}

Observe that if $|T^{-1} y_{T,i}| > 1/2$ and $|\hat z_{T,i} - T^{-1} y_{T,i}| < 1/2$, then $\operatorname{sign}(\hat z_{T,i}) = \operatorname{sign}(T^{-1} y_{T,i}).$ It follows that,
\begin{align}
\Prb\bigl[\operatorname{sign}(\hat z_{T,i}) \neq \operatorname{sign}(T^{-1} y_{T,i})\bigr]
&\le
\Prb\bigl[|T^{-1} y_{T,i}| \le 1/2\bigr]
+
\Prb\bigl[|\hat z_{T,i} - T^{-1} y_{T,i}| \ge 1/2\bigr] \\
&\le
\Prb\bigl[|x_i^* + \zeta_i| \le 1/2\bigr]
+
4 \bE\bigl[(\hat z_{T,i} - T^{-1} y_{T,i})^2\bigr] \\
&\le
\Prb\bigl[|\zeta_i| \ge 1/2\bigr]
+
4 \bE\bigl[(\hat z_{T,i} - T^{-1} y_{T,i})^2\bigr] \\
&\le
2 e^{-T/8}
+
4 \bE\bigl[(\hat z_{T,i} - T^{-1} y_{T,i})^2\bigr].
\end{align}
Substituting the above into \cref{eq:sl-rounding-coordinate},
\begin{align}
W_{2, n}^2\bigl(\operatorname{sign}(\hat z_T), \operatorname{sign}(T^{-1} y_T)\bigr)
&\le
8 e^{-T/8}
+
\frac{16}{n}\bE\bigl[\|\hat z_T - T^{-1} y_T\|^2\bigr] \\
&\le
8 e^{-T/8}
+
16 \Delta^2.
\end{align}
\end{proof}
\subsection{Picard Iteration}
In its canonical form, Picard iteration is a well-known analytic technique for proving the existence and uniqueness of solutions to well-posed ordinary differential equations \cite{perko2013differential}. Recently, it has been used as a technique for parallelizing Langevin-type algorithms for continuous sampling \cite{anari2024fast} as well as denoising diffusion models in machine learning \cite{shih2023parallel}. In this section, we present a concise introduction to the discrete-time analog of Picard iterations, which is the version used throughout this work.

Consider a discrete-time dynamical system $(x_{kh})_{0 \leq k \leq K}$ on $\R^n$ defined by the following iteration:
\begin{align}
\label{eq:discrete-diff}
    x_{(k+1)h} = x_{kh} + h f(kh, x_{kh}) + g(kh), \qquad k \in \{0, \dots, K-1\}
\end{align}
where $h > 0$ denotes the step-size and $T=Kh$ denotes the time horizon, and the functions $f$ and $g$ can potentially be random. We observe that solving the difference equation above is an inherently sequential task. Concretely, given an initialization $x_0$ and access to an oracle for computing $f$ and $g$, computing the trajectory $(x_{kh})_{0 \leq k \leq K}$ requires $\Theta(K)$ oracle calls. 

The key idea in Picard iteration is to view the entire trajectory $(x_{kh})_{0 \leq k \leq K}$ as the fixed point of a trajectory-level iteration in $\R^{n \times K}$ such that, given access to a \emph{parallel oracle} for computing $f$ and $g$, each round of this trajectory-level iteration can be computed in $\Theta(1)$ oracle calls. To begin, we unroll the recurrence in \cref{eq:discrete-diff} and observe that the solution $(x_{kh})_{0 \leq k \leq K}$ satisfies the following:
\begin{align}
\label{eq:trajectory-fp}
x_{kh} = x_0 + \sum_{j=0}^{k-1} \left(hf(jh, x_{jh}) + g(jh)\right), \qquad 0 \leq k \leq K
\end{align}
The Picard iteration associated with \cref{eq:discrete-diff} is the following sequence of trajectories $(x^{(r)}_{kh})_{0 \leq k \leq K}, \ r \in \N$. 
\begin{align}
\label{eq:picard-itr}
x^{(r+1)}_{kh} = x_0 + \sum_{j=0}^{k-1} \left(h f(jh, x^{(r)}_{jh}) + g(jh)\right)
\end{align}
We make two key observations regarding the above trajectory sequence:
\begin{itemize}
    \item The solution $x_{kh}$ to \cref{eq:discrete-diff} is a fixed point of \cref{eq:picard-itr}. Indeed, setting $x^{(r)}_{kh} = x_{kh}$ for every $k$ and using \cref{eq:trajectory-fp}, we conclude that $x^{(r+1)}_{kh} = x_{kh}$. 

    \item Each round of the iteration in \cref{eq:picard-itr} can be computed using $\Theta(1)$ calls to a parallel oracle for computing $f$ and $g$. Indeed, the values of $g(jh), \ 0 \leq j < K$ can be precomputed using one parallel oracle call, and in each round, the values of $f(jh, x^{(r)}_{jh}), \ 0 \leq j < K$ can be computed via one call to a parallel oracle to $f$. 
\end{itemize}
The parallel efficiency of Picard iteration then depends on how fast the sequence of trajectories in Picard iteration converges to its fixed point trajectory. Below, we prove that Picard iteration exhibits an exponential convergence rate whenever $f(t,x)$ is uniformly $L$-Lipschitz in $x$.
\begin{lemma}[Exponential convergence of Picard iteration under Lipschitzness]
\label{thm:picard-convergence}
Let $(x_{kh})_{0 \le k \le K}$ denote the solution to \cref{eq:discrete-diff}, where $T = Kh$, and let $(x_{kh}^{(r)})_{0 \le k \le K,\ r \ge 0}$ denote the associated Picard iteration as defined in \cref{eq:picard-itr}. Suppose $f(t,\cdot)$ is uniformly $L$-Lipschitz in $x$ and define $M = \max_{0 \le k \le K} \|x_{kh}^{(0)} - x_{kh}\|.$
Then, for any $r \ge 0$ and $0 \le k \le K$, 
\begin{align}
\|x_{kh}^{(r)} - x_{kh}\| \le \frac{(khL)^r}{r!} \cdot M.
\end{align}
In particular, $R = \max\{e^2 LT, \ln(\nicefrac{M}{\varepsilon})\}$ Picard iterations suffice to ensure the following:
\begin{align}
\max_{k \in \{0,\ldots,K\}} \|x_{kh}^{(R)} - x_{kh}\| \le \varepsilon.
\end{align}
\end{lemma}
\begin{proof}
Let $\Delta_{k,r} = \|x_{kh}^{(r)} - x_{kh}\|.$
We shall prove that $\Delta_{k,r} \le \frac{(khL)^r}{r!} \cdot M$ by induction on $r$. Clearly, this is true for $r=0$. Now, suppose this holds for $j=1,\ldots,r$. Then, by the uniform $L$-Lipschitzness of $f(t,\cdot)$,
\begin{align}
\Delta_{k,r+1}
&= \|x_{kh}^{(r+1)} - x_{kh}\| \nonumber\\
&\le h \sum_{j=0}^{k-1} \|f(jh,x_{jh}^{(r)}) - f(jh,x_{jh})\| \nonumber\\
&\le hL \sum_{j=0}^{k-1} \Delta_{j,r} \nonumber\\
&\le M \frac{(hL)^{r+1}}{r!} \sum_{j=0}^{k-1} j^r \nonumber\\
&\le M \frac{(hL)^{r+1}}{r!} \int_0^k t^r\,dt \nonumber\\
&\le \frac{(khL)^{r+1}}{(r+1)!}\, M.
\end{align}
Hence, by induction, $\Delta_{k,r} \le \frac{(khL)^r}{r!}\, M.$ Now, for $R \ge \max\{e^2LT,\ln(\nicefrac{M}{\varepsilon})\},$
\begin{align}
\max_{0 \le k \le K} \Delta_{k,R} \le M \cdot \frac{(KhL)^R}{R!} \le M \cdot \Bigl(\frac{eLT}{R}\Bigr)^R \le M e^{-R} \le \varepsilon.
\end{align}
\end{proof}
Generally, given some accuracy parameter $\varepsilon$, $T = \Tilde{\Theta}(1)$ and $h = O(\poly(\nicefrac{1}{\varepsilon}))$. Hence, the naive sequential algorithm for solving \cref{eq:discrete-diff} requires $K = \Otilde(\poly(\nicefrac{1}{\varepsilon}))$ oracle calls. In contrast, when $f$ is $\Otilde(1)$-Lipschitz, approximately solving \cref{eq:discrete-diff} to $O(\poly(\varepsilon))$ error (which generally suffices for downstream applications) requires $R = \Otilde(\ln(\nicefrac{1}{\varepsilon}))$ parallel oracle calls, leading to a significant parallel speedup at the cost of a benign (and easily tunable) approximation error.

\subsection{TAP Free Energy and the Tilted Mean Estimator}
The construction of the tilted mean estimator (\cref{alg:tap-amp}) of \citet{el2025sampling-p-spin} is motivated by a deep structural property of the $p$-spin model: The mean of the tilted Gibbs measure at 
high temperature is an approximate stationary point of the Thouless--Anderson--Palmer (TAP) free energy \cite{thouless1977solution,mezard1988spin, talagrand2010mean}, which we briefly introduce below. 

Recall that the mixture function of the $p$-spin model is defined as $\xi(t) = \sum_{p=2}^{P} \beta^2_p t^p$. For any $t \geq 0$, the asymptotic overlap $q^{\star}(t)$ is defined as follows:
\begin{align}
\label{eq:asymptotic-overlap}
    q_{k+1}(t) = \E*_{W \sim N(0,1)}{\operatorname{tanh}\left(t + \xi^{\prime}(q_k) +W \sqrt{\xi^{\prime}(q_k) + t}\right)^2}, \qquad q_k(t) = 0, \qquad q^{\star}(t) = \lim_{k \to \infty} q_k(t)
\end{align}
Since $q_k(t)$ is independent of $G$, computable to high accuracy via a one-dimensional Gaussian integral, and approaches the limit $q^{\star}(t)$ exponentially fast, we assume, for simplicity, that $q^{\star}(t)$ is known exactly for $t=kh, k \in [K]$. 

The TAP free energy associated with the $p$-spin Hamiltonian in \cref{eq:p-spin-model} is defined as follows:
\begin{align}
    \mathscr{F}_{\mathrm{TAP}}(m,y) &= -H(m) - \dotp{y}{m} - \sum_{i=1}^{n} h(m_i) - \operatorname{ONS}(Q(m)) \\
    \operatorname{ONS}(q) &= \frac{n}{2}\left[\xi(1)-\xi(Q)-(1-Q)\xi^{\prime}(Q)\right] \\
    Q(m) &= n^{-1}\|m\|^2, \qquad h(m) = -\frac{1+m}{2}\ln\left(\frac{1+m}{2}\right) - \frac{1-m}{2}\ln\left(\frac{1-m}{2}\right) 
\end{align}
For the sake of tractability, \citet{el2025sampling-p-spin} also define an approximate version of the TAP free energy which replaces the $\operatorname{ONS}(Q(m))$ term by a quadratic approximation.
\begin{align}
\label{eq:approx-tap}    
\hat{\mathscr{F}}_{\mathrm{TAP}}(m \ ;y, q) &= -H(m) - \dotp{y}{m} - \sum_{i=1}^{n} h(m_i) - \operatorname{ONS}(q) - \operatorname{ONS}^{\prime}(q)(Q(m)-q) + \frac{n \Gamma}{8}(Q(m)-q)^2
\end{align}
where $q$ is set to be the asymptotic overlap $q^{\star}(t)$ for $t=kh$, and $\Gamma = \Theta(1)$ is a tunable parameter depending only on $\xi$. \cite{el2025sampling-p-spin} show that the mean of the tilted Gibbs measure continues to be an approximate stationary point of $\hat{\mathscr{F}}_{\mathrm{TAP}}$. The tilted mean estimator in \cref{alg:tap-amp} then works in two phases. The first phase runs a series of AMP updates to compute an approximate stationary point $\hat{m}^{K_{\mathrm{AMP}}}$ for $\hat{\mathscr{F}}_{\mathrm{TAP}}$ around which it is locally strongly convex. The second stage exploits this local strong convexity to refine the AMP estimate via natural gradient descent.
\begin{lemma}[$\tanh$ is 1-Lipschitz]
\label{lem:tanh_lip}
For all $a,b\in\R^n$,
\[
\norm{\tanh(a)-\tanh(b)} \le \norm{a-b}.
\]
\end{lemma}

\begin{lemma}[$\sech^2$ is globally 1-Lipschitz on $\R$]
\label{lem:sech2_lip}
For all $s,t\in\R$,
\[
\big|\sech^2(s)-\sech^2(t)\big| \le |s-t| \frac{4}{3\sqrt{3}}.
\]
\end{lemma}

\section{High Accuracy Parallel Sampler for the $p$-Spin Model}
In this section, we analyze \cref{alg:parallel-p-spin-sampler} and prove the following Theorem.

\begin{theorem}[Analysis of \cref{alg:parallel-p-spin-sampler}]
\label{thm:parallel-p-spin-sampler-main}
For any $\varepsilon > 0$, the output $X^{(T)}$ of \cref{alg:parallel-p-spin-sampler} run with $s = \Theta(n^{\nicefrac{2}{3}} \log(\nicefrac{n}{\varepsilon})^{-\nicefrac{2}{3}})$ and $T = \Theta(n^{1/3}\log(\nicefrac{n}{\varepsilon})^{\nicefrac{5}{3}})$ satisfies $d_{\mathsf{TV}}(\mathsf{Law}(X^{(T)}), \mu) \leq \varepsilon$ with probability at least $1 - e^{-cn}$. Moreover, \cref{alg:parallel-p-spin-sampler} has a parallel runtime of $O(n^{1/3}\polylog(\nicefrac{n}{\varepsilon}))$ with $O(\poly(n, \log(\nicefrac{n}{\varepsilon})))$ work.
\end{theorem}

\subsection{Technical Preliminaries}
To analyze the mixing time of \cref{alg:parallel-p-spin-sampler}, we use the following result on the rapid mixing of $s$-Glauber dynamics for the $p$-spin model at high temperature. This result is implicit in \cite[Theorem 6]{anari2024universality} and also appears in \cite{lee2023parallelising}.
\begin{theorem}[Rapid Mixing of $s$-Glauber Dynamics]
\label{thm:s-glauber-rapid-mixing}
Let $P_s$ denote the kernel of $s$-Glauber dynamics for the $p$-spin Gibbs measure. Then, there exists an absolute constant $C > 0$ such that if $\mathfrak{C}(\beta) \leq C$, then the following holds with probability at least $1-\exp(-\Omega(n))$ for any initial law $\nu$
\begin{align}
    \KL{\nu P_s^{t}}{\mu} \leq \exp(-\tfrac{\alpha s t}{n}) \KL{\nu}{\mu}
\end{align}
where $\alpha$ is a constant depending only on $\mathfrak{D}(\beta)$.
\end{theorem}
We also use the following result, which proves that Restricted Gaussian Dynamics samples from an Ising model on $\{ \pm 1 \}^{m}$ up to $\varepsilon$ accuracy in total variation in $\polylog(\nicefrac{m}{\varepsilon})$ time and $\poly(m, \log(\nicefrac{1}{\varepsilon}))$ work. This result appears in \cite{chen2025efficient}, but with a suboptimal work dependence of $\poly(\nicefrac{m}{\varepsilon})$. 
\begin{theorem}[RNC Sampling of Ising Models via Restricted Gaussian Dynamics]
\label{thm:rgd-parallel-runtime}
Consider the Ising model $\nu_{J,h}(x) \propto \exp(\tfrac{1}{2}\dotp{x}{Jx} + \dotp{h}{x})$ where $J \in \R^{m \times m}$ satisfies $\|J\|_{\op} \leq 1-c$ for some numerical constant $c$. Then, \cref{alg:rgd-ising-parallel} outputs a sample $X^{(T)}$ satisfying $\TV{X^{(T)},\nu_{J,h}} \leq \epsilon$ in $O(\log^3(\nicefrac{m}{\varepsilon}) )$ parallel time using $O(\poly(m, \log(\nicefrac{m}{\varepsilon})))$ work.
\end{theorem}
\begin{proof}
By adding a constant multiple of the identity if necessary (which does not change the distribution $\nu_{J,h}$), we can assume that $\tfrac{c}{2} I \preceq J \preceq (1-\tfrac{c}{2})I$.
By Proposition 16, Proposition 27, and Theorem 49 of \cite{chen2022localization}, taking $T = \Theta(\log(\nicefrac{m}{\varepsilon}))$ suffices to ensure $\TV{X^{(T)},\nu_{J,h}} \leq \epsilon$. By \cite{csanky1975fast}, computing $J^{-1}$ and sampling from the Gaussian density $\cN(X^{(t)}, J^{-1})$ can be performed on a \Class{PRAM} in $O(\log^2(m))$ parallel time and $O(m^{\omega})$ work, where $2 \leq \omega \leq 3$ is the matrix multiplication exponent. Sampling from the product measure $p(x) \propto \exp(\dotp{x}{h+Jy})$ can be done in $O(1)$ parallel time on a \Class{PRAM} with $O(m)$ work. Hence, the overall parallel runtime is $O(\log^3(\nicefrac{m}{\varepsilon}))$ with work $O(\poly(m, \log(\nicefrac{m}{\varepsilon})))$. 
\end{proof}
\subsection{Bounds on the Hamiltonian of the Conditional Posterior}
Consider any arbitrary $S \in \binom{[n]}{s}$. For any $x \in \{\pm 1\}^n$, let $x = (y,\tau)$ where $y = x_S$ denotes the block spins in $S$ and $\tau = x_{\Bar{S}}$ denotes the pinning outside $S$. Then the conditional distribution $\mu_{S,\tau}(y) = \mu(y  \ | x_{\Bar{S}} = \tau)$ of the spins in $S$ satisfies $\mu_{S,\tau}(y) \propto \exp(H_{S,\tau}(y))$ where:
\begin{align}
\label{eqn:conditional-hamiltonian-decomposition}
H_{S,\tau}(y) = \dotp{b_{S,\tau}}{y} + \tfrac{1}{2} \dotp{y}{A_{S,\tau} y} + R_{S,\tau}(y)
\end{align}
Here, $A_{S,\tau}$ is a symmetric matrix with zero diagonals with:
\begin{align}
\label{eqn:Ast-entries}
(A_{S,\tau})_{ij} = \sum_{p=2}^{P} \frac{\beta_p \sqrt{p!}}{n^{\tfrac{p-1}{2}}} \sum_{K \subseteq \Bar{S}, \abs{K} = p-2 } G_{K \cup \{i,j\}} \tau^K, \qquad i,j \in S, \ i\neq j
\end{align}
The remainder term is given by 
\begin{align}
\label{eq:remainder-term}
    R_{S,\tau}(y) =
\sum_{p=3}^P \frac{\beta_p\sqrt{p!}}{n^{(p-1)/2}}
\sum_{\substack{J\subseteq [n],  \abs{J}=p,\\ \abs{J\cap S}\ge 3}}
 G_J \  \,y^{J\cap S} \  \,\tau^{J\cap \Bar{S}}\,
\end{align}
As discussed earlier, each step of $s$-Glauber dynamics requires sampling from $\mu_{S,\tau}$, which can be computationally expensive. We address this via a parallelized implementation of rejection sampling in \cref{alg:parallel-ising-rejection-sampler} where $\mu_{S,\tau}(y)$ is approximated by an Ising model of the form $\pi_{S,\tau}(y) \propto \exp(\tfrac{1}{2} \dotp{y}{A_{S,\tau} y} + \dotp{b_{S,\tau} + \hat{h}}{y})$, where $\hat{h}$ is an external field that satisfies $\hat{h} \approx \E_{y \sim \pi_{S,\tau}}{\nabla R_{S,\tau}(y)}$. 

To ensure the parallel efficiency of \cref{alg:parallel-ising-rejection-sampler}, we generate approximate samples from the Ising proposal $\pi_{S,\tau}$ using Restricted Gaussian Dynamics (\cref{alg:rgd-ising-parallel}).
To this end, the following theorem proves that $\|A_{S,\tau}\|_{\op} \leq \nicefrac{1}{4}$ holds with high probability uniformly over all choices of $S$ and $\tau$ whenever $s =  O(n^{\nicefrac{2}{3}})$. Hence, \cref{alg:rgd-ising-parallel} can approximately sample from the proposal in $\Otilde(1)$ parallel time. 

\begin{theorem}[Uniform Operator Norm Bound for Ising Proposal]
\label{thm:proposal-uniform-op-norm}
There exists a numerical constant $C > 0$ such that for $\mathfrak{C}(\beta) \leq C$ and $s = o(n)$, the following holds with probability at least $1-\exp(-cn)$:
\begin{align}
    \|A_{S,\tau}\|_{\op} \leq \frac{1}{4}, \qquad \forall \ S \in \binom{[n]}{s}, \
    \tau \in \{\pm 1\}^{\Bar{S}}
\end{align}
\end{theorem}
\begin{proof}
From \cref{eqn:Ast-entries},
\begin{align}
(A_{S,\tau})_{ij}
=
\sum_{p=2}^{P}
\frac{\beta_p \sqrt{p!}}{n^{(p-1)/2}}
\sum_{\substack{K \subseteq \bar S \\ |K| = p-2}}
G_{K \cup \{i,j\}} \tau^K,
\qquad i,j \in S,\ i \neq j.
\end{align}

Since $\tau \in \{\pm 1\}^{\bar S}$ and $g_{K \cup \{i,j\}} \stackrel{\mathrm{iid}}{\sim} \mathcal{N}(0,1)$ for each $K$ and $p$, $(A_{S,\tau})_{ij}$ are independent centered Gaussians with
\begin{align}
\mathsf{Var}\bigl[(A_{S,\tau})_{ij}\bigr]
&=
\sum_{p=2}^{P}
\frac{\beta_p^2 p!}{n^{p-1}}
\binom{n-s}{p-2} \leq
\frac{1}{n}
\sum_{p=2}^{P}
\beta_p^2 \frac{p!}{(p-2)!}
\leq
\frac{\mathfrak{C}(\beta)^2}{n}.
\end{align}

Let $\mathcal{N}$ be a $1/4$-net of the unit sphere in $\mathbb{R}^S$. By a standard covering argument, $\|A_{S,\tau}\|
\leq
2 \max_{u,v \in \mathcal{N}} u^\top A_{S,\tau} v$. Since $(A_{S,\tau})_{ij}$ are independent centered Gaussians for $i \neq j$, and $(A_{S,\tau})_{ii} = 0$, $u^\top A_{S,\tau} v$ is a centered Gaussian for any $u, v \in \cN$ with,
\begin{align}
\mathsf{Var}\bigl[u^\top A_{S,\tau} v\bigr]
=
\sum_{i \neq j}
\mathsf{Var}\bigl[(A_{S,\tau})_{ij}\bigr] (u_i v_j)^2 \leq
\frac{\mathfrak{C}(\beta)^2}{n}.
\end{align}

Since $|\mathcal{N}| \leq 9^s$, by a union bound
\begin{align}
\Prb\bigl[\|A_{S,\tau}\|_{\mathrm{op}} \geq t\bigr]
\leq
9^{2s} \exp\bigl(- \tfrac{n t^2}{\mathfrak{C}(\beta)^2}\bigr).
\end{align}

Taking a union bound over all $(S,\tau)$ and choosing $\mathfrak{C}(\beta) = O(1)$ sufficiently small,
\begin{align}
\Prb\bigl[\exists\,(S,\tau)\ \text{s.t.}\ \|A_{S,\tau}\|_{\mathrm{op}} \geq \tfrac{1}{4}\bigr]
&\leq
\binom{n}{s} 2^{n-s} 9^{2s} \exp\bigl(- \tfrac{cn}{\mathfrak{C}(\beta)^2}\bigr) \\
&\leq
\exp\bigl[n \ln 2 + c_1 s \ln n - \tfrac{cn}{\mathfrak{C}(\beta)^2}\bigr] \leq
\exp(-c_2 n).
\end{align}    
\end{proof}
Naturally, the success probability of our rejection sampling procedure is governed by how well the Ising proposal $\pi_{S,\tau}$ approximates the conditional distribution $\mu_{S,\tau}$, or, in other words, how large the influence of the remainder term $R_{S,\tau}$ is. We quantify this via a uniform bound on the discrete Hessian of $R_{S,\tau}$.

\begin{theorem}[Uniform Bound on Hessian of Remainder Term]
\label{thm:remainder-hessian-bound}
There exists a numerical constant $C > 0$ such that for $\mathfrak{C}(\beta) \leq C$, the following holds with probability at least $1-\exp(-cn)$ whenever $ \Omega(\sqrt{n}) \leq s \leq o(n)$
\begin{align}
    \|\nabla^2 R_{S,\tau}(y)\|^2_{F} \leq \frac{C_1 \mathfrak{C}(\beta)^2 s^3}{n^2}, \qquad \forall \ S \in \binom{[n]}{s}, \  (y,\tau) \in \{\pm 1\}^n
\end{align}
\end{theorem}
Moreover, there exists a $C_P = \Theta(1)$ such that the following holds with probability at least $1-\exp(-cn)$
\begin{align}
    \|\nabla R_{S,\tau}(y)\|_{\infty} \leq n^{C_P}, \qquad \forall \ S \in \binom{[n]}{s}, \  (y,\tau) \in \{\pm 1\}^n
\end{align}
\begin{proof}
We first derive a uniform Frobenius norm bound on the Hessian. From \cref{eq:remainder-term},
\begin{align}
\bigl(\nabla^2 R_{S,\tau}\bigr)_{ij}
=
\sum_{p=3}^P
\frac{\beta_p \sqrt{p!}}{n^{\tfrac{p-1}{2}}}
\sum_{\substack{J \subseteq [n],\, |J| = p \\ |J \cap S| \ge 3 \\ \{i,j\} \subseteq J}}
G_J \tau^{J \cap \overline{S}} y^{J \cap S \setminus \{i,j\}}.
\end{align}
Let $\mathcal{J} = \{J \subseteq [n] : 3 \le |J| \le P\}$ and let $g = (G_J)_{J \in \mathcal{J}}$. Note that $g$ is a standard Gaussian in $\mathbb{R}^{\mathcal{J}}$. Writing $x = (y,\tau)$, define the matrix $B = (B_{I,J})_{I \in \binom{S}{2},\, J \in \mathcal{J}}$ by
\begin{align}
\label{eq:B-matrix-elements}
B_{I,J}
=
\frac{\beta_{|J|}\sqrt{|J|!}}{n^{\tfrac{|J|-1}{2}}}
\mathbbm{1}(I \subseteq J,\, |J \cap S| \ge 3)\, x^{J}.
\end{align}
Note that for $I = \{i,j\}$, $i,j \in S$,
\begin{align}
\bigl(\nabla^2 R_{S,\tau}(y)\bigr)_{ij}
&=
\sum_{J \in \mathcal{J}}
\frac{\beta_{|J|}\sqrt{|J|!}}{n^{\tfrac{|J|-1}{2}}}
\mathbbm{1}(I \subseteq J,\, |J \cap S| \ge 3)\, x^{J \setminus I} g_J =
\frac{1}{x^I}\sum_{J \in \mathcal{J}} B_{I,J} g_J
=
\frac{1}{x^I}(Bg)_I.
\end{align}
Since $x \in \{\pm 1\}^n$, $\|\nabla^2 R_{S,\tau}(y)\|_F^2 = \|Bg\|^2$. Let $M = BB^\top$. Then, by the Hanson-Wright inequality (\citet{vershynin2018high}, Theorem 6.2.1),
\begin{align}
\label{eq:hanson-wright-bound}
\Prb\bigl[\|\nabla^2 R_{S,\tau}(y)\|_F^2 \ge 2 \Tr(M)\bigr]
\le
\exp\Bigl(
-c \min\Bigl\{
\frac{\Tr(M)^2}{\|M\|_F^2},
\frac{\Tr(M)}{\|M\|_{\op}}
\Bigr\}
\Bigr).
\end{align}
We now compute $\Tr(M)$, $\|M\|_{\op}$ and $\|M\|_F$.

Since $x \in \{\pm 1\}^n$, we obtain the following from \cref{eq:B-matrix-elements} for any $I_1,I_2 \in \binom{S}{2}$,
\begin{align}
\label{eq:M-matrix-def}
M_{I_1,I_2}
&=
\sum_{J \in \mathcal{J}} B_{I_1,J} B_{I_2,J} =
\sum_{p=3}^P \frac{\beta_p^2 p!}{n^{p-1}}\, N_p(I_1,I_2), \\
N_p(I_1,I_2)
&=
\#\Bigl\{
J \in \binom{[n]}{p} :
I_1 \cup I_2 \subseteq J,\,
|J \cap S| \ge 3
\Bigr\}.
\end{align}
Clearly, $N_p(I_1,I_2)$ depends only on $|I_1 \cap I_2| \in \{0,1,2\}$. Hence, we consider three cases.

\textbf{Case 1: $I_1 = I_2 = I$. } 

Write $J = I \cup L$ where $L \cap I = \varnothing$. Since $|J \cap S| \ge 3$, we have $|L \cap S| \ge 1$, i.e. $|L \cap (S \setminus I)| \ge 1$. Hence,
\begin{align}
N_p(I_1,I_2)
&=
\#\{L \subseteq [n] \setminus I : |L| = p-2\}
-
\#\{L \subseteq [n] \setminus I : L \cap (S \setminus I) = \varnothing,\ |L| = p-2\} \nonumber \\
&=
\binom{n-2}{p-2} - \binom{n-s}{p-2}.
\label{eq:case1-bound}
\end{align}

\textbf{Case 2: $|I_1 \cap I_2| = 1$.} 

Since $I_1 \cup I_2 \subseteq S$ and $|I_1 \cup I_2| = 3$, any $p$-set $J \supseteq I_1 \cup I_2$ satisfies $|J \cap S| \ge 3$. Hence,
\begin{align}
\label{eq:case2-bound}
N_p(I_1,I_2)
=
\#\{J \in \binom{[n]}{p} : I_1 \cup I_2 \subseteq J\}
=
\binom{n-3}{p-3}.
\end{align}

\emph{Case 3:} $I_1 \cap I_2 = \varnothing$.

Then $I_1 \cup I_2 \subseteq S$ and $|I_1 \cup I_2| = 4$. If $p=3$, no such $J$ can exist. If $p \ge 4$, any $p$-set $J$ containing $I_1 \cup I_2$ satisfies $|J \cap S| \ge 3$. Hence,
\begin{align}
\label{eq:case3-bound}
N_p(I_1,I_2)
=
\begin{cases}
\binom{n-4}{p-4}, & p \ge 4,\\
0, & \text{otherwise}.
\end{cases}
\end{align}

To combine the three cases, we define
\begin{align}
\Gamma_0
&=
\sum_{p=4}^P
\frac{\beta_p^2 p!}{n^{p-1}}
\binom{n-4}{p-4}, \\
\Gamma_1
&=
\sum_{p=3}^P
\frac{\beta_p^2 p!}{n^{p-1}}
\binom{n-3}{p-3}, \\
\Gamma_2
&=
\sum_{p=3}^P
\frac{\beta_p^2 p!}{n^{p-1}}
\Biggl[
\binom{n-2}{p-2} - \binom{n-s}{p-2}
\Biggr].
\end{align}
Then, by \cref{eq:M-matrix-def}, \cref{eq:case1-bound}, \cref{eq:case2-bound} and \cref{eq:case3-bound}, we obtain the following:
\begin{align}
\label{eq:M-matrix-elements}
M_{I_1,I_2}
=
\begin{cases}
\Gamma_2, & I_1 = I_2,\\
\Gamma_1, & |I_1 \cap I_2| = 1,\\
\Gamma_0, & I_1 \cap I_2 = \varnothing.
\end{cases}
\end{align}
It follows that,
\begin{align}
\label{eq:trace-M}
\Tr(M)
=
\sum_{I \in \binom{S}{2}} M_{I,I}
=
\binom{s}{2} \Gamma_2.
\end{align}
From \cref{eq:M-matrix-elements}, we observe that $M$ is symmetric with non-negative entries and constant row sums. Hence, by the Perron-Frobenius theorem, 
\begin{align}
\|M\|_{\op}
=
\max_{I_1} \sum_{I_2} M_{I_1,I_2}.
\end{align}
For any $I_1 \in \binom{S}{2}$, there are exactly $2(s-2)$ possible $I_2$ such that $|I_1 \cap I_2| = 1$ and $\binom{s-2}{2}$ possible $I_2$ such that $I_1 \cap I_2 = \varnothing$. Hence,
\begin{align}
\label{eq:op-M}
\|M\|_{\op}
=
\Gamma_2 + 2(s-2)\Gamma_1 + \binom{s-2}{2} \Gamma_0.
\end{align}
By a similar argument,
\begin{align}
\label{eq:frob-M}
\|M\|_F^2
&=
\sum_{I_1,I_2} M_{I_1,I_2}^2 \\
&=
\binom{s}{2}
\Biggl[
\Gamma_2^2 + 2(s-2)\Gamma_1^2 + \binom{s-2}{2} \Gamma_0^2
\Biggr].
\end{align}

To apply \cref{eq:hanson-wright-bound}, we need to lower bound $\Tr(M)/\|M\|_{\op}$ and $\Tr(M)^2/\|M\|_F^2$. To this end, define
\begin{align}
A_3
=
\sum_{p=2}^P \beta_p^2 p(p-1)(p-2)
\le
\mathfrak{C}(\beta)^2.
\end{align}
Then,
\begin{align}
\label{eq:gamma1-ub}
\Gamma_1
=
\sum_{p=3}^P
\frac{\beta_p^2 p!}{n^{p-1}}
\binom{n-3}{p-3} \le
\frac{1}{n^2}
\sum_{p=3}^P
\beta_p^2 \frac{p!}{(p-3)!}
=
\frac{A_3}{n^2},
\end{align}
\begin{align}
\label{eq:gamma0-ub}
\Gamma_0 =
\sum_{p=4}^P
\frac{\beta_p^2 p!}{n^{p-1}}
\binom{n-4}{p-4} \le
\frac{1}{n^3}
\sum_{p=4}^P
\beta_p^2 \frac{p!}{(p-4)!}
\le
\frac{P A_3}{n^3}.
\end{align}

To control $\Gamma_2$, recall the identity
\begin{align}
Q_p
:=
\binom{n-2}{p-2} - \binom{n-s}{p-2}
=
\sum_{i=0}^{s-3} \binom{n-s+i}{p-3}.
\end{align}
Then,
\begin{align}
(s-2)\binom{n-s}{p-3}
\le
Q_p
\le
(s-2)\binom{n-3}{p-3}.
\end{align}
It follows that,
\begin{align}
\label{eq:gamma2-ub}
\Gamma_2
=
\sum_{p=3}^P
\frac{\beta_p^2 p!}{n^{p-1}} Q_p \le
(s-2)\sum_{p=3}^P
\frac{\beta_p^2 p!}{n^{p-1}}
\binom{n-3}{p-3} \le
\frac{s A_3}{n^2}.
\end{align}
Similarly, we can lower bound $\Gamma_2$ as follows:
\begin{align}
\label{eq:gamma2-lb}
\Gamma_2
\ge
(s-2)\sum_{p=3}^P
\frac{\beta_p^2 p!}{n^{p-1}}
\binom{n-s}{p-3} \ge
\frac{s-2}{n^2}
\sum_{p=3}^P
\frac{\beta_p^2 p!}{(p-3)!}
\prod_{j=0}^{p-4} \Bigl(1 - \frac{s+j}{n}\Bigr).
\end{align}
For $n \ge \Theta(P)$, since $s = o(n)$, $\tfrac{s+j}{n} \leq \nicefrac{1}{2}$ for any $j \leq P-4$. Hence, from \cref{eq:gamma2-ub} and \cref{eq:gamma2-lb}, we conclude:
\begin{align}
\label{eq:gamma2-bothsides}
c_P \frac{s A_3}{n^2}
\le
\Gamma_2
\le
\frac{s A_3}{n^2}.
\end{align}

From \cref{eq:trace-M}, \cref{eq:op-M}, \cref{eq:frob-M} \cref{eq:gamma1-ub}, \cref{eq:gamma0-ub} and \cref{eq:gamma2-bothsides}, we obtain:
\begin{align}
\Tr(M)
\ge
c_P \frac{s^3 A_3}{n^2}, \qquad
\|M\|_{\op}
\le
C_P \frac{s A_3}{n^2}, \qquad
\|M\|_F^2
\le
C_P \frac{s^4 A_3^2}{n^4}.
\end{align}
It follows that,
\begin{align}
\label{eq:exponent-bound}
\min\Bigl\{
\frac{\Tr(M)}{\|M\|_{\op}},
\frac{\Tr(M)^2}{\|M\|_F^2}
\Bigr\}
\ge
c_P s^2.
\end{align}
Moreover, by \cref{eq:trace-M} and \cref{eq:gamma2-bothsides}
\begin{align}
\label{eq:tail-bound}
\Tr(M)
=
\binom{s}{2} \Gamma_2
\le
\frac{s^3 A_3}{n^2}
\le
\frac{s^3 \mathfrak{C}(\beta)^2}{n^2}.
\end{align}

Substituting \cref{eq:exponent-bound} and \cref{eq:tail-bound} into \cref{eq:hanson-wright-bound}, we obtain:
\begin{align}
\Prb\Bigl[
\|\nabla^2 R_{S,\tau}(y)\|_F^2
\ge
\frac{s^3 \mathfrak{C}(\beta)^2}{n^2}
\Bigr]
\le
\exp(-c_P s^2).
\end{align}

Finally, taking a union bound over all $S \subseteq [n]$, $|S|=s$, and $(y,\tau) \in \{\pm 1\}^n$ and using $s \geq \Omega(\sqrt{n})$, we obtain:
\begin{align}
\Prb\Biggl[
\max_{S \in \binom{[n]}{s}}
\max_{(y,\tau) \in \{\pm 1\}^n}
\|\nabla^2 R_{S,\tau}(y)\|_F^2
\ge
\frac{s^3 \mathfrak{C}(\beta)^2}{n^2}
\Biggr] &\le \binom{n}{s} 2^n \exp(-c_P s^2) \\
&\le 
\exp(-c_P s^2 + c_1 n + c_2 s \log n) \\
&\le
\exp(-c_P s^2),
\end{align}

To derive the uniform bound on the gradient, we observe that.
\begin{align}
\bigl|\{J \subseteq [n] : 2 \le |J| \le P\}\bigr|
=
\sum_{p=2}^P \binom{n}{p}
\le
C_P n^P.
\end{align}
Then, by a union bound and Gaussian concentration,
\begin{align}
\Prb\Bigl[\max_{\substack{J \subseteq [n] \\ 2 \le |J| \le P}} |G_J| > n\Bigr]
\le
C_P n^P e^{-n^2/2}
\le
e^{-cn}.
\end{align}
Then, by \cref{eq:remainder-term},
\begin{align}
\partial_i R_{S,\tau}(y)
=
\sum_{p=3}^P
\frac{\beta_p \sqrt{p!}}{n^{\tfrac{p-1}{2}}}
\sum_{\substack{J \subseteq [n],\, |J| = p \\
|J \cap S| \ge 3,\ J \ni i}}
G_J \tau^{J \cap \overline{S}} y^{(J \cap S)\setminus \{i\}}.
\end{align}
Thus, on the event $\max_J |g_J| \le n$,
\begin{align}
|\partial_i R_{S,\tau}(y)|
\le
\sum_{p=3}^P
\frac{\beta_p \sqrt{p!}}{n^{\tfrac{p-1}{2}}}
\sum_{\substack{J \subseteq [n] \\
|J| = p,\ i \in J}}
|g_J| 
\le
n \sum_{p=3}^P
\frac{\beta_p \sqrt{p!}}{n^{\tfrac{p-1}{2}}}
\binom{n-1}{p-1} 
\le
n^{C_P}.
\end{align}
Since the above bound holds uniformly in $(S,\tau,y,i)$, we conclude that there exists a $C_P = \Theta(1)$ such that $\sup_{S,\tau,y} \|\nabla R_{S,\tau}(y)\|_\infty \le n^{C_P}.$
\end{proof}

\subsection{Analysis of Parallel Rejection Sampling}
\label{sec:parallel-ising-rejection-sampler-analysis}
In this section, we shall establish the accuracy and parallel efficiency of \cref{alg:parallel-ising-rejection-sampler}. Henceforth, we condition on the event $\cE$ that the bounds in \cref{thm:proposal-uniform-op-norm} and \cref{thm:remainder-hessian-bound} hold uniformly over all $S$ and $\tau$.

\begin{theorem}[Analysis of Parallel Rejection Sampler]
\label{thm:parallel-ising-rejection-sampler-main}
For any $S \in \binom{[n]}{s}$ and $\tau \in \{\pm 1\}^{\overline{S}}$,
let $\tilde{\mu}_{S,\tau}$ denote the law of the sample output by
\cref{alg:parallel-ising-rejection-sampler}. Conditioned on the event $\cE$, the following
holds whenever $s = \Theta(n^{\nicefrac{2}{3}})$:
\begin{align}
    d_{\mathsf{TV}}\!\left(\tilde{\mu}_{S,\tau}, \mu_{S,\tau}\right)
    \leq \varepsilon_{\mathrm{step}}
    \qquad
    \forall S \in \binom{[n]}{s},\ 
    \tau \in \{\pm 1\}^{\overline{S}} .
\end{align}
Moreover, \cref{alg:parallel-ising-rejection-sampler} has parallel
runtime $O(\polylog(\nicefrac{n}{\varepsilon_{\mathrm{step}}}))$ with $O(\poly(n, \log(\nicefrac{1}{\estep})))$ work.
\end{theorem}

The proof of \cref{thm:parallel-ising-rejection-sampler-main}, which is presented in \cref{sec:proof-parallel-ising-rejection-sampler}, involves several technical components, the first of which is an accuracy guarantee for the approximate rejection sampling step in \cref{alg:parallel-ising-rejection-sampler}. This result appears as Lemma 4.4 in \citet{lee2023parallelising} and Lemma 2 in \cite{fan2023improved}.

\begin{lemma}[Accuracy of Rejection Sampling]
\label{lem:rejection-sampling-accuracy}
Let $P$ and $Q$ be probability measures satisfying $\frac{dP}{dQ}(y) \propto e^{\psi(y)}$. Let $Y,Z \stackrel{\mathrm{iid}}{\sim} Q$, and define $R = \exp(\psi(Y)-\psi(Z))$. Consider the following rejection sampling procedure: Draw $U \sim \Uniform[0,1]$ and accept $Y$ if $U \leq \min \{ 1, \nicefrac{R}{c} \}$ for some $c \geq 1$. Let $\Tilde{P}$ denote the law of the accepted sample. Then,
\begin{align}
    \TV{P, \Tilde{P}} \leq \frac{1}{\bE[R]} \cdot \bE[\max\{R-c,0\}], \qquad
    \Prb\left[Y \textrm{ is accepted}\right] = \frac{1}{c} \cdot \bE[\min\{R,c\}]\geq \tfrac{1}{2c}
\end{align}
In addition, suppose $\psi(Y) - \psi(Z)$ exhibits subexponential tails of the form $\Prb[\psi(Y) - \psi(Z) \geq t] \leq a\exp(-bt)$ for $b > 1$. Then, the following holds:
\begin{align}
    \TV{P,\Tilde{P}} \leq \frac{a}{b-1} \cdot c^{-(b-1)}
\end{align}
\end{lemma}

In \cref{alg:parallel-ising-rejection-sampler}, $(X,Z)$ corresponds to $(U_\ell,V_\ell)$, both of which are approximately distributed as $\pi_{S,\tau}$ (modulo the negligible sampling error of \cref{alg:rgd-ising-parallel}, which we ignore for now). Since $\tfrac{d\mu_{S,\tau}}{d\pi_{S,\tau}}(y) \propto e^{\psi(y)}$, where $\psi(y) = R_{S,\tau}(y) - \langle \hat h, y\rangle$, \cref{lem:rejection-sampling-accuracy} suggests that we can bound the TV error of \cref{alg:parallel-ising-rejection-sampler} by controlling the sampling error of \cref{alg:rgd-ising-parallel} and proving a subexponential tail bound for $|\psi(\widetilde U)-\psi(\widetilde V)|$, where $\widetilde U,\widetilde V \stackrel{\mathrm{iid}}{\sim} \pi_{S,\tau}$, or equivalently, for $|\psi(Y)-\mathbb{E}[\psi(Y)]|$ under $Y \sim \pi_{S,\tau}$. Now, by \cref{thm:ising-ate} and \cref{thm:proposal-uniform-op-norm}, $\pi_{S,\tau}$ satisfies $\nicefrac{4}{3}$-ATE with high probability. Thus, \cref{thm:sambale-sinulis} gives us the following concentration bound:
\begin{align}
    \Prb_{Y \sim\pi_{S,\tau}}\left[|\psi(Y) - \bE_{\pi_{S,\tau}}[\psi(Y)]|\right] \leq 2 \exp(-c \min\{\frac{t^2}{\bE_{\pi_{S,\tau}}[\|\nabla \psi\|^2]}, \frac{t}{\max\limits_{y \in \{\pm 1\}^S} \|\nabla^2 \psi(y)\|_F} \})
\end{align}
Since $\|\nabla^2 \psi(y)\|^2_F = \|\nabla^2 R_{S,\tau}(y)\|_F \lesssim \mathfrak{C}(\beta) \cdot \tfrac{s^3}{n^2}$ holds whp uniformly for any $S,\tau$ and $y$ by \cref{thm:remainder-hessian-bound}. Hence, setting $s = \Theta(n^{\nicefrac{2}{3}})$ and $\mathfrak{C}(\beta)$ sufficiently small, $\max\limits_{y \in \{\pm 1\}^S} \|\nabla^2 \psi(y)\|_F \leq \delta$ for any $\delta = O(1)$. To control $\bE_{\pi_{S,\tau}}[\|\nabla \psi\|^2]$, we note that:
\begin{align}
\label{eq:mean-decomposition}
    \bE_{\pi_{S,\tau}}[\|\nabla \psi\|^2] = \| \ \hat{h} - \bE_{\pi_{S,\tau}}[\nabla R_{S,\tau}] \ \|^2 + \bE_{\pi_{S,\tau}}\left[\| \ \nabla R_{S,\tau} - \bE[\nabla R_{S,\tau}] \ \|^2\right]
\end{align}
Using the Poincare Inequality for $\pi_{S,\tau}$ and \cref{thm:remainder-hessian-bound}, the second term in \cref{eq:mean-decomposition} can be made $O(\delta^2)$ for any $\delta = O(1)$ by setting $s=\Theta(n^{\nicefrac{2}{3}})$ and $\mathfrak{C}(\beta)$ sufficiently small. The centering field $\hat h$ must be chosen so as to make the first term $\|\hat h - \bE_{\pi_{S,\tau}}[\nabla R_{S,\tau}]\|^2$ as small as possible. Since $\pi_{S,\tau} = \nu_{A_{S,\tau},\, b_{S,\tau}+\hat h}$ where $\nu_{J,v}(y) \propto \exp(\frac{1}{2}\langle y, Jy\rangle + \langle v, y\rangle),$
the optimal choice of $\hat h$ is given by the solution to the following fixed point problem:
\begin{align}
\label{eq:centering-fixed-point}
h = \bE_{\nu_{A_{S,\tau},\, h + b_{S,\tau}}}[\nabla R_{S,\tau}].
\end{align}

As we shall demonstrate, \cref{eq:centering-fixed-point} admits a unique solution which can be efficiently approximated by an inexact fixed point iteration, namely \cref{alg:centering-external-field}, which replaces the expectation in \cref{eq:centering-fixed-point} by an empirical average over approximate samples drawn via RGD. In particular, we prove the following guarantee in \cref{prf:centering-analysis}
\begin{theorem}[Analysis of \cref{alg:centering-external-field}]
\label{thm:centering-analysis}
Conditioned on the event $\mathcal{E}$, the output $\hat h$ of \cref{alg:centering-external-field} satisfies $\|\hat{h} - \bE_{\pi_{S,\tau}}[\nabla R_{S,\tau}]\| \leq \eta$ with probability at least $1-\delta$. Moreover, \cref{alg:centering-external-field} runs in $O(\polylog(n, \eta^{-1}, \delta^{-1}))$ parallel time with $O\bigl(\poly(\tfrac{n}{\eta},\ln(\nicefrac{1}{\delta}))\bigr)$ work.
\end{theorem}

\begin{algorithm}[t]
\caption{Centering External Field}
\label{alg:centering-external-field}
\DontPrintSemicolon

\KwIn{$A_{S,\tau}$, $b_{S,\tau}$, $R_{S,\tau}$, accuracy $\eta > 0$, failure probability $\delta \in (0,1)$}

Set iterations
$
T = \Theta\!\left(\log\!\left(\nicefrac{n}{\eta}\right)\right)
$
and sample size
$
K = \Theta\!\left(s\eta^{-2}\log\!\left(\nicefrac{sT}{\delta}\right)\right)
$\;

Initialize $z_0 \gets 0$\;

\For{$t = 0,\ldots,T-1$}{
    \For(In parallel){$i = 1,\ldots,K$}{
        $U_i^{(t)} \sim \textsc{RGD}\!\left(A_{S,\tau},\, b_{S,\tau} + z_t,\, \nicefrac{\delta}{100KT}\right)$\;
    }
    $z_{t+1} = \frac{1}{K}\sum_{i=1}^K \nabla R_{S,\tau}(U_i^{(t)})$\;
}

\KwRet{$\hat h_{S,\tau} = z_T$}
\end{algorithm}

\subsubsection{Proof of \cref{thm:parallel-ising-rejection-sampler-main}}
\label{sec:proof-parallel-ising-rejection-sampler}
\begin{proof}
Recall that $\pi_{S,\tau}(y) \propto \exp\bigl(\tfrac{1}{2}\langle y, A_{S,\tau} y\rangle + \langle b_{S,\tau} + \hat h, y\rangle\bigr)$
 where $\hat h$ is the output of \cref{alg:centering-external-field}. Let $(\widetilde U_\ell,\widetilde V_\ell)_{\ell \in [L]}
\stackrel{\mathrm{iid}}{\sim}
\pi_{S,\tau}$. Then, by \cref{thm:rgd-parallel-runtime}, setting $\varepsilon_{\mathrm{RGD}} = \Theta(\nicefrac{\varepsilon_{\mathrm{step}}}{L})$ suffices to ensure the following:
\begin{align}
\label{eq:perfect-sampling-event}
\Prb\bigl[(U_\ell,V_\ell) = (\widetilde U_\ell,\widetilde V_\ell)\ \forall \ell \leq L\bigr]
\geq
1 - \tfrac{\varepsilon_{\mathrm{step}}}{100}.
\end{align}
Henceforth, condition on the above event, and treat $(U_\ell,V_\ell)$ as if they were exact i.i.d.\ samples from $\pi_{S,\tau}$. Now, define $\psi(y) = R_{S,\tau}(y) - \langle \hat h, y\rangle$. As discussed in \cref{sec:parallel-ising-rejection-sampler-analysis}, since we condition on the event $\cE$ that the bounds in \cref{thm:proposal-uniform-op-norm} and \cref{thm:remainder-hessian-bound} hold, the following holds for some universal constant $\Delta = \Theta(1)$ by setting $s= (n \Delta)^{\nicefrac{2}{3}} \ln(\nicefrac{1}{\estep})^{-\nicefrac{2}{3}}$ and $\mathfrak{C}(\beta) = O(1)$ sufficiently small.
\begin{align}
\max \limits_{y \in \{\pm 1\}^s} \|\nabla^2 \psi(y)\|_F^2 = \max \limits_{y \in \{\pm 1\}^s} \|\nabla^2 R_{S,\tau}(y)\|_F^2
\leq
C_1 \frac{\mathfrak{C}(\beta)^2 s^3}{n^2}
\leq
\frac{\Delta^2}{\ln(\nicefrac{1}{\estep})^2}
\end{align}

By \cref{thm:centering-analysis}, setting $\eta \leq \Delta/8$ and $\delta = \nicefrac{\varepsilon_{\mathrm{step}}}{100}$,
\begin{align}
\label{eq:centering-event}
\Prb\Bigl[\bigl\|\hat h - \bE_{\pi_{S,\tau}}[\nabla R_{S,\tau}]\bigr\| \leq \frac{\Delta}{8\ln(\nicefrac{1}{\estep})} \Bigr]
\geq
1 - \frac{\varepsilon_{\mathrm{step}}}{100}.
\end{align}
Henceforth, we condition on the above event as well. Now,
\begin{align}
\bE_{\pi_{S,\tau}}\bigl[\|\nabla \psi\|^2\bigr]
&=
\bigl\|\hat h - \bE_{\pi_{S,\tau}}[\nabla R_{S,\tau}]\bigr\|^2
+
\bE_{y \sim \pi_{S,\tau}}\Bigl[\bigl\|\nabla R_{S,\tau}(y) - \bE_{\pi_{S,\tau}}[\nabla R_{S,\tau}]\bigr\|^2\Bigr] \\
&\leq
\frac{\Delta^2}{64 \ln(\nicefrac{1}{\estep})^2}
+
\sum_{i \in S} \mathsf{Var}_{\pi_{S,\tau}}[\partial_i R_{S,\tau}].
\end{align}

By \cref{thm:proposal-uniform-op-norm} and \cref{thm:ising-ate}, $\pi_{S,\tau}$ satisfies ATE (and hence, PI) with constant $\nicefrac{4}{3}$. Hence,
\begin{align}
\sum_{i \in S} \mathsf{Var}_{\pi_{S,\tau}}[\partial_i R_{S,\tau}]
\leq
\tfrac{4}{3}
\sum_{i \in S}
\bE_{\pi_{S,\tau}}\bigl[\|\nabla \partial_i R_{S,\tau}\|^2\bigr] 
=
\tfrac{4}{3}
\bE_{\pi_{S,\tau}}\bigl[\|\nabla^2 R_{S,\tau}\|_F^2\bigr] 
\leq
\tfrac{4 \Delta^2}{3 \ln(\nicefrac{1}{\estep})^2}
\end{align}
Hence, $\bE_{\pi_{S,\tau}}\bigl[\|\nabla \psi\|^2\bigr] \leq \tfrac{2\Delta^2}{\ln(\nicefrac{1}{\estep})^2}$. Since $\pi_{S,\tau}$ satisfies ATE with constant $\nicefrac{4}{3}$, by \cref{thm:sambale-sinulis},
\begin{align}
\Prb_{\pi_{S,\tau}}\Bigl[\bigl|\psi(Y) - \bE_{\pi_{S,\tau}}[\psi]\bigr| \geq t\Bigr]
\leq
2 \exp\Bigl(- C \min\bigl\{\frac{t^2 \ln(\nicefrac{1}{\estep})^2}{\Delta^2}, \frac{t \ln(\nicefrac{1}{\estep})}{\Delta}\bigr\}\Bigr).
\end{align}
Taking a union bound and choosing $\Delta = \Theta(1)$ small enough, we obtain:
\begin{align}
\label{eq:rejection-sampling-tail-bound}
\Prb_{\widetilde U,\widetilde V \sim \pi_{S,\tau}}
\Bigl[\bigl|\psi(\widetilde U) - \psi(\widetilde V)\bigr| \geq t\Bigr]
\leq
4e^{-2t\ln(\nicefrac{1}{\estep})}.
\end{align}

Now consider the parallel rejection sampling step in \cref{alg:parallel-ising-rejection-sampler}. By \cref{lem:rejection-sampling-accuracy}, setting $L = \Theta(\Bar{c} \ln(\nicefrac{1}{\varepsilon_{\mathrm{step}}}))$,
we conclude the following:
\begin{align}
\label{eq:accept-event}
\Prb[\text{at least one of the $L$ parallel trials is accepted}]
\geq
1 - \Bigl(1 - \tfrac{1}{2\Bar{c}}\Bigr)^L
\geq
1 - \tfrac{\varepsilon_{\mathrm{step}}}{100}.
\end{align}
Henceforth, we condition on the event above. Now, let $\widetilde\mu_{S,\tau}$ denote the conditional law of the first accepted sample among the $L$ parallel rejection sampling trials. Then, by \cref{lem:rejection-sampling-accuracy}, setting $\Bar{c} = \Theta(1)$ yields the following:
\begin{align}
\label{eq:conditional-tv-bound}
d_{\mathsf{TV}}(\widetilde\mu_{S,\tau}, \mu_{S,\tau})
\leq
4  \bar{c}^{-\ln(\nicefrac{1}{\estep})}
\leq
\tfrac{\varepsilon_{\mathrm{step}}}{100}
\end{align}
To conclude, let $\hat{\mu}_{S,\tau}$ denote the law of the output of \cref{alg:parallel-ising-rejection-sampler}. From \cref{eq:conditional-tv-bound}, \cref{eq:accept-event}, \cref{eq:centering-event} and \cref{eq:perfect-sampling-event}, we conclude the following after taking the appropriate union bounds.
\begin{align}
d_{\mathsf{TV}}\bigl(\hat{\mu}_{S,\tau}, \mu_{S,\tau}\bigr)
\leq
\varepsilon_{\mathrm{step}}.
\end{align}
It remains to analyze the parallel runtime of \cref{alg:parallel-ising-rejection-sampler}. Note that \cref{alg:parallel-ising-rejection-sampler} involves one call to \cref{alg:centering-external-field}, which has parallel runtime $O(\mathsf{polylog}(\nicefrac{n}{\estep}))$ with $O(\mathsf{poly}(n,\ln(1/\varepsilon_{\mathrm{step}})))$ work as per \cref{thm:centering-analysis}; and $L = \Theta(\ln(\nicefrac{1}{\varepsilon_{\mathrm{step}}}))$ \emph{parallel calls} to \cref{alg:rgd-ising-parallel}, each of which requires $\mathsf{polylog}(n/\varepsilon_{\mathrm{step}})$ parallel time with $\mathsf{poly}(n, \ln(\nicefrac{1}{\varepsilon_{\mathrm{step}}}))$ work as per \cref{thm:rgd-parallel-runtime}. Thus, \cref{alg:parallel-ising-rejection-sampler} exhibits a parallel runtime $O(\mathsf{polylog}(n/\varepsilon_{\mathrm{step}}))$ using $\mathsf{poly}(n, \ln(\nicefrac{1}{\varepsilon_{\mathrm{step}}}))$ work.
\end{proof}

\subsection{Proof of \cref{thm:parallel-p-spin-sampler-main}}
\begin{proof}
Let $\mu_t = \mathsf{Law}(X^{(t)})$. Since $\mu_0(x) \propto \exp(\langle h,x\rangle)$, by Lemma~5.5 of \cite{lee2023parallelising},
\begin{align}
\Prb\bigl[\KL{\mu_0}{\mu} \leq Cn\bigr] \geq 1 - \exp(-cn).
\end{align}
Let $\mathcal{G} = \mathcal{E} \cap \bigl\{\KL{\mu_0}{\mu} \leq Cn\bigr\}$. Then, $\Prb[\mathcal{G}] \geq 1 - e^{-cn}$, and henceforth, we condition on $\mathcal{G}$. 

Let $P_s$ and $\widetilde{P}_s$ denote the Markov kernel for $s$-Glauber dynamics and the Markov chain in \cref{alg:parallel-p-spin-sampler} respectively. By \cref{thm:parallel-ising-rejection-sampler-main}, for $s = \Theta((\tfrac{n}{\log(\nicefrac{n}{\varepsilon})})^{\nicefrac{2}{3}})$,
\begin{align}
\max_{x \in \{\pm 1\}^n} d_{\mathsf{TV}}\bigl(P_s(x,\cdot), \widetilde{P}_s(x,\cdot)\bigr)
\leq
\varepsilon_{\mathrm{step}}
=
\tfrac{\varepsilon}{10T}.
\end{align}
Hence,
\begin{align}
d_{\mathsf{TV}}\bigl(\mu_T, \mu_0 P_s^T\bigr)
\leq
T \varepsilon_{\mathrm{step}}
\leq
\tfrac{\varepsilon}{10}.
\end{align}

By \cref{thm:s-glauber-rapid-mixing}, the following holds with probability $1 - e^{-cn}$:
\begin{align}
\KL{\mu_0 P_s^T}{\mu}
\leq
\exp\bigl(-\tfrac{csT}{n}\bigr)\, Cn
\leq
\exp\left(-\frac{cT}{n^{\nicefrac{1}{3}} \log(\nicefrac{n}{\varepsilon})^{\nicefrac{2}{3}}}\right)\, Cn.
\end{align}
Setting $T = \Theta(n^{\nicefrac{1}{3}} \log(\nicefrac{n}{\varepsilon})^{\nicefrac{5}{3}})$, applying Pinsker's inequality and taking appropriate union bounds, we conclude that the following holds with probability at least $1-\exp(-cn)$:
\begin{align}
d_{\mathsf{TV}}(\mu_T,\mu)
\leq
d_{\mathsf{TV}}\bigl(\mu_T,\mu_0 P_s^T\bigr)
+
d_{\mathsf{TV}}\bigl(\mu_0 P_s^T,\mu\bigr) \leq
\tfrac{\varepsilon}{10} + \tfrac{\varepsilon}{2}
\leq
\varepsilon.
\end{align}

To compute the parallel runtime, note that \cref{alg:parallel-p-spin-sampler} makes $T$ calls to \cref{alg:parallel-ising-rejection-sampler} each of which has parallel runtime $O(\mathsf{polylog}(n/\varepsilon))$ with $O(\mathsf{poly}(n/\varepsilon))$ processors as per \cref{thm:parallel-ising-rejection-sampler-main}. Hence, \cref{alg:parallel-p-spin-sampler} exhibits a parallel runtime of $O(n^{1/3}\mathsf{polylog}(n/\varepsilon))$ with $O(\mathsf{poly}(n/\varepsilon))$ processors.
\end{proof}

\subsection{Proof of \cref{thm:centering-analysis}}
\label{prf:centering-analysis}
For any choice of $(S,\tau)$, define the measure $\gamma_z$ on $\{\pm 1\}^s$ as $\gamma_z(y) \propto \exp\Bigl(\tfrac{1}{2}\langle y, A_{S,\tau} y\rangle + \langle b_{S,\tau} + z, y\rangle\Bigr).$
Since we condition on the event $\mathcal{E}$, $\|A_{S,\tau}\|_{\op} \leq \tfrac{1}{4}$. Thus, by \cref{thm:ising-ate}, $\gamma_z$ satisfies LSI with constant $4/3$ for any $z \in \mathbb{R}^s$. Moreover, setting $s = O(n^{2/3})$ and $\mathfrak{C}(\beta) = O(1)$ sufficiently small, one can set $\sup_{y \in \{\pm 1\}^s} \|\nabla^2 R_{S,\tau}(y)\|_F \leq \tfrac{1}{100}.$ Now, define $F : \mathbb{R}^s \to \mathbb{R}^s$ as follows:
\begin{align}
F(z)
=
\bE_{\gamma_z}\bigl[\nabla R_{S,\tau}(y)\bigr]
=
\frac{\sum_{y \in \{\pm 1\}^s} \nabla R_{S,\tau}(y)\exp\bigl(\tfrac{1}{2}\langle y, A_{S,\tau} y\rangle + \langle b_{S,\tau}+z, y\rangle\bigr)}
{\sum_{y \in \{\pm 1\}^s} \exp\bigl(\tfrac{1}{2}\langle y, A_{S,\tau} y\rangle + \langle b_{S,\tau}+z, y\rangle\bigr)}.
\end{align}
Taking derivatives with respect to $z$ and rearranging terms,
\begin{align}
\nabla F(z)
=
\bE_{\gamma_z}\bigl[\nabla R_{S,\tau}(y) y^\top\bigr]
-
\bE_{\gamma_z}\bigl[\nabla R_{S,\tau}(y)\bigr]\bE_{\gamma_z}\bigl[y^\top\bigr].
\end{align}
Consider any unit vectors $u,v \in \mathbb{R}^s$. Then, by Cauchy--Schwarz and the Poincar\'e inequality for $\gamma_z$,
\begin{align}
u^\top \nabla F(z) v
=
\mathrm{Cov}_{\gamma_z}\bigl(\langle u,\nabla R_{S,\tau}\rangle,\langle v,y\rangle\bigr) 
\leq
\sqrt{
\mathrm{Var}_{\gamma_z}\bigl[\langle u,\nabla R_{S,\tau}\rangle\bigr]
\mathrm{Var}_{\gamma_z}\bigl[\langle v,y\rangle\bigr]
} 
\leq
\frac{4}{3}
\sqrt{
\bE_{\gamma_z}\bigl[\langle u,\nabla^2 R_{S,\tau}(y)u\rangle\bigr]
} 
<
\frac{1}{4}.
\end{align}
Since $\|\nabla F(z)\|_{\op} < \tfrac{1}{4}$, by the Banach fixed point theorem, there exists a unique $h^*$ satisfying the following fixed point iteration:
\begin{align}
h^* = F(h^*) = \bE_{\gamma_{h^*}}\bigl[\nabla R_{S,\tau}(y)\bigr].
\end{align}

Recall that for $t \in \{0,\ldots,T-1\}$, $z_{t+1}
=
\frac{1}{K}\sum_{i=1}^K \nabla R_{S,\tau}\bigl(U_i^{(t)}\bigr),$ where $U_i^{(t)}
=
\mathrm{RGD}\Bigl(A_{S,\tau},\, b_{S,\tau}+z_t,\, \tfrac{\delta}{100KT}\Bigr).$
Thus, by \cref{thm:rgd-parallel-runtime}, there exist $Y_i^{(t)} \stackrel{\mathrm{iid}}{\sim} \gamma_{z_t}$ such that the following holds with probability at least $1-\delta$:
\begin{align}
U_i^{(t)} = Y_i^{(t)}
\qquad
\forall\, t<T,\ i \in [K].
\end{align}
Henceforth, we condition on this event, so that $z_{t+1}
=
\tfrac{1}{K}\sum_{i=1}^K \nabla R_{S,\tau}\bigl(Y_i^{(t)}\bigr)$. Since $\|\nabla^2 R_{S,\tau}\|_F \leq \tfrac{1}{100}$ and $\gamma_{z_t}$ satisfies LSI with constant $4/3$, the following holds by Theorem~C and a union bound over $j \in S$:
\begin{align}
\Prb\Bigl[
\Bigl\|
\frac{1}{K}\sum_{i=1}^K \nabla R_{S,\tau}\bigl(Y_i^{(t)}\bigr)
-
\bE_{\gamma_{z_t}}\bigl[\nabla R_{S,\tau}\bigr]
\Bigr\|
\geq u
\Bigr]
\leq
2s e^{-cKu^2/s}.
\end{align}
Setting $u = \frac{3\eta}{16\sqrt{s}},
K = \Theta\bigl(s\eta^{-2}\ln(sT/\delta)\bigr),$ and taking a union bound over $t<T$, we conclude that the following holds with probability at least $1-\delta/2$:
\begin{align}
\|z_{t+1} - F(z_t)\| \leq \frac{3\eta}{16}
\qquad
\forall\, t<T.
\end{align}

Let $\Delta_t = \|z_t - h^*\|.$ Since $F(h^*) = h^*$ and $F$ is $1/4$-Lipschitz,
\begin{align}
\Delta_{t+1}
=
\|z_{t+1} - h^*\| 
=
\|z_{t+1} - F(z_t) + F(z_t) - F(h^*)\| 
\leq
\frac{1}{4}\Delta_t + \frac{3\eta}{16}.
\end{align}
Unrolling the recurrence,
\begin{align}
\Delta_T \leq 4^{-T}\Delta_0 + \frac{\eta}{4}.
\end{align}
Since $z_0 = 0$ and $\|\nabla R_{S,\tau}\|_\infty \leq n^{\Theta(1)}$ when conditioned on the event $\mathcal{E}$,
\begin{align}
\|\Delta_0\|
=
\|h^*\|
=
\|F(h^*)\|
=
\Bigl\|
\bE_{\gamma_{h^*}}\bigl[\nabla R_{S,\tau}\bigr]
\Bigr\|
\leq
n^{\Theta(1)}.
\end{align}
Setting $T = \Theta\bigl(\ln(n/\eta)\bigr)$ suffices to ensure $\Delta_T \leq \frac{\eta}{2}.$. Finally, since $\hat h = z_T$, we have $\|\hat h - h^*\| \leq \frac{\eta}{2}$ which then implies the following:
\begin{align}
\Bigl\|
\hat h - \bE_{\pi_{S,\tau}}\bigl[\nabla R_{S,\tau}\bigr]
\Bigr\|
=
\Bigl\|
\hat h - \bE_{\gamma_{\hat h}}\bigl[\nabla R_{S,\tau}\bigr]
\Bigr\| 
=
\|\hat h - F(\hat h)\| 
\leq
\|\hat h - h^*\| + \|F(\hat h) - F(h^*)\| 
\leq
\eta.
\end{align}
To analyze the parallel runtime and work, note that \cref{alg:centering-external-field} runs for $T$ iterations and each iteration performs $K$ \emph{parallel} calls to \cref{alg:rgd-ising-parallel} with error tolerance $\tfrac{\delta}{100KT}$. Hence, by Theorem \ref{thm:rgd-parallel-runtime}, \cref{alg:centering-external-field} has a parallel runtime of $O(\polylog(n, \eta^{-1}, \delta^{-1}))$ with $O\bigl(\poly(\tfrac{n}{\eta},\ln(\nicefrac{1}{\delta}))\bigr)$ work.


\section{Low Accuracy Parallel Sampler for the $p$-Spin Model}

In this section, we prove that Picard Algorithmic Stochastic Localization
(\cref{alg:picard-asl}) samples from the mixed $p$-spin model with
Wasserstein guarantees.

We first illustrate our results for the SK model, a special case of the $p$-spin model, because the
analysis is simpler and the calculations are closely related. Moreover, for the
SK model, we obtain a parallelization result for a broader range of temperatures
than the best previous parallel algorithm of \textcite{chen2025efficient}. 

The TAP-AMP mean-estimation algorithm for the SK model
(\cref{alg:mean-of-tilted-gibbs-measure-sk}) is slightly different from that for
the mixed $p$-spin model, but it uses the same AMP iterations followed by NGD
iterations.
We prove that the TAP-AMP algorithms proposed in prior work for computing mean
approximations in the SK and mixed $p$-spin models,
\cref{alg:mean-of-tilted-gibbs-measure-sk} and \cref{alg:tap-amp}, are
Lipschitz with respect to the external field and $\widetilde O(1)$ parallel computable. This enables us to apply Picard
iteration to the steps of algorithmic stochastic localization
\cref{eq:picard-asl}, yielding a parallel speedup.

For the rest of the section, we first present an improved parameter dependence
for algorithmic stochastic localization in
\cref{subsec:improved-parameter-dependency-of-asl}. Then, using the improved
parameters, we prove the parallelization result for algorithmic stochastic
localization in \cref{subsec:parallel-time-analysis-of-picard-asl}, relying on
the following two results:
\begin{itemize}
    \item 
    The TAP-AMP mean approximations are executable in $\widetilde O(1)$
    parallel time (\cref{subsec:tap-amp-in-polylog(n)-overview}).
    \item
    The TAP-AMP mean approximations are Lipschitz
    (\cref{subsec:lipschitz-tap-amp-of-sk}).
\end{itemize}

\subsection{Sign Rounding for Algorithmic Stochastic Localization}
\label{subsec:improved-parameter-dependency-of-asl}
\textcite{el2022sampling-sk,el2025sampling-p-spin} provide Wasserstein
guarantees for algorithmic stochastic localization using TAP fixed points as
mean approximations. Their analysis controls the propagation of the
mean-estimation error along the discretized stochastic localization trajectory. They obtain
an $\varepsilon$ upper bound in normalized Wasserstein distance between the output of the sampling
algorithm and the target distribution by using $\exp(\poly(1/\varepsilon))$ discrete
steps to run the stochastic localization up to a time horizon $\poly(1/\varepsilon)$.

\begin{proposition}\label{prop:andrea-runtime-dependency}
    For $\varepsilon_n < \varepsilon $, \textcite{el2022sampling-sk, el2025sampling-p-spin} provide algorithmic stochastic localization for a time horizon $\poly(1/\varepsilon)$ with $\exp(\poly(1/\varepsilon))$ discretization steps. 
\end{proposition}
\begin{proof}
    We defer the proof to the appendix, see \cref{sec:andrea-runtime}.
\end{proof}

In this work, in addition to the Picard parallelization and Lipschitz analysis, we improve the
dependence of the discretization step size and the time horizon in stochastic localization on the parameter $\varepsilon$ by analyzing sign rounding for algorithmic stochastic localization.

\begin{theorem}
\label{thm:optimized-algorithmic-stochastic-localization}
For any $n$, there exists a value $\varepsilon_n$ such that, for any
$\varepsilon > \varepsilon_n$ and any high-temperature mixed $p$-spin model with Hamiltonian $H_n$ and no external field, there is an algorithm based on algorithmic stochastic
localization using TAP-AMP mean approximations
(\cref{alg:mean-of-tilted-gibbs-measure-sk} and \cref{alg:tap-amp}) whose output
is within $\varepsilon$ of the target distribution in normalized Wasserstein
distance with probability $1 - o(1)$ over the disorder. 

Moreover, the algorithmic stochastic localization procedure uses a discretization step of size $\delta$ and time horizon $t$ satisfying
\[
    \delta = \poly(\varepsilon),
    \qquad
    t = {O}(\log(1/\varepsilon)).
\]
The output of the algorithm is generated by sign rounding, namely
$\sign{\hat{{y}}_t}$, where $\hat{{y}}_t$ is the
output of the discrete algorithmic stochastic localization process at time
$t$. Furthermore, if ${y}_t$ denotes the true stochastic
localization process, then
\[
    W_{2,n}({y}_t,\hat{{y}}_t) \leq O(\varepsilon).
\]
\end{theorem}

\begin{proof}
We sketch the proof here; see \cref{app1:number-of-discrete-steps} for the
full proof. Let ${y}_t$ denote the true stochastic localization
process at time $t$. Then ${y}_t$ satisfies the distributional
identity
\[
    {y}_t/t \sim x^* + B_t/t,
\]
where $x^* \sim \mu$ and $B_t/t \sim \mathcal{N}(0, I/t)$ due to the characterization by \cref{thm:sl-marginal}.

Consider the sign-rounding procedure applied to ${y}_t$. Due to \cref{lem:sl-sign-rounding}, we have 
\[
    W_{2,n}(\sign{{y}_t}, \mu) \leq O(\exp(-t/2)).
\] Therefore, choosing
$t = O(\log(1/\varepsilon))$ makes the rounding error of the true stochastic
localization trajectory at most $\poly(\varepsilon)$.

Using the error-propagation analysis of
\cite{el2022sampling-sk,el2025sampling-p-spin}, one can show that the
approximate trajectory $\hat{{y}}_t$, obtained using TAP fixed
points as mean approximations, satisfies
\[
    W_{2,n}({y}_t,\hat{{y}}_t)
    \leq \varepsilon,
\]
with discretization step size polynomial in $\varepsilon$ and time horizon
$t = {O}(\log(1/\varepsilon))$. Finally, applying the stability
of sign rounding, as in \cref{lem:stability-sl-rounding}, transfers this
trajectory-level approximation to the desired normalized Wasserstein guarantee
for the output distribution.
\end{proof}

\begin{remark}
\label{remark:best-error-dependency-asl-with-tap-amp}
In \cite{el2022sampling-sk, el2025sampling-p-spin}, the authors prove that there exists a sequence $\varepsilon_n \to 0$ such that Algorithmic Stochastic Localization outputs a distribution whose normalized Wasserstein distance from the target measure is at most $\varepsilon_n$. Their guarantee and analysis are asymptotic, and the precise non-asymptotic dependence of $\varepsilon_n$ on $n$ is left implicit.

They show that the error accumulated by the TAP-AMP mean approximation up to a certain time horizon can be controlled. Combining this propagated-error bound with \cref{lem:stability-sl-rounding} yields a normalized Wasserstein guarantee.

For this work, let $\varepsilon_n$ be the smallest value such that
\[
    W_{2,n}\!\left(
        \hat{y}_{\log(1/\varepsilon_n)},
        y_{\log(1/\varepsilon_n)}
    \right)
    \leq \varepsilon_n,
\]
where $\hat{y}_t$ is generated by the implementation of Algorithmic Stochastic Localization using TAP-AMP mean estimates, with discretization step size $\delta=\operatorname{poly}(\varepsilon_n)$. Also, take $T_n = O(\log(\frac{1}{\varepsilon_n}))$ as the time horizon where the error between the TAP fixed point and the true mean can be upper bounded by $\poly(\varepsilon_n)$.
\end{remark}

\subsection{Parallel Time Analysis of \cref{alg:picard-asl}}
\label{subsec:parallel-time-analysis-of-picard-asl}
In this section, we analyze the parallelization result for algorithmic stochastic
localization using Picard iteration. The Picard iterations produce a trajectory
that is close, in Wasserstein distance, to the true solution of the discrete
stochastic localization process using the TAP-AMP mean approximation. This
discrete process is itself close to the true continuous-time stochastic
localization process by
\cref{thm:optimized-algorithmic-stochastic-localization}. Therefore, by the
triangle inequality for Wasserstein distance, the input to the rounding step is
close to the output of the true stochastic localization process.

However, in our proposed algorithm (\cref{alg:picard-asl}), the final output is
obtained by applying the sign function. It is not immediate that small
Wasserstein distance before rounding implies small Wasserstein distance after
rounding. Indeed, if two coordinates are both close to zero but have opposite
signs, applying the sign function can amplify their discrepancy. To address this issue, we use the stability result for sign rounding in
\cref{lem:stability-sl-rounding}. This lemma shows that once the Picard
trajectory is sufficiently close to the true stochastic localization trajectory,
the corresponding sign-rounded outputs remain close in normalized Wasserstein
distance.
 
We present our results for the SK model and the mixed $p$-spin model separately. 
For the SK model, the Gibbs measure is defined by
\[
\mu_A(x) \propto \exp\left(\frac{\beta}{2}\langle x,Ax\rangle\right),
\]
where $A \sim \mathrm{GOE}(n)$. Here $\mathrm{GOE}(n)$ denotes the distribution over symmetric matrices whose entries are Gaussian with variance of order $1/n$.
Although the proof structure is nearly the same, we separate the two cases
because, for the SK model, we obtain a parallel sampling algorithm in a
temperature regime not covered by previous parallel algorithms. 
The analysis for
the SK model yields the threshold $\beta_0 \approx 0.3753$ for Lipschitzness and
parallelization, improving on the $\beta < 1/4$ threshold at which the previous
$\polylog(n)$-depth algorithm for sampling from the SK model applies
\cite{chen2025efficient}. However, unlike the algorithm of
\textcite{chen2025efficient}, our low-accuracy algorithm cannot achieve
arbitrarily small error.

\begin{theorem}[SK Model Parallelization]
\label{thm:sk-model-picard-asl}
For any $\varepsilon \geq \varepsilon_n$ and inverse temperature
$\beta < \beta_0 \approx 0.3753$,
there exist parameters (\cref{thm:optimized-algorithmic-stochastic-localization})
\[
\eta, K_{\mathrm{AMP}}, K_{\mathrm{NGD}} = \polylog(n/\varepsilon),\qquad K = \poly(1/\varepsilon),
\qquad
\delta = \poly(\varepsilon), \qquad t = K\delta = O(\log(1/\varepsilon))
\]
such that the following holds. Let
$R = {O}(\frac{\log^2(1/\varepsilon)}{\beta_0 - \beta})$, and execute
\cref{alg:picard-asl} on a Hamiltonian of the form
\[
H_n(x) = x^\top {A}x
\]
using TAP-AMP (\cref{alg:mean-of-tilted-gibbs-measure-sk}) as the approximate mean function $\hat{m}$ with parameters
$(\eta, K_{\mathrm{AMP}}, K_{\mathrm{NGD}}, K, \delta)$.

Then \cref{alg:picard-asl} outputs a random point
$\mathbf{x}^{\mathrm{alg}} \in \{-1,+1\}^n$ with law
$\mu_{\mathbf{A}}^{\mathrm{alg}}$ such that, with probability $1-o(1)$ over
$\mathbf{A} \sim \mathrm{GOE}(n)$,
\[
W_{2,n}(\mu_{\mathbf{A}}^{\mathrm{alg}}, \mu_{\mathbf{A}})
\leq
O(\varepsilon).
\label{eq:wasserstein-error-of-p-spin}
\]
The parallel runtime of the algorithm is
\[
\poly(\log(n/\varepsilon)),
\]
and its total work is $\poly(n/\varepsilon)$.
\end{theorem}

\begin{proof}
First, we argue that with the suitable choice of parameters, the output of
\cref{alg:picard-asl} is close in Wasserstein distance to the
output of the sequential algorithmic stochastic localization process \cref{eq:sl-euler}. Then we prove that with the parameter choices, the algorithm runs in
$\poly(\log(n/\varepsilon))$ parallel time.

Since the Brownian motions are sampled and fixed throughout the Picard
iteration, \cref{alg:picard-asl} can be viewed as a 
differential equation of the form in \cref{eq:discrete-diff}, where the drift
function is the mean-computation function $\hat{m}$.

By \cref{cor:sk-model-lipschitzness-with-high-temperature}, the approximate
mean TAP-AMP, i.e. $\hat{m}$, given by \cref{alg:mean-of-tilted-gibbs-measure-sk}, is
$O(1/(\beta_0-\beta))$-Lipschitz with respect to the tilt, with high probability over the disorder. Hence, we can apply
the Picard convergence theorem \cref{thm:picard-convergence} to obtain fast
convergence of the Picard iteration.

The Picard iteration is initialized with
$\hat{{y}}_{i\delta}^{(0)}=\boldsymbol{0}$ for
$0 \leq i \leq K$. Let
\[
M =
\max_{0 \leq i \leq K}
\left\|
\hat{{y}}_{i\delta}^*
-
\hat{{y}}_{i\delta}^{(1)}
\right\|,
\]
where $\hat{{y}}_{i\delta}^*$ denotes the true sequential
solution of the discrete stochastic localization equation at time $i\delta$.
After one Picard iteration,
\[
\left\|
\hat{{y}}_{i\delta}^{(1)}
-
\hat{{y}}_{i\delta}^*
\right\|
=
\left\|
\sum_{\ell=0}^{i-1}
\delta
\Bigl(
\hat{{m}}(H_n,{0})
-
\hat{{m}}(H_n,\hat{{y}}_{\ell\delta}^*)
\Bigr)
\right\|
\leq
\delta i \sqrt{n}.
\]
Since $i\delta \leq K\delta = t =
{O}(\log(1/\varepsilon))$, this gives the uniform bound
\[
M
\leq
{O}\bigl(\sqrt{n}\log(1/\varepsilon)\bigr).
\]

By the exponential convergence of the Picard iterations
\cref{thm:picard-convergence}, we have
\[
\left\|
\hat{{y}}_{K\delta}^{(R)}
-
\hat{{y}}_{K\delta}^*
\right\|
\leq
\sqrt{n}
\frac{
\left(\frac{K\delta}{\beta_0-\beta}\right)^R
}{R!}.
\]
Since $t=K\delta={O}(\log(1/\varepsilon))$, taking
\[
R \geq
\Omega\left(
\frac{K\delta \log(1/\varepsilon)}{\beta_0-\beta}
\right)
=
{\Omega}\left(
\frac{\log^2(1/\varepsilon)}{\beta_0-\beta}
\right)
\]
implies
\[
\left\|
\hat{{y}}_{K\delta}^{(R)}
-
\hat{{y}}_{K\delta}^*
\right\|
\leq
\sqrt{n}\,\poly(\varepsilon).
\]
Consequently,
\[
W_{2,n}
\left(
\hat{{y}}_{K\delta}^{(R)}/t,
\hat{{y}}_{K\delta}^*/t
\right)
\leq
\poly(\varepsilon).
\]

We now use \cref{lem:stability-sl-rounding} to bound the normalized
$\ell_2$ Wasserstein distance between the rounded output of the algorithm and
the target distribution. Let ${y}_t$ denote the true continuous-time
stochastic localization process at time $t=K\delta$, and let
$\mu_{{A}}$ denote the target distribution. By the triangle
inequality,
\[
W_{2,n}
\left(
\hat{{y}}_{K\delta}^{(R)}/t,
{y}_t/t
\right)
\leq
W_{2,n}
\left(
\hat{{y}}_{K\delta}^{(R)}/t,
\hat{{y}}_{K\delta}^*/t
\right)
+
W_{2,n}
\left(
\hat{{y}}_{K\delta}^*/t,
{y}_t/t
\right)
\leq
\varepsilon+\poly(\varepsilon).
\]
Here,
$W_{2,n}
\left(
\hat{{y}}_{K\delta}^*/t,
{y}_t/t
\right) \leq \varepsilon$ holds with probability $1 - o(1)$ according to \cref{thm:optimized-algorithmic-stochastic-localization}.
Applying \cref{lem:stability-sl-rounding} with
\[
\Delta =
W_{2,n}
\left(
\hat{{y}}_{K\delta}^{(R)}/t,
{y}_t/t
\right)
\leq
\varepsilon+\poly(\varepsilon),
\]
we obtain
\[
W_{2,n}
\bigl(
\sign{\hat{{y}}_{K\delta}^{(R)}},
\mu_{\boldsymbol{A}}
\bigr)
\leq
4(\varepsilon+\poly(\varepsilon)) + e^{-\Omega(t)}
=
4\varepsilon+\poly(\varepsilon) = O(\varepsilon).
\]

Thus, the normalized Wasserstein distance between the output of the parallel
algorithm and the target distribution is bounded by
$O(\varepsilon)$.
The stated parallel runtime and total work follow from the parameter choices
above and the runtime of the parallel Picard implementation, given the $\poly(\log(n/\varepsilon))$ parallel runtime for the TAP-AMP mean-approximation algorithm (\cref{prop:log-n-evaluation-of-mean}).
\end{proof}

\begin{theorem}[$p$-spin Model Parallelization]
\label{thm:p-spin-model-picard-asl}
For any $\varepsilon \geq \varepsilon_n$ and temperature coefficients
$\{\beta_p\}_{p \geq 2}$ satisfying
\[
    \mathfrak{C}(\beta) \coloneq \sum_{p=2}^{P} \beta_p \sqrt{p^3 \ln(p)}  < \gamma_0 \qquad \mathfrak{D}(\beta) = \sum_{p=2}^{P} \beta_p \sqrt{2^p p^3 \ln(p)} < \infty
\]
there exist parameters (\cref{thm:optimized-algorithmic-stochastic-localization})
\[
\eta, K_{\mathrm{AMP}}, K_{\mathrm{NGD}}, K = \poly(1/\varepsilon),
\qquad
\delta = \poly(\varepsilon),
\qquad
t = K\delta = {O}(\log(1/\varepsilon)),
\]
such that the following holds. Let
$
R =
{O}_{\mathfrak{C}(\beta)}\bigl(\log^2(1/\varepsilon)\bigr)
$
be the number of Picard iterations. Given query access to the $p$-spin Hamiltonian as defined in \cref{eq:p-spin-model}, \cref{alg:picard-asl}, run with parameters
$(\eta, K_{\mathrm{AMP}}, K_{\mathrm{NGD}}, K, \delta)$, outputs a random point
$\mathbf{x}^{\mathrm{alg}} \in \{-1,+1\}^n$ with law
$\mu_{H_n}^{\mathrm{alg}}$ such that, with probability $1-o(1)$ over the
Gaussian coefficients $\{g_J\}_{J \subseteq [n]}$,
\[
W_{2,n}(\mu_{H_n}^{\mathrm{alg}}, \mu_{H_n})
\leq
O(\varepsilon).
\]

Here, $\varepsilon$ is the error incurred by the sequential algorithm in
\cref{thm:optimized-algorithmic-stochastic-localization}, while
$\poly(\varepsilon)$ is the additional error incurred by replacing the
sequential updates with Picard iteration.

The parallel runtime of the algorithm is $\poly(\log(n/\varepsilon))$ and its total work is $\poly(n,1/\varepsilon)$.
\end{theorem}

\begin{proof}
The proof follows the same argument as the SK case
\cref{thm:sk-model-picard-asl}; we give only the brief sketch.

Fix the Brownian increments used by
\cref{alg:picard-asl}. With these increments fixed, the algorithm
can be viewed as a discrete differential equation whose drift is given by the
approximate mean map $\hat{m}$. Since the temperature condition holds,
\cref{cor:mixed-p-spin-wp-lipschitz} implies that the TAP-AMP mean approximation 
is Lipschitz with respect to the external field with high probability over the disorder.
Then
\cref{thm:picard-convergence}, with the same initialization as in the SK case,
implies that the Picard trajectory contracts to the discrete stochastic
localization trajectory up to error $\sqrt{n}\,\poly(\varepsilon)$. Hence,
after normalization,
\[
W_{2,n}(\hat{y}_{K\delta}^{(R)}/t,\hat{y}_{K\delta}^*/t)
\leq \poly(\varepsilon).
\]

The sequential stochastic localization guarantee from
\cref{thm:optimized-algorithmic-stochastic-localization} contributes the
$O(\varepsilon)$ error. Combining this with the Picard error, the triangle
inequality gives
\[
W_{2,n}(\hat{y}_{K\delta}^{(R)}/t,{y}_t/t)
\leq W_{2,n}(\hat{y}_{t}^{(R)}/t,\hat{y}^{*}_t/t) + W_{2,n}(\hat{y}_{t}^{*}/t,{y}_t/t)
 \leq O(\varepsilon)+\poly(\varepsilon).
\]
Finally, applying the stability rounding lemma (\cref{lem:stability-sl-rounding}), we have 
\[
W_{2,n}(\mu_{H_n}^{\mathrm{alg}},\mu_{H_n})
\leq O(\varepsilon)+\poly(\varepsilon) = O(\varepsilon).
\]

The stated parallel runtime and total work follow from the parameter choices
above and the runtime of the parallel Picard implementation, given the $\poly(\log(n/\varepsilon))$ parallel runtime for the TAP-AMP mean-approximation algorithm (\cref{prop:log-n-evaluation-of-mean}).
\end{proof}

\subsection{TAP-AMP in $\polylog(n)$ Parallel Time} \label{subsec:tap-amp-in-polylog(n)-overview}

The TAP-AMP mean-approximations can be computed in
$\polylog(n/\varepsilon)$ parallel depth.
The TAP-AMP algorithms
\cref{alg:mean-of-tilted-gibbs-measure-sk} and \cref{alg:tap-amp} consist of
AMP iterations followed by NGD iterations. Each individual AMP or NGD iteration
is computable in parallel with depth $O(\log n)$; thus, it suffices to prove
$\widetilde O(1)$ bounds on the numbers of AMP and NGD iterations. Since this
proof is rather involved and does not affect the rest of the argument, we give
only a brief sketch here and defer the details to the appendix; see
\cref{sec:mean-approximate-parallel-depth}.

\begin{proposition}
\label{prop:log-n-evaluation-of-mean}
The mean-approximation algorithms
\cref{alg:mean-of-tilted-gibbs-measure-sk} and \cref{alg:tap-amp} are
computable with parallel depth $\polylog(n/\varepsilon)$, given query access to
the mixed $p$-spin Hamiltonian and the parameter choices specified in
\cref{thm:optimized-algorithmic-stochastic-localization} for attaining
normalized Wasserstein error $\varepsilon$.
\end{proposition}

\begin{proof}
We give a proof sketch here and defer the full proof to
\cref{sec:mean-approximate-parallel-depth}. The mean-approximation algorithms
consist of two phases: AMP iterations followed by NGD iterations. Each AMP
iteration is computable in $\widetilde O(1)$ parallel depth. Thus, the parallel
depth of the AMP phase is controlled by the parameter $K_{\mathrm{AMP}}$, for
which we give an explicit bound in the full proof.

For the NGD phase, given gradient-query access to the Hamiltonian, the parallel
depth is controlled by the number of NGD iterations. In the full proof, we show
that this number is at most $\polylog(n)$.
\end{proof}

\subsection{Lipschitz Property of the TAP-AMP Mean Approximation for the SK Model} \label{subsec:lipschitz-tap-amp-of-sk}

In this section, we prove that, for the SK model at high temperature
$\beta < \beta_0 \approx 0.3753$, the TAP-AMP algorithm
(\cref{alg:mean-of-tilted-gibbs-measure-sk}) for mean approximation is
$O\!\left(\frac{1}{\beta_0-\beta}\right)$-Lipschitz with respect to the tilt
parameter, or external field, with high probability.
We provide an analogous result for the mixed $p$-spin model in
\cref{subsec:mixed-p-spin-lipschitz}.

\begin{Algorithm}[H]
\caption{{Mean of the Tilted Ising Measure \cite{el2022sampling-sk}}}
\label{alg:mean-of-tilted-gibbs-measure-sk}
\textbf{Input:} Data $ {A} \in \mathbb{R}^{n \times n}$, ${y} \in \mathbb{R}^n$, parameters $\beta, \eta > 0$, $q \in (0, 1)$, iteration numbers $K_{\text{AMP}}, K_{\text{NGD}}$\;
${{m}}_{-1} = {z}_0 = 0$\;
\For{$k = 0, \dots, K_{\text{AMP}} - 1$}{
    ${{m}}_k = \tanh({z}_k), \quad {b}_k = \frac{\beta^2}{n} \sum_{i=1}^n (1 - \tanh^2((z_k)_i))$\;
    ${z}_{k+1} = \beta {A} {{m}}_k + {y} - {b}_k {{m}}_{k-1}$\;
}
${u}^0 = {z}_{K_{\text{AMP}}}$\;
\For{$k = 0, \dots, K_{\text{NGD}} - 1$}{
    ${u}_{k+1} = {u}_k - \eta \cdot \nabla \hat{\mathscr{F}}_{\text{TAP}}({{m}}^{+}_k; {y}, q)$\;
    ${{m}}^{+}_{k + 1} = \tanh({u}_{k+1})$\;
}
\textbf{return} $ {{m}}^{+}_{K_{\text{NGD}}}$\;
\end{Algorithm}


\cref{alg:mean-of-tilted-gibbs-measure-sk} consists of two phases: it first uses AMP iterations attempting to get close to the fixed point of the TAP iteration, and then it uses an NGD algorithm to close the gap to the fixed point as much as needed.
Here, we will take a closer look at the two phases of \cref{alg:mean-of-tilted-gibbs-measure-sk}.

\paragraph{Phase 1: AMP phase.}\label{alg:tilted-mean-amp-phase}
In the first phase, initialize
\[
m_{-1}=0,\qquad z_0=0,
\]
and for $k=0,1,\dots,K_{\mathrm{AMP}}-1$ define
\begin{align}
m_k &= \tanh(z_k), \label{eq:mk_def}\\
b_k &= \frac{\beta^2}{n}\sum_{i\in[n]}\sech^{2}\!\bigl((z_k)_i\bigr), \label{eq:bk_def}\\
z_{k+1} &= \beta A m_k + y - b_k\, m_{k-1}. \label{eq:amp_update}
\end{align}

\paragraph{Phase 2: NGD phase on TAP free energy.}\label{alg:tilted-mean-ngd-phase}

Initialize $u_0 := z_{K_{\mathrm{AMP}}}$.
For $t=0,1,\dots,K_{\mathrm{NGD}}-1$ define
\begin{align}
m_t^+ &= \tanh(u_t), \label{eq:mplus_def}\\
u_{t+1} &= u_t - \eta \,\nabla_m \hat{\mathscr{F}}_{\mathrm{TAP}}(m_t^+;y,q). \label{eq:ngd_update}
\end{align}
The algorithm outputs
\[
\hat m(A,y) := m_{K_{\mathrm{NGD}}}^+ = \tanh(u_{K_{\mathrm{NGD}}}).
\]

\begin{definition}\label{def:tilted-mean-computation}
    For a given external field $y \in \R^n$ and interaction matrix $A \in \R^{n \times n}$, let $\hat m_{(K_{\text{AMP}}, K_{\text{NGD}})}(A, y)$ be the output of the algorithm after running the AMP iteration for $K_{\text{AMP}}$ steps and then running NGD for $K_{\text{NGD}}$ steps. 
    We refer to $\hat m(y)$ as the tilted-mean computation for the external field $y$ whenever $A$, $K_{\text{AMP}}$, and $K_{\text{NGD}}$ are clear from the context.
\end{definition}

Our goal in this section is to find the temperature $\beta(A) = \beta(\opnorm{A})$ such that, for any two tilts $y$ and $\widetilde y$, the TAP fixed point of the Ising model with interaction matrix $A$ and temperature $\beta$ satisfies
\[
\|\hat{m}_{(K_{\mathrm{AMP}}, K_{\mathrm{NGD}})}(A, y) - \hat{m}_{(K_{\mathrm{AMP}}, K_{\mathrm{NGD}})}(A, \widetilde y)\| \leq \frac{\gamma({A})}{\beta({A}) - \beta}\norm{y-\widetilde y}.
\]
Note that if the matrix $A$ is sampled from the $\mathrm{GOE}(n)$ distribution, then its maximum eigenvalue concentrates around $2$; thus, if we have inverse temperature $\beta$ satisfying $\gamma(2,\beta)< 1$, then we have the Lipschitz property for the TAP-AMP mean approximation $\hat{m}$ with high probability. 

To prove the Lipschitz property of the mean-approximation algorithm (\cref{alg:mean-of-tilted-gibbs-measure-sk}), we will use the Lipschitz property of the following two functions:
\begin{itemize}
    \item $\tanh$ is $1$-Lipschitz (\cref{lem:tanh_lip}).
    \item $\sech^{2}$ is $\frac{4}{3\sqrt{3}}$-Lipschitz (\cref{lem:sech2_lip}).
\end{itemize}

We will prove that the AMP and NGD iterations are both Lipschitz. Throughout the rest of this section, we fix two external fields $y$ and $\widetilde y$ and use the following notation, where the unadorned variables correspond to the computation of $\hat{m}_{(K_{\text{AMP}}, K_{\text{NGD}})}(A, y)$ and the tilded variables correspond to the computation of $\hat{m}_{(K_{\text{AMP}}, K_{\text{NGD}})}(A, \widetilde y)$:
\begin{itemize}
    \item $\Delta y := y-\widetilde y$ (tilt difference)
    \item $\Delta z_k := z_k-\widetilde z_k$ (AMP iterations)
    \item $\Delta m_k := m_k-\widetilde m_k$ (AMP iterations)
    \item $\Delta u_k := u_k - \widetilde u_k$ (NGD iterations)
\end{itemize}

Also, $\Delta z_0=\Delta z_{-1}=0$ and $m_{-1}=\widetilde m_{-1}=0$ for the initialization.

\begin{theorem}[Lipschitzness of AMP phase]\label{thm:amp_lip}
Assume $\gamma(A, \beta) :=  \beta\opnorm{A}+(1 + \frac{4}{3 \sqrt{3}})\beta^2 < 1$.
Then for all $k \ge 0$,
\[
\norm{\Delta z_k}\le \frac{1}{1-\gamma}\,\norm{\Delta y},
\qquad
\norm{\Delta m_k}\le \frac{1}{1-\gamma}\,\norm{\Delta y}.
\]
\end{theorem}

\begin{proof}
Let $V_k:=\max_{0\le j\le k}\norm{\Delta z_j}$.
We will prove the following two inequalities for each iteration:
\begin{itemize}
    \item $\norm{\Delta z_{k+1}}
\le (\beta\opnorm{A}+\frac{4\beta^2}{3\sqrt{3}})\norm{\Delta z_k}
+ \beta^2 \norm{\Delta z_{k-1}}
+ \norm{\Delta y}$
    \item $\norm{\Delta m_k} \leq \norm{\Delta z_k}$
\end{itemize}
The second item follows immediately from the $1$-Lipschitz property in \cref{lem:tanh_lip}. Thus, we focus on proving the first item. By taking the difference of \cref{eq:amp_update} for the two tilts $y$ and $\widetilde y$, we have the following identity:
\begin{align}
\norm{\Delta z_{k+1}}
= & \ \norm{\beta A \Delta m_k + \Delta y - b_k\Delta m_{k-1} - (b_k-\widetilde b_k)\widetilde m_{k-1}} \notag
\\ 
\label{eq:after-triangle-inequality}
\leq & \beta \opnorm{A} \norm{\Delta m_k} + \norm{\Delta y} + |{b_k}| \cdot \norm{\Delta m_{k - 1}} + |{b_k - \widetilde b_k}| \cdot  \norm{\widetilde m_{k-1}}
\end{align}
Now, we upper bound each term in \cref{eq:after-triangle-inequality}. Since $\tanh$ is $1$-Lipschitz, we have $\norm{\Delta m_k} \leq \norm{\Delta z_k} \leq V_k$ and $\norm{\Delta m_{k - 1}} \leq \norm{\Delta z_{k - 1}} \leq V_k$. Also, $0\leq b_k\leq \beta^2$.
The remaining term is $|b_k-\widetilde b_k| \cdot \norm{\widetilde m_{k - 1}}$. We upper bound the absolute value of the scalar $b_k - \widetilde b_k$ as follows:
\begin{align}
|b_k-\widetilde b_k| = & \ 
\big|\frac{\beta^2}{n}\sum_{i=1}^n \left(\sech^{2}\!\bigl((z_k)_i\bigr)-\sech^{2}\!\bigl((\widetilde z_k)_i\bigr)\right)\big| \notag \\
\le & \ \frac{\beta^2}{n}\sum_{i=1}^n \big|\sech^{2}\!\bigl((z_k)_i\bigr)-\sech^{2}\!\bigl((\widetilde z_k)_i\bigr)\big| \tag{Triangle inequality}
\\ 
\le & \ \frac{\beta^2}{n}\sum_{i=1}^n \frac{4}{3\sqrt{3}}|(z_k)_i-(\widetilde z_k)_i|  \tag{$\sech^{2}$ is $\frac{4}{3\sqrt{3}}$-Lipschitz}
\\
\le & 
\ \frac{4}{3\sqrt{3}}\cdot\frac{\beta^2}{\sqrt{n}}\norm{\Delta z_k}, \tag{Cauchy--Schwarz}
\end{align}

Since $\norm{\widetilde m_{k-1}}\leq \sqrt n$, this implies
$|{b_k - \widetilde b_k}| \cdot  \norm{\widetilde m_{k-1}} \leq \beta^2 \frac{4}{3\sqrt{3}}\norm{\Delta z_k} \leq \beta^2 \frac{4}{3\sqrt{3}}V_k$.
Using the inequalities for the terms in \cref{eq:after-triangle-inequality}, we have
\begin{align}
    \norm{\Delta z_{k+1}} \leq & \ V_k\beta \opnorm{A} + \norm{\Delta y} + \beta^2 V_k + \beta^2\frac{4}{3\sqrt{3}}V_k \notag \\
     = & \ \gamma(A, \beta) V_k + \norm{\Delta y} \notag
\end{align}
Hence, $V_{k+1}\le \max\{V_k,\gamma V_k+\norm{\Delta y}\}$. As long as $\gamma(A, \beta) < 1$, the uniform bound $\norm{\Delta y}/(1-\gamma)$ follows by a simple induction. 

\end{proof}

\subsubsection*{Lipschitzness of the full AMP+NGD algorithm for the SK model}

We have proved so far that after an arbitrary number of iterations, $\norm{\Delta m_k}$ is upper bounded by $\frac{1}{1 - \gamma(A, \beta)}\norm{\Delta y}$. Now, we prove that applying NGD after AMP preserves the same bound for the difference of the mean estimates. 
For convenience, and to avoid overloading the notation $m_k$ from the AMP phase, we analyze the NGD step in the variables $u_i$ defined by
$m=\tanh(u)$ (equivalently $u=\atanh(m)$ component-wise).

\paragraph{Rewriting the gradient of TAP in $u$-space.}
The TAP equation is slightly different for the SK model than for the mixed $p$-spin model.
For a prescribed overlap parameter $q$, the SK TAP free energy satisfies
\begin{align}    
    \nabla \mathscr{F}_{\TAP}(m;y,q)
    =
    -\beta A m-y+\arctanh(m)+\beta^2(1-q)m. \label{eq:tap-for-sk}
\end{align}

Plugging $m=\tanh(u)$ into~\eqref{eq:tap-for-sk} gives
\begin{align*}
\nabla_m \hat{\mathscr{F}}_{\mathrm{TAP}}(\tanh(u);y,q)
&= -\beta A\tanh(u) - y + u + \beta^2(1-q)\tanh(u)\nonumber\\
&= -\Big(\beta A - \beta^2(1-q)I\Big)\tanh(u) - y + u. \label{eq:grad_u_form}
\end{align*}
Define
\[
B := \beta A - \beta^2(1-q)I.
\]
Then the NGD update~\eqref{eq:ngd_update} becomes
\begin{equation}
\label{eq:u_update}
u_{t+1} = (1-\eta)u_t + \eta\, B\tanh(u_t) + \eta\, y.
\tag{NGD update}
\end{equation}

\begin{theorem}[TAP Lipschitz property]\label{thm:full_lip}
Assume $\gamma(A, \beta) = \beta\opnorm{A}+(1 + \frac{4}{3\sqrt{3}})\beta^2<1$ and $q\in(0,1)$.
Let $0<\eta\leq 1$.
Set $u_0 = z_{K_{\mathrm{AMP}}}$ and $\widetilde u_0 = \widetilde z_{K_{\mathrm{AMP}}}$. Then for every $t\ge 0$,
\[
\norm{u_t(y)-u_t(\widetilde y)} \le \frac{1}{1-\gamma(A, \beta)}\,\norm{y-\widetilde y}.
\]
Consequently, the final output satisfies
\[
\norm{\hat m(A,y)-\hat m(A,\widetilde y)}
\le \frac{1}{1-\gamma(A, \beta)}\,\norm{y-\widetilde y}.
\]
\end{theorem}

\begin{proof}
Recall that $\Delta u_t := u_t(y)-u_t(\widetilde y)$.
By \cref{thm:amp_lip} and initialization $u_0=z_{K_{\mathrm{AMP}}}$,
\[
\norm{\Delta u_0}=\norm{\Delta z_{K_{\mathrm{AMP}}}}\le \frac{1}{1-\gamma}\norm{\Delta y}.
\]

Subtract~\eqref{eq:u_update} for $y$ and $\widetilde y$:
\begin{align}    
\norm{\Delta u_{t+1}} &= \norm{(1-\eta)\Delta u_t + \eta\, B\big(\tanh(u_t)-\tanh(\widetilde u_t)\big) + \eta\,\Delta y} \notag \\
&\leq (1 - \eta)\norm{\Delta u_t} + \eta \opnorm{B} \norm{\tanh(u_t) - \tanh(\widetilde u_t)} + \eta \norm{\Delta y} \tag{triangle inequality} \\
&\leq (1 - \eta)\norm{\Delta u_t} + \eta \opnorm{B} \norm{\Delta u_t} + \eta \norm{\Delta y} \tag{$\tanh$ is Lipschitz}
\end{align}

Now, it is enough to prove that if $\norm{\Delta u_t} \leq \frac{\norm{\Delta y}}{1 - \gamma(A, \beta)}$, then $\norm{\Delta u_{t + 1}}$ satisfies the same upper bound. Since $B = \beta A - \beta^2(1 - q)I$, we have $\opnorm{B} \leq \beta \opnorm{A} + \beta^2(1-q) \leq \beta \opnorm{A} + \beta^2$ since $q \in (0, 1)$. Thus, we need to prove that
\[
\left(1 - \eta + \eta(\beta \opnorm{A} + \beta^2)\right) \frac{1}{1 - \gamma(A, \beta)} + \eta \leq \frac{1}{1 - \gamma(A, \beta)}
\iff \eta( \beta \opnorm{A} + \beta^2 - \gamma(A, \beta)) \leq 0.
\]

Therefore, the inequality holds because $\gamma(A, \beta) = \beta \opnorm{A} + (1 + \frac{4}{3\sqrt{3}})\beta^2 \geq \beta \opnorm{A} + \beta^2$. This concludes the Lipschitz property of \cref{alg:mean-of-tilted-gibbs-measure-sk} given $\gamma(A, \beta) < 1$.
\end{proof}

\begin{remark}
The key stability requirement for AMP and NGD iterations is
\[
\gamma(A, \beta) =\beta\opnorm{A}+(1 + \frac{4}{3\sqrt{3}})\beta^2<1.
\]
In random matrix settings (e.g.\ $A\sim\mathrm{GOE}(n)$), one typically has $\opnorm{A} = 2 + o(1)$ with high probability,
so any sufficiently small $\beta$ ensuring $\gamma(2, \beta) <1$ implies that the TAP-AMP mean approximation is $\frac{1}{1 - \gamma(2, \beta) - c}$-Lipschitz with high probability for every sufficiently small constant $c$.
\end{remark}

\begin{corollary}\label{cor:sk-model-lipschitzness-with-high-temperature}
    Fix $\varepsilon > 0$. If the SK inverse temperature $\beta$ is at most $\beta_0 - \varepsilon$, where $\beta_0 \approx 0.3753$ is the positive root of the equation 
    \[
\gamma(2, \beta) = 2\beta +(1 + \frac{4}{3\sqrt{3}})\beta^2 = 1,
    \]
    then the tilted-mean computation in \cref{alg:mean-of-tilted-gibbs-measure-sk} is $O(\frac{1}{\varepsilon})$-Lipschitz for the SK model with high probability.
\end{corollary}
\begin{proof}
    We know that $\opnorm{A}$ is concentrated around $2$, and the probability that it exceeds $2 + \varepsilon/{\beta_0}$ is $e^{-\Omega_\varepsilon(n)}$~\cite{anderson2010introduction}.
    Recall that $\beta_0$ is the positive solution of the equation $2x + (1 + \frac{4}{3\sqrt{3}}) x^2 = 1$. Thus, we only need to show that the parameter $\gamma = \beta\opnorm{A}+(1 + \frac{4}{3\sqrt{3}})\beta^2 \leq 1 - \Omega(\varepsilon)$ with high probability.

    \begin{align*}    
      \beta\opnorm{A}+(1 + \frac{4}{3\sqrt{3}})\beta^2
      &\leq  \beta(2 + \varepsilon/\beta_0) + (1 + \frac{4}{3\sqrt{3}})\beta^2 \\
      &\leq \varepsilon + 2\beta + (1+ \frac{4}{3\sqrt{3}})\beta^2 \\
      &\leq -\varepsilon + 2\beta_0 + (1 + \frac{4}{3\sqrt{3}})\beta_0^2 = 1 - \varepsilon 
    \end{align*}
    Since $\text{P}\{\opnorm{A} > 2 + \varepsilon/\beta_0\} \leq e^{-\Omega_\varepsilon(n)}$, the tilted-mean computation is $O(\frac{1}{\varepsilon})$-Lipschitz with high probability if $\beta \leq \beta_0 - \varepsilon$.
\end{proof}

\PrintBibliography

\appendix
\setcounter{section}{0}
\renewcommand{\thesection}{A.\arabic{section}}
\renewcommand{\theHsection}{appendix.\arabic{section}}


\section{Sign Rounding Analysis for Algorithmic Stochastic Localization}
\label{app1:number-of-discrete-steps}

In this appendix, we show the improved parameter dependence for Algorithmic Stochastic
Localization (ASL) using the notation of \cref{alg:picard-asl} and the
low-accuracy discussion in the introduction. 
We use
\[
        K \quad \text{for the number of localization steps},
        \qquad
        \delta \quad \text{for the discretization step size},
\]
so that the localization time horizon is $K\delta$.

The accuracy parameter $\varepsilon$ below is the target error in normalized
$2$-Wasserstein distance $W_{2,n}$ from \cref{def:normalized-wasserstein-l2}.
The threshold $\varepsilon_n=o_n(1)$ is the accuracy threshold of the
tilted-mean estimator used in the ASL analysis.

Let $\widehat m_{t}(y)$ be the TAP-AMP estimate of the tilted mean at
localization time $t$; the dependence on $t$ enters through the overlap
parameter $q_\star(t)$ in \cref{alg:tap-amp}.
The sequential ASL Euler trajectory is
\[
        \widehat y_{(k+1)\delta}
        =
        \widehat y_{k\delta}
        +
        \delta\,\widehat m_{k\delta}(\widehat y_{k\delta})
        +
        \bigl(B_{(k+1)\delta}-B_{k\delta}\bigr),
        \qquad
        \widehat y_0=0,
\]
or equivalently, with
$w_{k+1}=(B_{(k+1)\delta}-B_{k\delta})/\sqrt{\delta}
\sim \mathcal N(0,I_n)$,
\[
        \widehat y_{(k+1)\delta}
        =
        \widehat y_{k\delta}
        +
        \delta\,\widehat m_{G,k\delta}(\widehat y_{k\delta})
        +
        \sqrt{\delta}\,w_{k+1}.
\]

\begin{lemma}[Lemma 4.10 of \cite{el2022sampling-sk}]
\label{lem:sk-local-tap-stability}
Let $\beta < \beta_0$, let $c = O(1)$, and fix a time horizon
$T_{\SL} < t_n$ and $\varepsilon > \varepsilon_n$. Let $A\sim \text{GOE}(n)$, and let $y\in\R^n$. Let
$q_\star(\beta,t)$ be the AMP state-evolution overlap.

Then for $\rho \in \poly(\varepsilon)$ there exist
$K_{\AMP}=K_{\AMP}(\beta,T_{\SL},\varepsilon) \leq \polylog(n/\varepsilon)$ and
$K_{\NGD} \leq \polylog(n/\varepsilon)$ with the following
property.
With probability $1-o_n(1)$ over $(A,y(\cdot))$, simultaneously for all
$t\in(0,T_{\SL}]$ the stationary point (fixed point of TAP) satisfies:
\[
  \|m(A,y(t))-m_\star(A,y(t);t)\|
  \le
  \rho\sqrt{tn},
\]
\end{lemma}
\begin{proof}
    The proof is essentially the same as \cite{el2022sampling-sk}. For the choices of the parameters $K_\text{NGD}$ and $K_\text{AMP}$, see \cref{prop:log-n-evaluation-of-mean}.
\end{proof}
Define the normalized trajectory and mean-estimation errors
\[
        \Delta^y_k
        :=
        \frac{1}{\sqrt n}
        \left\|\widehat y_{k\delta}-y_{k\delta}\right\|,
        \qquad
        \Delta^m_k
        :=
        \frac{1}{\sqrt n}
        \left\|
            \widehat m_{G,k\delta}(\widehat y_{k\delta})
            -
            m_G(y_{k\delta})
        \right\| .
\]

\begin{lemma}[ASL error propagation, \cite{el2022sampling-sk}]
\label{lem:asl-error-propagation-current-notation}
Fix a time horizon $t < t_n$ (see remark \cref{remark:best-error-dependency-asl-with-tap-amp}). In the high-temperature regime covered by the ASL
analysis of \cite{el2022sampling-sk,el2025sampling-p-spin}, there exist
constants $C=O_\beta(1)$, depending only on the temperature parameters, and
a deterministic sequence $\xi_n\to 0$ such that the following holds with
probability $1-o_n(1)$ over the disorder. For every $k\geq 0$ satisfying
$k\delta\leq t$ and every $\delta\in(0,1)$,
\[
        \Delta^y_k
        \leq
        C e^{Ck\delta}k\delta
        \left(\rho\sqrt{k\delta}+\sqrt{\delta}\right)
        +
        \xi_n,
\]
and
\[
        \Delta^m_k
        \leq
        C e^{Ck\delta}k\delta
        \left(\rho\sqrt{k\delta}+\sqrt{\delta}\right)
        +
        C\rho\sqrt{k\delta}
        +
        \xi_n .
\]
Here $\rho = \poly(\varepsilon)$ denotes the accuracy parameter in the tilted-mean
estimation guarantee.
\end{lemma}

\begin{proof}
We imitate the proof of Lemma 4.14 of \cite{el2022sampling-sk} with new parameters. We couple the continuous stochastic-localization
trajectory and the discretized ASL trajectory using the same Brownian
increments. 
Throughout the proof, let $\xi_n$ denote a
deterministic non-negative sequence tending to zero. 

Set $t_k = k\delta$. The proof is by induction on $k$. The case $k=0$ is immediate since
$\widehat y_0=y_0=0$. Assume the claimed bounds hold up to time $t_k$.
Subtracting the two coupled recursions gives
\[
    \widehat y_{t_{k+1}}-y_{t_{k+1}}
    =
    \widehat y_{t_k}-y_{t_k}
    +
    \delta\bigl(
        \widehat m_{G,t_k}(\widehat y_{t_k})
        -
        m_G(y_{t_k})
    \bigr)
    +
    \int_{t_k}^{t_{k+1}}
        \bigl(m_G(y_{t_k})-m_G(y_s)\bigr)\,ds .
\]
Hence
\[
    \Delta^y_{k+1}
    \le
    \Delta^y_k+\delta\Delta^m_k
    +
    \frac{1}{\sqrt n}
    \int_{t_k}^{t_{k+1}}
        \|m_G(y_s)-m_G(y_{t_k})\|\,ds .
\]
By the continuity estimate for the tilted mean along the
stochastic-localization trajectory \cite[Lemma 4.9]{el2022sampling-sk},
\[
    \frac{1}{\sqrt n}
    \int_{t_k}^{t_{k+1}}
        \|m_G(y_s)-m_G(y_{t_k})\|\,ds
    \le
    C\delta^{3/2}+\xi_n
\]
with probability $1-o_n(1)$. Therefore
\[
    \Delta^y_{k+1}
    \le
    \Delta^y_k+\delta\Delta^m_k+C\delta^{3/2}+\xi_n .
\]
It remains to provide an analogous inequality for the mean-estimation bound. By the trajectory bound just proved, taking $\rho, \delta = \poly(\varepsilon)$ as in \cref{lem:asl-error-propagation-current-notation}, and following the analysis of \cite[Lemma 4.14]{el2022sampling-sk} for mean estimation, we have
\[
    \Delta^m_{k+1}
    \le
    C\Delta^y_{k+1}
    +
    2\rho\sqrt{t_{k+1}}
    +
    \xi_n .
\]
Here, $C$ is the Lipschitz constant of the AMP iteration, which is $O(1)$ as long as $\beta < \beta_0$ is fixed.
Now, we need to prove that $\Delta_{K}^m$ and $\Delta_K^y$ are $O(\varepsilon)$ given $\delta$ and $\rho$, where $t_K = K\delta = O(\log(1/\varepsilon))$. We upper bound $\Delta_{K}^m$ and $\Delta_K^y$ by writing the following continuous differential-inequality upper bound for the variables $M_t = \Delta_t^m$ and $Y_t =\Delta_t^y$:

\begin{align}
    M_t \leq CY_t + 2\rho\sqrt{t} \\
    dY_t \leq M_t + C_1 \sqrt{\delta}
\end{align}
Since $\rho$ and $\delta$ are polynomially small in $\varepsilon$ and $C_1$ is an absolute constant depending only on $\beta$ (see \cite[Lemma 4.9 and Lemma 4.5]{el2022sampling-sk}),
\[
dY_t \leq CY_t + \alpha 
\]
where $\alpha$ is polynomially small in $\varepsilon$.
Given $Y_0 = 0$, using the exact solution as an upper bound, we conclude that $Y_{K\delta} \leq \frac{\alpha}{C}e^{C K\delta} - \frac{\alpha}{C} = O(\alpha e^{CK\delta}) = \poly(\varepsilon)$.

\end{proof}
\begin{theorem}[Discretization error at logarithmic localization time]
\label{thm:asl-discretization-error-current-notation}
Let the accuracy parameter satisfy $\varepsilon > \varepsilon_n$, and let $t_\varepsilon = O(\log(\frac{1}{\varepsilon})) < t_n$ be the
chosen localization time. Assume, for notational simplicity, that
$K\delta=t_\varepsilon$. Let $\hat{y}_{K\delta}$ be the output of algorithmic stochastic localization using the parameter choices of \cref{prop:log-n-evaluation-of-mean} for TAP-AMP and the $\rho$ and $\delta$ of \cref{lem:asl-error-propagation-current-notation}. Then, with probability $1 - o(1)$, 
\[
    W_{2, n}(\hat{y}_{t_\varepsilon}/t_\varepsilon, y_{t_\varepsilon}/t_\varepsilon) \leq O(\varepsilon) 
\]
\end{theorem}

\begin{proof}
Apply \cref{lem:asl-error-propagation-current-notation} at $k=K$, so that
$K\delta=t_\varepsilon$. With $\rho=\sqrt{\delta}$, we have
\[
        \Delta^m_K
        \leq
        C e^{C t_\varepsilon}t_\varepsilon
        \left(
            \sqrt{\delta}\sqrt{t_\varepsilon}+\sqrt{\delta}
        \right)
        +
        C\sqrt{\delta}\sqrt{t_\varepsilon}
        +
        \xi_n .
\]
Since $t_\varepsilon\geq1$, the trajectory error is bounded by
\[
        \Delta^y_K/t_\varepsilon \leq \poly(\varepsilon)/\log(\frac{1}{\varepsilon}) \leq \poly(\varepsilon)
\]
Using sign rounding (\cref{lem:sl-sign-rounding}) and the stability argument \cref{lem:stability-sl-rounding}, we obtain
\[
    W_{2, n}(\sign{\hat{y}_{t_\varepsilon}}, \mu) \leq O(\varepsilon) 
\]
where $\mu$ denotes the true distribution.
\end{proof}


\section{Parallelization of TAP-AMP Mean Approximation Algorithms}
\label{sec:mean-approximate-parallel-depth}

Here, we prove the $\widetilde O(1)$ parallel-depth claim for the tilted-mean oracle used in
\cref{thm:optimized-algorithmic-stochastic-localization}.
We keep the notation of the main body: $K$ and
$\delta$ denote, respectively, the number of stochastic-localization steps and
the discretization step size, while $K_{\mathrm{AMP}}$ and
$K_{\mathrm{NGD}}$ denote the internal iteration counts of the TAP-AMP tilted
mean estimator. The localization horizon is $t_{\varepsilon}=K\delta$, and by
\cref{thm:optimized-algorithmic-stochastic-localization} we may take
\[
        T_{\varepsilon}=O\!\left(\log(1/\varepsilon)\right),
        \qquad
        \delta=\operatorname{poly}(\varepsilon).
\]
The overlap parameter in the mixed $p$-spin algorithm is denoted by
$q^{\star}(t)$, as in \cref{eq:asymptotic-overlap}. In the SK specialization we
write $q^{\mathrm{SK}}_\star(\beta,t)$ for the corresponding fixed point in
\cite{el2022sampling-sk}; this avoids overloading the mixed $p$-spin notation.

\begin{proof}[Proof of \cref{prop:log-n-evaluation-of-mean}]
Both mean-estimation routines in the main body,
\cref{alg:mean-of-tilted-gibbs-measure-sk} for the SK model and
\cref{alg:tap-amp} for the mixed $p$-spin model, have the same structure: an
AMP phase followed by a natural-gradient-descent phase on the approximate TAP
free energy. We count parallel time in the \Class{PRAM} model with parallel
Hamiltonian-derivative oracle access, as in the oracle model described in the
introduction. Equivalently, for the explicit fixed-degree mixed $p$-spin
polynomial representation, every gradient or Hessian-vector computation is a
fixed-degree tensor contraction and has $\operatorname{polylog}(n)$ parallel
depth and polynomial work by parallel reductions.

Consider one AMP iteration. The coordinate-wise nonlinearities
$\tanh(\cdot)$ and $\sech^2(\cdot)$ are computed independently over the
coordinates. The empirical overlap, Onsager coefficient, and inner products
appearing in the update are computed by parallel reductions in $O(\log n)$
depth. The only model-dependent operation is the evaluation of
$\nabla H(m_k)$, or $A m_k$ in the SK specialization, which is one
parallel Hamiltonian-derivative query. Hence each AMP iteration has
$\operatorname{polylog}(n)$ parallel depth and polynomial work. The same
argument applies to one NGD iteration: it evaluates
$\nabla \hat{\mathscr F}_{\mathrm{TAP}}(m;y,q)$ and then performs only
coordinate-wise operations and parallel reductions.

It remains to bound $K_{\mathrm{AMP}}$ and $K_{\mathrm{NGD}}$. For the mixed
$p$-spin model, the high-temperature TAP-AMP analysis of
\cite{el2025sampling-p-spin} gives fixed constants, depending only on the
mixture $\xi$ and on the requested tilted-mean accuracy, for which the AMP
phase enters the locally strongly convex basin of
$\hat{\mathscr F}_{\mathrm{TAP}}(\cdot;y,q^{\star}(t))$ uniformly along the ASL
trajectory. Taking the auxiliary tilted-mean accuracy to be
$\rho=\operatorname{poly}(\varepsilon)$ as in \cref{lem:asl-error-propagation-current-notation} gives
\[
        K_{\mathrm{AMP}}
        =\operatorname{polylog}(n/\varepsilon)
\]
for the parameter regime used by
\cref{thm:optimized-algorithmic-stochastic-localization}. In the SK case,
\cref{lem:explicit-sk-amp-iteration-bound} below gives the sharper explicit
bound
\[
        K_{\mathrm{AMP}}
        =O_{\beta}\!\left(
            \log(1/\alpha)+\log(1+t_{\varepsilon})
        \right),
	\]
where $\alpha$ is the local AMP accuracy parameter. With
$\alpha=\operatorname{poly}(\varepsilon/n)$ and
$t_{\varepsilon}=O(\log(1/\varepsilon))$, this gives a
$\operatorname{polylog}(n/\varepsilon)$ bound. A similar upper bound can be
proved for the mixed $p$-spin case, as long as the temperature is sufficiently
high, by essentially imitating the proof for the SK case.

For the NGD phase, the TAP analysis supplies a local strong-convexity and
smoothness window around the AMP initialization. Namely, in the relevant
neighborhood of the TAP fixed point, the objective in the natural parameter
$u=\operatorname{arctanh}(m)$ is $\mu$-strongly convex and $M$-smooth, with
$0<\mu\leq M<\infty$ depending only on the high-temperature parameters. Taking
$\eta\leq 1/M$, the standard contraction estimate gives
\[
        \|u_s-u_\star\|_2
        \leq
        (1-\eta\mu)^s\|u_0-u_\star\|_2 .
\]
where $u_\star$ is the minimizer. 
Thus, reducing the NGD error to the auxiliary accuracy
$\rho\sqrt{t_{\varepsilon}n}$ requires
\[
        K_{\mathrm{NGD}}
        =O_{\beta}\!\left(\log(1/\rho)+\log n+\log(1+T_{\varepsilon})\right)
        =\operatorname{polylog}(n/\varepsilon),
\]
again for $\rho=\operatorname{poly}(\varepsilon/n)$.

Since every AMP and NGD iteration has $\operatorname{polylog}(n)$ parallel
depth and the number of such iterations is
$\operatorname{polylog}(n/\varepsilon)$, each tilted-mean call in
\cref{alg:picard-asl} is computable in parallel depth
$\operatorname{polylog}(n/\varepsilon)$, as claimed.
\end{proof}

\begin{lemma}[Explicit SK AMP iteration bound]
\label{lem:explicit-sk-amp-iteration-bound}
Fix $0<\beta<1/2$, a horizon $T\geq 1$, and a local AMP accuracy parameter
$\alpha\in(0,1)$. Let $W\sim\mathcal N(0,1)$ and define
\[
    \operatorname{mmse}(\gamma)
    :=
    1-\E{\tanh\!\left(\gamma+\sqrt{\gamma}\,W\right)^2},
    \qquad \gamma\geq 0.
\]
Let $\gamma_0(\beta,t)=0$ and
\[
    \gamma_{k+1}(\beta,t)
    =
    \beta^2\Bigl(1-
    \operatorname{mmse}\bigl(\gamma_k(\beta,t)+t\bigr)\Bigr),
    \qquad t\in(0,T].
\]
Let $\gamma_\star(\beta,t)$ denote the nonnegative fixed point of this recursion,
and define
\[
    q^{\mathrm{SK}}_k(\beta,t)
    :=\frac{\gamma_{k+1}(\beta,t)}{\beta^2},
    \qquad
    q^{\mathrm{SK}}_\star(\beta,t)
    :=\frac{\gamma_\star(\beta,t)}{\beta^2}.
\]
Set
\[
    a_\beta:=\frac{\beta^2}{1-\beta^2},
    \qquad
    \lambda_\beta:=\frac{1-2\beta}{4},
\]
and, for $x\in\mathbb R$, write
$\lceil x\rceil_+:=\max\{1,\lceil x\rceil\}$. Define
\[
K_q :=
\left\lceil
\frac{\log(16a_\beta/\alpha)}{-2\log\beta}
\right\rceil_+,
\]
\[
K_{g,1}:=
\left\lceil
\frac{\log\!\left(16\sqrt{a_\beta}/(\lambda_\beta\sqrt{\alpha})\right)}
{-\log\beta}
\right\rceil_+,
\]
and
\[
K_{g,2}:=
\left\lceil
\frac{\log\!\left(24a_\beta\sqrt t_\varepsilon/(\lambda_\beta\sqrt{\alpha})\right)}
{-2\log\beta}
\right\rceil_+.
\]
Let
\[
    K^{\mathrm{SK}}_{\mathrm{AMP}}
    :=
    \max\{K_q,K_{g,1},K_{g,2}\}.
\]
For the SK AMP phase, let $(z_k,m_k)_{k\geq 0}$ denote the iterates of
\cref{alg:mean-of-tilted-gibbs-measure-sk} with external field $y(t)$. Thus
\[
    m_{-1}=0,
    \qquad
    z_0=0,
    \qquad
    m_k=\tanh(z_k),
\]
and
\[
    b_k
    :=
    \frac{\beta^2}{n}\sum_{i=1}^n \sech^2\bigl((z_k)_i\bigr),
    \qquad
    z_{k+1}=\beta A m_k+y(t)-b_k m_{k-1},
\]
where $A$ is the SK interaction matrix.
Then $K_{\mathrm{AMP}}=K^{\mathrm{SK}}_{\mathrm{AMP}}$ is sufficient for the
large-$K$ AMP requirements used in \cite[Lemmas 4.10 and 4.11]{el2022sampling-sk}:
\[
    \sup_{t\in(0,T]}
    \sup_{k_1,k_2\geq K_{\mathrm{AMP}}}
    \frac{|q^{\mathrm{SK}}_{k_1}(\beta,t)-q^{\mathrm{SK}}_{k_2}(\beta,t)|}{t}
    \leq
    \frac{\alpha}{16},
\]
and
\[
    \sup_{t\in(0,T]}
    \sup_{q\in[q^{\mathrm{SK}}_{K_{\mathrm{AMP}}}(\beta,t),
                q^{\mathrm{SK}}_\star(\beta,t)]}
    \operatorname*{p-lim}_{n\to\infty}
    \frac{
    \|\nabla \hat{\mathscr F}_{\mathrm{TAP}}(m_{K_{\mathrm{AMP}}};y(t),q)\|_2
    }{\sqrt{tn}}
    \leq
    \frac{\lambda_\beta\sqrt{\alpha}}{4}.
\]
Consequently,
\[
    K^{\mathrm{SK}}_{\mathrm{AMP}}
    =
    O_\beta\!\left(
        \log(1/\alpha)+\log(1+T)
    \right).
\]
\end{lemma}

\begin{proof}
Let
\[
    f_t(\gamma)
    :=
    \beta^2\bigl(1-\operatorname{mmse}(\gamma+t)\bigr).
\]
The recursion is $\gamma_{k+1}(\beta,t)=f_t(\gamma_k(\beta,t))$. Since
$\operatorname{mmse}$ is decreasing on $\mathbb R_{\geq 0}$, $f_t$ is increasing,
and hence
\[
        0=\gamma_0(\beta,t)
        \leq \gamma_1(\beta,t)
        \leq \gamma_2(\beta,t)
        \leq \cdots
        \leq \gamma_\star(\beta,t).
\]
By \cite[Lemma 4.5]{el2022sampling-sk}, for $\beta<1$ and $t>0$,
\[
    0\leq \gamma_\star(\beta,t)-\gamma_k(\beta,t)
    \leq
    \beta^{2k}\gamma_\star(\beta,t),
\]
and
\[
    \frac{\gamma_1(\beta,t)}{\gamma_\star(\beta,t)}
    \geq 1-\beta^2.
\]
Moreover,
$\gamma_1(\beta,t)=\beta^2(1-\operatorname{mmse}(t))\leq \beta^2 t$, using
$1-\operatorname{mmse}(t)\leq t$. Therefore
\[
    \gamma_\star(\beta,t)
    \leq
    \frac{\gamma_1(\beta,t)}{1-\beta^2}
    \leq
    a_\beta t .
\]
Thus
\[
\begin{aligned}
    0
    \leq q^{\mathrm{SK}}_\star(\beta,t)-q^{\mathrm{SK}}_k(\beta,t)
    &=
    \frac{\gamma_\star(\beta,t)-\gamma_{k+1}(\beta,t)}{\beta^2} \\
    &\leq
    \beta^{2k}\gamma_\star(\beta,t)
    \leq
    a_\beta\beta^{2k}t .
\end{aligned}
\]
Since $q^{\mathrm{SK}}_k(\beta,t)$ is nondecreasing in $k$, for any
$k_1,k_2\geq K$,
\[
    |q^{\mathrm{SK}}_{k_1}(\beta,t)-q^{\mathrm{SK}}_{k_2}(\beta,t)|
    \leq
    a_\beta\beta^{2K}t .
\]
The first displayed requirement follows as soon as
$a_\beta\beta^{2K}\leq \alpha/16$, which is guaranteed by $K\geq K_q$.

It remains to control the TAP-gradient term. Let
\[
    \Delta_k(t):=\gamma_{k+1}(\beta,t)-\gamma_k(\beta,t).
\]
The monotonicity of $\gamma_k$ and the preceding bound imply
\[
    0\leq \Delta_k(t)
    \leq
    \gamma_\star(\beta,t)-\gamma_k(\beta,t)
    \leq
    a_\beta\beta^{2k}t .
\]
State evolution for the SK AMP iteration gives
\[
    \operatorname*{p-lim}_{n\to\infty}
    \frac{\|z_{k+1}-z_k\|_2^2}{n}
    =
    \Delta_k(t)^2+\Delta_k(t).
\]
Consequently, uniformly over $t\in(0,T]$,
\[
\begin{aligned}
    \operatorname*{p-lim}_{n\to\infty}
    \frac{\|z_{k+1}-z_k\|_2}{\sqrt{tn}}
    &\leq
    \sqrt{\frac{\Delta_k(t)^2+\Delta_k(t)}{t}} \\
    &\leq
    \sqrt{a_\beta}\,\beta^k
    +
    a_\beta\sqrt T\,\beta^{2k} .
\end{aligned}
\]
Since $\tanh$ is $1$-Lipschitz,
$\|m_k-m_{k-1}\|_2\leq \|z_k-z_{k-1}\|_2$. The TAP-gradient
identity in \cite[Lemma 4.11]{el2022sampling-sk} gives, for
$q\in[q^{\mathrm{SK}}_k(\beta,t),q^{\mathrm{SK}}_\star(\beta,t)]$,
\[
\begin{aligned}
    \frac{\|\nabla \hat{\mathscr F}_{\mathrm{TAP}}(m_k;y(t),q)\|_2}{\sqrt{tn}}
    &\leq
    \frac{\|z_{k+1}-z_k\|_2}{\sqrt{tn}}
    +
    \beta^2
    \frac{\|m_{k-1}-m_k\|_2}{\sqrt{tn}}  \\
    &\qquad
    +
    \beta^2
    \frac{q^{\mathrm{SK}}_\star(\beta,t)-q^{\mathrm{SK}}_k(\beta,t)}{\sqrt t}
    +
    o_{n,\mathbb P}(1).
\end{aligned}
\]
Taking the $p$-limit and using the bounds above yields
\[
\begin{aligned}
&\sup_{t\in(0,T]}
\sup_{q\in[q^{\mathrm{SK}}_k(\beta,t),q^{\mathrm{SK}}_\star(\beta,t)]}
\operatorname*{p-lim}_{n\to\infty}
\frac{\|\nabla \hat{\mathscr F}_{\mathrm{TAP}}(m_k;y(t),q)\|_2}{\sqrt{tn}}
\\
&\qquad\leq
\left(\sqrt{a_\beta}\,\beta^k+a_\beta\sqrt T\,\beta^{2k}\right)
+
\beta^2\left(\sqrt{a_\beta}\,\beta^{k-1}
+a_\beta\sqrt T\,\beta^{2k-2}\right)
+
\beta^2 a_\beta\sqrt T\,\beta^{2k} .
\end{aligned}
\]
Because $0<\beta<1$, the right-hand side is at most
\[
    2\sqrt{a_\beta}\,\beta^k
    +
    3a_\beta\sqrt T\,\beta^{2k}.
\]
It is therefore enough to require
\[
    2\sqrt{a_\beta}\,\beta^K
    \leq
    \frac{\lambda_\beta\sqrt{\alpha}}{8},
    \qquad
    3a_\beta\sqrt T\,\beta^{2K}
    \leq
    \frac{\lambda_\beta\sqrt{\alpha}}{8},
\]
which are guaranteed by $K\geq K_{g,1}$ and $K\geq K_{g,2}$, respectively.
This proves the two large-$K$ requirements. The asymptotic expression for
$K^{\mathrm{SK}}_{\mathrm{AMP}}$ follows because $\beta\in(0,1/2)$ is fixed,
so $-\log\beta$ is bounded away from zero.
\end{proof}

\section{Lipschitz Property of the TAP Fixed-Point Computation Algorithm for Mixed $p$-Spin Models}
\label{subsec:mixed-p-spin-lipschitz}

In this section, we prove that, for the mixed $p$-spin model at sufficiently high temperature, the TAP fixed-point computation algorithm is Lipschitz with respect to the tilt parameter $y$ with high probability. The proof follows the same two-phase structure as in the SK case: first an AMP phase gets close to the TAP fixed point, and then a natural-gradient descent step in the dual variables closes the remaining gap.

We recall the standard notation for the mixed $p$-spin glass. First, the Hamiltonian takes the following form:
\begin{align}
\label{eq:p-spin-model-zero-field}
    H(x) = \sum_{p=2}^{P} \frac{\beta_p \sqrt{p!}}{n^{\tfrac{p-1}{2}}} \cdot \sum_{J \subseteq [n], |J|=p} G_J \ \prod_{k \in J}x_k , \qquad \mu(x) = \frac{1}{Z}\exp(H(x))
\end{align}
where we assume zero external field.
Also recall the associated mixture function
\[
\xi(t)=\sum_{p=2}^{P} \beta_p^2 t^p \qquad \text{equivalently,} \qquad 
    \E{H_n(x)H_n(x')}
    = n\xi\left(\frac{\ip{x}{x'}}{n}\right).
\]
Throughout this section, we condition on the high-probability event that the Hamiltonian has a dimension-free Hessian bound. Namely, for a deterministic constant $L_H$ depending only on $\xi$, we condition on the following event:
\refstepcounter{equation}\label{eq:mixed-good-hessian-event}
\begin{align}
    \sup_{m\in[-1,1]^n}\norm{\nabla^2 H_n(m)}_{\op}\le L_H.
    \tag{Hessian Bound}
\end{align}
Using Gaussian concentration (see \cite[Theorem 3.3]{arous2020algorithmic}) one can prove that the above event holds with probability at least $1-e^{-cn}$.
\begin{definition}[Lipschitz threshold for the mixed $p$-spin model]
\label{def:mixed-p-spin-lipschitz-threshold}
Let

\[
    L_0:=\frac{4}{3\sqrt{3}},
    \qquad
    \phi(t):=(1-t)\xi''(t),
    \qquad
    L_\phi:=\sup_{t\in[0,1]}\abs{\phi'(t)}.
\]
Fix $q\in(0,1)$ and the quadratic TAP parameter $T=T(\xi)$.  Define
\begin{align}    
    \gamma_{\AMP}(L_H,\xi)
    := L_H+\xi''(1)+L_0L_\phi,
    \label{eq:amp-p-spin-lipschitz}
    \tag{AMP-Lipchitz} \\
    \gamma_{\NGD}(L_H,\xi,q,T)
    := L_H+\frac{3T}{2}
       +(1-q)\xi''(q), \label{eq:ngd-p-spin-lipschitz}\tag{NGD-Lipchitz}
\end{align}
and
\[
    \gamma_{\mathrm{mix}}(L_H,\xi,q,T)
    :=\max\left\{\gamma_{\AMP}(L_H,\xi),
                   \gamma_{\NGD}(L_H,\xi,q,T)\right\}.
\]
\end{definition}

The TAP fixed-point computation for the $p$-spin model is very similar to that for the SK model, but it is more complicated because it uses derivatives of the mixture function. Here, we review the TAP fixed-point computation used in \cite{el2023algorithmic-spin}.

\paragraph{Phase 1: AMP phase.}
\label{alg:mixed-p-spin-amp-phase}

We have the following initialization:
\[
    m_{-1}=0,\qquad z_0=0.
\]
For $k=0,1,\dots,K_{\AMP}-1$, define
\[
    m_k = \tanh(z_k),
\]
\[
    q_k = \frac{1}{n}\sum_{i=1}^{n}\tanh^2\bigl((z_k)_i\bigr)
         =\frac{1}{n}\norm{m_k}^2,
\]
\[
    b_k = (1-q_k)\xi''(q_k),
\]
and
\refstepcounter{equation}\label{eq:mixed-amp-update}
\[
    z_{k+1} = \nabla H_n(m_k)+y-b_km_{k-1}.
    \tag{\theequation}
\]

\paragraph{Phase 2: NGD phase on the approximate TAP free energy.}
\label{alg:mixed-p-spin-ngd-phase}
Initialize
\[
    u_0:=z_{K_{\AMP}}.
\]
For $t=0,1,\dots,K_{\NGD}-1$, define
\[
    m_t^+ = \tanh(u_t),
\]
and
\refstepcounter{equation}\label{eq:mixed-ngd-update}
\[
    u_{t+1} = u_t-\eta\nabla_m \hat{\mathscr{F}}_{\TAP}(m_t^+;y,q).
    \tag{\theequation}
\]
The algorithm outputs
\[
    \widehat m(G,y):=m_{K_{\NGD}}^+=\tanh(u_{K_{\NGD}}).
\]

\begin{definition}[Tilted-mean computation for mixed $p$-spin]
\label{def:mixed-tilted-mean-computation}
For a given disorder realization $G$, external field $y\in\R^n$, and iteration counts $K_{\AMP}$ and $K_{\NGD}$, let
\[
    \widehat m_{(K_{\AMP},K_{\NGD})}(G,y)
\]
denote the output of the AMP+NGD algorithm above.  When $G$, $K_{\AMP}$, and $K_{\NGD}$ are clear from context, we simply write $\widehat m(y)$.
\end{definition}

Our goal is to prove that, for any two tilts $y,\widetilde y\in\R^n$, the output satisfies
\[
    \norm{\widehat m(G,y)-\widehat m(G,\widetilde y)}
    \le \frac{1}{1-\gamma_{\mathrm{mix}}}\norm{y-\widetilde y},
\]
provided $\gamma_{\mathrm{mix}}<1$.  The proof uses the following elementary Lipschitz facts:
\begin{itemize}
    \item $\tanh$ is $1$-Lipschitz;
    \item $\tanh^2$ is $L_0=4/(3\sqrt3)$-Lipschitz;
    \item $\phi(t)=(1-t)\xi''(t)$ is $L_\phi$-Lipschitz on $[0,1]$.
\end{itemize}

Throughout the proof, fix two external fields $y$ and $\widetilde y$.  Unadorned variables correspond to the computation with tilt $y$, and tilded variables correspond to the computation with tilt $\widetilde y$.  We use the notation
\[
\begin{aligned}
    \Delta y &:= y-\widetilde y, &
    \Delta z_k &:= z_k-\widetilde z_k, &
    \Delta m_k &:= m_k-\widetilde m_k, \\
    \Delta q_k &:= q_k-\widetilde q_k, &
    \Delta b_k &:= b_k-\widetilde b_k, &
    \Delta u_t &:= u_t-\widetilde u_t.
\end{aligned}
\]
Also, $\Delta z_0=0$, $\Delta z_{-1}=0$, and $m_{-1}=\widetilde m_{-1}=0$ by initialization.

\begin{theorem}[Lipschitzness of AMP phase for mixed $p$-spin]
\label{thm:mixed-amp-lip}
Assume the Hessian bound~\eqref{eq:mixed-good-hessian-event} holds and
\[
    \gamma_{\AMP}(L_H,\xi)
    =L_H+\xi''(1)+L_0L_\phi<1.
\]
Then for all $k\ge 0$,
\[
    \norm{\Delta z_k}\le \frac{1}{1-\gamma_{\AMP}}\norm{\Delta y},
    \qquad
    \norm{\Delta m_k}\le \frac{1}{1-\gamma_{\AMP}}\norm{\Delta y}.
\]
\end{theorem}

\begin{proof}
Let
\[
    V_k:=\max_{0\le j\le k}\norm{\Delta z_j}.
\]
We first prove the one-step inequality
\refstepcounter{equation}\label{eq:mixed-amp-one-step}
\[
    \norm{\Delta z_{k+1}}
    \le
    \left(L_H+L_0L_\phi\right)\norm{\Delta z_k}
    +\xi''(1)\norm{\Delta z_{k-1}}
    +\norm{\Delta y}.
    \tag{\theequation}
\]
Taking the difference of~\eqref{eq:mixed-amp-update} for $y$ and $\widetilde y$ gives
\[
    \Delta z_{k+1}
    =\left[\nabla H_n(m_k)-\nabla H_n(\widetilde m_k)\right]
      +\Delta y-b_k\Delta m_{k-1}-\Delta b_k\,\widetilde m_{k-1}.
\]
Since~\eqref{eq:mixed-good-hessian-event} holds, $\nabla H_n$ is $L_H$-Lipschitz on $[-1,1]^n$.  Since $\tanh$ is $1$-Lipschitz,
\[
    \norm{\Delta m_k}\le \norm{\Delta z_k},
    \qquad
    \norm{\Delta m_{k-1}}\le \norm{\Delta z_{k-1}}.
\]
Moreover,
\[
    0\le b_k=(1-q_k)\xi''(q_k)\le \xi''(1).
\]
It remains to control the Onsager coefficient difference.  Since $b_k=\phi(q_k)$,
\[
    \abs{\Delta b_k}\le L_\phi\abs{\Delta q_k}.
\]
Using the $L_0$-Lipschitzness of $\tanh^2$,
\[
\begin{aligned}
    \abs{\Delta q_k}
    &=\left|\frac{1}{n}\sum_{i=1}^{n}
    \left[\tanh^2\bigl((z_k)_i\bigr)
    -\tanh^2\bigl((\widetilde z_k)_i\bigr)\right]\right| \\
    &\le \frac{L_0}{n}\norm{\Delta z_k}_1
    \le \frac{L_0}{\sqrt n}\norm{\Delta z_k}.
\end{aligned}
\]
Since $\norm{\widetilde m_{k-1}}\le \sqrt n$, we get
\[
    \abs{\Delta b_k}\,\norm{\widetilde m_{k-1}}
    \le L_0L_\phi\norm{\Delta z_k}.
\]
Combining the above estimates with the triangle inequality proves~\eqref{eq:mixed-amp-one-step}.

Therefore,
\[
    \norm{\Delta z_{k+1}}
    \le \gamma_{\AMP} V_k+\norm{\Delta y},
\]
and hence
\[
    V_{k+1}\le \max\{V_k,\gamma_{\AMP}V_k+\norm{\Delta y}\}.
\]
Since $\gamma_{\AMP}<1$, induction gives
\[
    V_k\le \frac{1}{1-\gamma_{\AMP}}\norm{\Delta y}
\]
for every $k$.  The same bound for $\Delta m_k$ follows from the $1$-Lipschitzness of $\tanh$.
\end{proof}

\subsubsection*{Lipschitzness of the full AMP+NGD algorithm for mixed $p$-spin}

We next prove that the NGD phase preserves the AMP Lipschitz bound, up to replacing $\gamma_{\AMP}$ by the larger stability parameter $\gamma_{\mathrm{mix}}$.

\paragraph{Rewrite the gradient of TAP in $u$-space.}
We recall several definitions from the preliminaries regarding the TAP equation.

\begin{align}
    \mathscr{F}_{\mathrm{TAP}}(m,y) &= -H(m) - \dotp{y}{m} - \sum_{i=1}^{n} h(m_i) - \operatorname{ONS}(Q(m)) \\
    \operatorname{ONS}(q) &= \frac{n}{2}\left[\xi(1)-\xi(Q)-(1-Q)\xi^{\prime}(Q)\right] \\
        \ONS'(q)&=-\frac{n}{2}(1-q)\xi''(q)
\\
    Q(m) &= n^{-1}\|m\|^2, \qquad h(m) = -\frac{1+m}{2}\ln\left(\frac{1+m}{2}\right) - \frac{1-m}{2}\ln\left(\frac{1-m}{2}\right) 
\end{align}
A direct calculation gives
\[
\begin{aligned}
    \nabla_m F_{\TAP}(m;y,q)
    &= -\nabla H_n(m)-y+\operatorname{arctanh}(m)
    +(1-q)\xi''(q)m\\
    &\quad +\frac{T}{2}\bigl(Q(m)-q\bigr)m.
\end{aligned}
\]
Plugging in $m=\tanh(u)$ gives
\[
    \nabla_m F_{\TAP}(\tanh(u);y,q)=u-y-\mathcal G(u),
\]
where
\[
    \mathcal G(u)
    :=\nabla H_n(\tanh(u))
    -(1-q)\xi''(q)\tanh(u)
    -\frac{T}{2}J(\tanh(u)),
\]
and
\[
    J(m):=\bigl(Q(m)-q\bigr)m.
\]
Thus the NGD update~\eqref{eq:mixed-ngd-update} becomes
\refstepcounter{equation}\label{eq:mixed-u-update}
\[
    u_{t+1}=(1-\eta)u_t+\eta y+\eta\mathcal G(u_t).
    \tag{\theequation}
\]

\begin{lemma}[Lipschitzness of the quadratic TAP correction]
\label{lem:mixed-J-lip}
For $J(m)=\bigl(Q(m)-q\bigr)m$ with $q\in(0,1)$,
\[
    \sup_{m\in[-1,1]^n}\norm{\nabla J(m)}_{\op}\le 3.
\]
\end{lemma}

\begin{proof}
Since
\[
    \nabla J(m)=\bigl(Q(m)-q\bigr)I+\frac{2mm^\top}{n},
\]
and $Q(m)\in[0,1]$, $q\in(0,1)$, and $\norm{m}\le\sqrt n$, we have
\[
    \norm{\nabla J(m)}_{\op}
    \le 1+\frac{2\norm{m}^2}{n}\le 3.
\]
\end{proof}

\begin{theorem}[TAP Lipschitz property for mixed $p$-spin]
\label{thm:mixed-full-lip}
Assume the Hessian bound~\eqref{eq:mixed-good-hessian-event} holds and
\[
    \gamma_{\mathrm{mix}}(L_H,\xi,q,T)<1.
\]
Let $0<\eta\le 1$.  Set $u_0=z_{K_{\AMP}}$ and $\widetilde u_0=\widetilde z_{K_{\AMP}}$.  Then for every $t\ge 0$,
\[
    \norm{u_t(y)-u_t(\widetilde y)}
    \le \frac{1}{1-\gamma_{\mathrm{mix}}}\norm{y-\widetilde y}.
\]
Consequently, the final output satisfies
\[
    \norm{\widehat m_{(K_\text{AMP}, K_{\text{NGD}})}(G,y)-\widehat m_{(K_\text{AMP}, K_{\text{NGD}})}(G,\widetilde y)}
    \le \frac{1}{1-\gamma_{\mathrm{mix}}}\norm{y-\widetilde y}.
\]
If the NGD phase is run to convergence and the limit is the TAP fixed point $m_{\TAP}(G,y)$, then the same Lipschitz bound holds for the TAP fixed point map $y\mapsto m_{\TAP}(G,y)$.
\end{theorem}

\begin{proof}
By \cref{thm:mixed-amp-lip} and the initialization of NGD,
\[
    \norm{\Delta u_0}
    =\norm{\Delta z_{K_{\AMP}}}
    \le \frac{1}{1-\gamma_{\AMP}}\norm{\Delta y}
    \le \frac{1}{1-\gamma_{\mathrm{mix}}}\norm{\Delta y}.
\]
We now show that the NGD step preserves this bound.

First, $\mathcal G$ is $\gamma_{\NGD}$-Lipschitz.  Indeed, by the Hessian bound, the $1$-Lipschitzness of $\tanh$, and \cref{lem:mixed-J-lip},
\[
\begin{aligned}
    \norm{\mathcal G(u)-\mathcal G(v)}
    &\le L_H\norm{u-v}
    +(1-q)\xi''(q)\norm{u-v}
    +\frac{3T}{2}\norm{u-v} \\
    &=\gamma_{\NGD}\norm{u-v}.
\end{aligned}
\]
Subtracting~\eqref{eq:mixed-u-update} for $y$ and $\widetilde y$ gives
\[
\begin{aligned}
    \norm{\Delta u_{t+1}}
    &\le (1-\eta)\norm{\Delta u_t}
    +\eta\norm{\mathcal G(u_t)-\mathcal G(\widetilde u_t)}
    +\eta\norm{\Delta y} \\
    &\le (1-\eta+\eta\gamma_{\NGD})\norm{\Delta u_t}
    +\eta\norm{\Delta y}.
\end{aligned}
\]
Assume inductively that
\[
    \norm{\Delta u_t}\le \frac{1}{1-\gamma_{\mathrm{mix}}}\norm{\Delta y}.
\]
Since $\gamma_{\NGD}\le\gamma_{\mathrm{mix}}$, we obtain
\[
\begin{aligned}
    \norm{\Delta u_{t+1}}
    &\le \left(1-\eta+\eta\gamma_{\NGD}\right)
    \frac{1}{1-\gamma_{\mathrm{mix}}}\norm{\Delta y}
    +\eta\norm{\Delta y} \\
    &\le \frac{1}{1-\gamma_{\mathrm{mix}}}\norm{\Delta y}.
\end{aligned}
\]
Thus the bound holds for all $t$.  Finally, since $\tanh$ is $1$-Lipschitz,
\[
    \norm{\widehat m(G,y)-\widehat m(G,\widetilde y)}
    \le \norm{u_{K_{\NGD}}(y)-u_{K_{\NGD}}(\widetilde y)}
    \le \frac{1}{1-\gamma_{\mathrm{mix}}}\norm{y-\widetilde y}.
\]
If $u_{K_{\NGD}}(y)$ and $u_{K_{\NGD}}(\widetilde y)$ converge as $K_{\NGD}\to\infty$, the same inequality passes to the limit, giving the TAP fixed-point statement.
\end{proof}

\begin{corollary}[High-probability Lipschitzness for fixed mixed $p$-spin]
\label{cor:mixed-p-spin-wp-lipschitz}
Fix a finite mixture $\xi$, $q\in(0,1)$, and $\varepsilon>0$.
Assume that the temperature coefficients $\{\beta_p\}_{p=2}^{P}$ satisfy
\[
    \mathfrak{C}(\{\beta_p\})
    \coloneqq
    \sum_{p=2}^{P} \beta_p \sqrt{p^3 \log p}
    < \gamma_0,
    \qquad
    \mathfrak{D}(\{\beta_p\})
    \coloneqq
    \sum_{p=2}^{P} \beta_p \sqrt{2^p p^3 \log p}
    < \infty .
\]
Then, with high probability over the disorder realization $G$, the TAP-AMP
mean computation in \cref{alg:tap-amp} is Lipschitz. Specifically, for all
$y,\widetilde y\in\R^n$,
\[
    \norm{\widehat m(G,y)-\widehat m(G,\widetilde y)}
    \le
    O_{\xi,q,T}(1)\,
    \norm{y-\widetilde y}.
\]
Equivalently, the contraction parameter satisfies
\[
    \gamma_{\mathrm{mix}}(L_H,\xi,q,T)
    \coloneqq
    \max\left\{
        \gamma_{\AMP}(L_H,\xi),
        \gamma_{\NGD}(L_H,\xi,q,T)
    \right\}
    < 1 .
\]
\end{corollary}

\begin{proof}
Under the high-temperature assumption, the constants appearing in the
TAP-AMP Lipschitz analysis, namely $L_\phi$, $L_H$, $T$, and the relevant
derivatives of $\xi$, are bounded by a constant $c=c(\xi,q,T)$ which can be
made sufficiently small by taking the temperature coefficients small enough.
In particular, using the bounds from the TAP-AMP analysis,
\begin{align}
    \gamma_{\AMP}(L_H,\xi)
    &= L_H+\xi''(1)+L_0L_\phi \\
    &\leq c\left(2+\frac{4}{3\sqrt{3}}\right).
\end{align}
Similarly,
\begin{align}
    \gamma_{\NGD}(L_H,\xi,q,T)
    &= L_H+\frac{3T}{2}+(1-q)\xi''(q) \\
    &\leq \frac{7}{2}c .
\end{align}
Therefore, choosing the high-temperature constant $\gamma_0$ sufficiently
small ensures that both
\[
    \gamma_{\AMP}(L_H,\xi)<1
    \qquad\text{and}\qquad
    \gamma_{\NGD}(L_H,\xi,q,T)<1 .
\]
Hence
\[
    \gamma_{\mathrm{mix}}(L_H,\xi,q,T)<1 .
\]
By the Lipschitzness criterion for the TAP-AMP mean computation, this implies
that, with high probability over $G$,
\[
    \norm{\widehat m(G,y)-\widehat m(G,\widetilde y)}
    \le
    O_{\xi,q,T}(1)\,
    \norm{y-\widetilde y}
\]
for all $y,\widetilde y\in\R^n$.
\end{proof}

\section{$\exp(\poly(1/\varepsilon))$ Steps for Algorithmic Stochastic Localization}
\label{sec:andrea-runtime}
\begin{proof}[Proof of \cref{prop:andrea-runtime-dependency}]
Let \(t_\varepsilon=K\delta\) be the localization time horizon, using notation consistent with
\textcite{el2022sampling-sk}. 

Their Lemma 4.14 gives,
for the Euler ASL trajectory and with the mean-estimation accuracy parameter chosen
as \(\rho=\sqrt{\delta}\),
\[
B_K
\le
C e^{C t_\varepsilon}
t_\varepsilon\bigl(\rho\sqrt{t_\varepsilon}+\sqrt{\delta}\bigr)
+
C \rho\sqrt{t_\varepsilon}
+o_n(1).
\]
Here, the parameter $C$ is of order $6^{K_\text{AMP}} = \poly(1/\varepsilon)$ (\cref{prop:log-n-evaluation-of-mean}). 
Hence, for \(t_\varepsilon\ge 1\),
\[
B_K
\le
C e^{C t_\varepsilon}t_\varepsilon^{3/2}\sqrt{\delta}
+o_n(1).
\]
Combining this with the localization error \(t_\varepsilon^{-1/2}\), the EMS proof
obtains the pre-rounding bound (\cite[Lemma 4.15]{el2022sampling-sk})
\[
\mathbb E W_{2,n}\bigl(\mu_A,\mathcal L(\widehat m(A,\widehat y_K))\bigr)
\le
t_\varepsilon^{-1/2}
+
C e^{C t_\varepsilon}t_\varepsilon^{3/2}\sqrt{\delta}
+
o_n(1).
\]
To make the final rounded output have \(W_{2,n}\)-error at most \(\varepsilon\), their
rounding lemma requires the pre-rounding error to be \(O(\varepsilon^2)\). Thus, they take their parameters to be 
\[
t_\varepsilon=\operatorname{poly}(1/\varepsilon),
\qquad
\sqrt{\delta}
\le
C \varepsilon^2 e^{-C t_\varepsilon}t_\varepsilon^{-3/2}.
\]
Equivalently,
\[
\delta
\le
C \varepsilon^4 e^{-2C t_\varepsilon}t_\varepsilon^{-3}.
\]
Therefore the number of discretization steps satisfies
\[
K=\frac{t_\varepsilon}{\delta}
\ge
C^{-1}\varepsilon^{-4}
e^{2C t_\varepsilon}t_\varepsilon^{4}
=
\exp\!\bigl(\operatorname{poly}(1/\varepsilon)\bigr).
\]
\end{proof}


@article{eldan2013thin,
  title={Thin shell implies spectral gap up to polylog via a stochastic localization scheme},
  author={Eldan, Ronen},
  journal={Geometric and Functional Analysis},
  volume={23},
  number={2},
  pages={532--569},
  year={2013},
  publisher={Springer}
}

@article{thouless1977solution,
  title={Solution of'solvable model of a spin glass'},
  author={Thouless, David J and Anderson, Philip W and Palmer, Robert G},
  journal={Philosophical Magazine},
  volume={35},
  number={3},
  pages={593--601},
  year={1977},
  publisher={Taylor \& Francis}
}

@book{anderson2010introduction,
  title={An Introduction to Random Matrices},
  author={Anderson, Greg W. and Guionnet, Alice and Zeitouni, Ofer},
  series={Cambridge Studies in Advanced Mathematics},
  volume={118},
  year={2010},
  publisher={Cambridge University Press}
}

@book{talagrand2010mean,
  title={Mean field models for spin glasses: Volume I: Basic examples},
  author={Talagrand, Michel},
  volume={54},
  year={2010},
  publisher={Springer Science \& Business Media}
}

@article{el2025sampling-p-spin,
  title={Sampling from mean-field Gibbs measures via diffusion processes},
  author={El Alaoui, Ahmed and Montanari, Andrea and Sellke, Mark},
  journal={Probability and Mathematical Physics},
  volume={6},
  number={3},
  pages={961--1022},
  year={2025},
  publisher={Mathematical Sciences Publishers}
}

@inproceedings{chen2022localization,
  title={Localization schemes: A framework for proving mixing bounds for Markov chains},
  author={Chen, Yuansi and Eldan, Ronen},
  booktitle={2022 IEEE 63rd Annual Symposium on Foundations of Computer Science (FOCS)},
  pages={110--122},
  year={2022},
  organization={IEEE}
}

@article{galanis2024sampling,
  title={On sampling from Ising models with spectral constraints},
  author={Galanis, Andreas and Kalavasis, Alkis and Kandiros, Anthimos Vardis},
  journal={arXiv preprint arXiv:2407.07645},
  year={2024}
}

@inproceedings{sly2012computational,
  title={The computational hardness of counting in two-spin models on d-regular graphs},
  author={Sly, Allan and Sun, Nike},
  booktitle={2012 IEEE 53rd Annual Symposium on Foundations of Computer Science},
  pages={361--369},
  year={2012},
  organization={IEEE}
}

@inproceedings{galanis2020complexity,
  title={The complexity of approximating averages on bounded-degree graphs},
  author={Galanis, Andreas and {\v{S}}tefankovi{\v{c}}, Daniel and Vigoda, Eric},
  booktitle={2020 IEEE 61st Annual Symposium on Foundations of Computer Science (FOCS)},
  pages={1345--1355},
  year={2020},
  organization={IEEE}
}

@techreport{van2014probability,
  title={Probability in high dimension},
  author={Van Handel, Ramon},
  year={2014}
}

@article{el2022information,
  title={An information-theoretic view of stochastic localization},
  author={El Alaoui, Ahmed and Montanari, Andrea},
  journal={IEEE Transactions on Information Theory},
  volume={68},
  number={11},
  pages={7423--7426},
  year={2022},
  publisher={IEEE}
}

@article{panchenko2012sherrington,
  title={The Sherrington-Kirkpatrick model: an overview},
  author={Panchenko, Dmitry},
  journal={Journal of Statistical Physics},
  volume={149},
  number={2},
  pages={362--383},
  year={2012},
  publisher={Springer}
}

@article{montanari2024friendly,
  title={A friendly tutorial on mean-field spin glass techniques for non-physicists},
  author={Montanari, Andrea and Sen, Subhabrata},
  journal={Foundations and Trends in Machine Learning},
  volume={17},
  number={1},
  pages={1--173},
  year={2024},
  publisher={Emerald Publishing Limited}
}

@article{adhikari2022spectral,
  title={Spectral Gap Estimates for Mixed $ p $-Spin Models at High Temperature},
  author={Adhikari, Arka and Brennecke, Christian and Xu, Changji and Yau, Horng-Tzer},
  journal={arXiv preprint arXiv:2208.07844},
  year={2022}
}

@article{jerrum1986random,
  title={Random generation of combinatorial structures from a uniform distribution},
  author={Jerrum, Mark R and Valiant, Leslie G and Vazirani, Vijay V},
  journal={Theoretical computer science},
  volume={43},
  pages={169--188},
  year={1986},
  publisher={Elsevier}
}

@article{dyer1991random,
  title={A random polynomial-time algorithm for approximating the volume of convex bodies},
  author={Dyer, Martin and Frieze, Alan and Kannan, Ravi},
  journal={Journal of the ACM (JACM)},
  volume={38},
  number={1},
  pages={1--17},
  year={1991},
  publisher={ACM New York, NY, USA}
}

@article{jerrum2004polynomial,
  title={A polynomial-time approximation algorithm for the permanent of a matrix with nonnegative entries},
  author={Jerrum, Mark and Sinclair, Alistair and Vigoda, Eric},
  journal={Journal of the ACM (JACM)},
  volume={51},
  number={4},
  pages={671--697},
  year={2004},
  publisher={ACM New York, NY, USA}
}

@article{glauber1963time,
  title={Time-dependent statistics of the Ising model},
  author={Glauber, Roy J},
  journal={Journal of mathematical physics},
  volume={4},
  number={2},
  pages={294--307},
  year={1963},
  publisher={American Institute of Physics}
}

@inproceedings{anari2022entropic,
  title={Entropic independence: optimal mixing of down-up random walks},
  author={Anari, Nima and Jain, Vishesh and Koehler, Frederic and Pham, Huy Tuan and Vuong, Thuy-Duong},
  booktitle={Proceedings of the 54th Annual ACM SIGACT Symposium on Theory of Computing},
  pages={1418--1430},
  year={2022}
}

@article{celentano2024sudakov,
  title={Sudakov--Fernique post-AMP, and a new proof of the local convexity of the TAP free energy},
  author={Celentano, Michael},
  journal={The Annals of Probability},
  volume={52},
  number={3},
  pages={923--954},
  year={2024},
  publisher={Institute of Mathematical Statistics}
}

@article{gheissari2019spectral,
  title={On the spectral gap of spherical spin glass dynamics},
  author={Gheissari, Reza and Jagannath, Aukosh},
  year={2019}
}

@article{huang2024sampling,
  title={Sampling from spherical spin glasses in total variation via algorithmic stochastic localization},
  author={Huang, Brice and Montanari, Andrea and Pham, Huy Tuan},
  journal={arXiv preprint arXiv:2404.15651},
  year={2024}
}

@article{lai2025principles,
  title={The principles of diffusion models},
  author={Lai, Chieh-Hsin and Song, Yang and Kim, Dongjun and Mitsufuji, Yuki and Ermon, Stefano},
  journal={arXiv preprint arXiv:2510.21890},
  year={2025}
}

@article{shih2023parallel,
  title={Parallel sampling of diffusion models},
  author={Shih, Andy and Belkhale, Suneel and Ermon, Stefano and Sadigh, Dorsa and Anari, Nima},
  journal={Advances in Neural Information Processing Systems},
  volume={36},
  pages={4263--4276},
  year={2023}
}

@article{hu2025diffusion,
  title={Diffusion Models are Secretly Exchangeable: Parallelizing DDPMs via Autospeculation},
  author={Hu, Hengyuan and Das, Aniket and Sadigh, Dorsa and Anari, Nima},
  journal={arXiv preprint arXiv:2505.03983},
  year={2025}
}

@inproceedings{mulmuley1987matching,
  title={Matching is as easy as matrix inversion},
  author={Mulmuley, Ketan and Vazirani, Umesh V and Vazirani, Vijay V},
  booktitle={Proceedings of the nineteenth annual ACM symposium on Theory of computing},
  pages={345--354},
  year={1987}
}

@inproceedings{anari2023parallel,
  title={Parallel discrete sampling via continuous walks},
  author={Anari, Nima and Huang, Yizhi and Liu, Tianyu and Vuong, Thuy-Duong and Xu, Brian and Yu, Katherine},
  booktitle={Proceedings of the 55th Annual ACM Symposium on Theory of Computing},
  pages={103--116},
  year={2023}
}

@article{anari2025parallel,
  title={Parallel Sampling via Autospeculation},
  author={Anari, Nima and Baronio, Carlo and Chen, CJ and Haqi, Alireza and Koehler, Frederic and Li, Anqi and Vuong, Thuy-Duong},
  journal={arXiv preprint arXiv:2511.07869},
  year={2025}
}

@inproceedings{csanky1975fast,
  title={Fast parallel matrix inversion algorithms},
  author={Csanky, Laszlo},
  booktitle={16th Annual Symposium on Foundations of Computer Science (sfcs 1975)},
  pages={11--12},
  year={1975},
  organization={IEEE}
}

@inproceedings{anari2024parallel,
  title={Parallel sampling via counting},
  author={Anari, Nima and Gao, Ruiquan and Rubinstein, Aviad},
  booktitle={Proceedings of the 56th Annual ACM Symposium on Theory of Computing},
  pages={537--548},
  year={2024}
}

@inproceedings{feng2021distributed,
  title={Distributed metropolis sampler with optimal parallelism},
  author={Feng, Weiming and Hayes, Thomas P and Yin, Yitong},
  booktitle={Proceedings of the 2021 ACM-SIAM Symposium on Discrete Algorithms (SODA)},
  pages={2121--2140},
  year={2021},
  organization={SIAM}
}

@inproceedings{liu2022simple,
  title={Simple parallel algorithms for single-site dynamics},
  author={Liu, Hongyang and Yin, Yitong},
  booktitle={Proceedings of the 54th Annual ACM SIGACT Symposium on Theory of Computing},
  pages={1431--1444},
  year={2022}
}

@article{arous2020algorithmic,
  title={Algorithmic thresholds for tensor PCA},
  author={Arous, Gerard Ben and Gheissari, Reza and Jagannath, Aukosh},
  journal={The Annals of Probability},
  volume={48},
  number={4},
  pages={2052--2087},
  year={2020},
  publisher={JSTOR}
}

@inproceedings{huang2025weak,
  title={Weak Poincar{\'e} inequalities, simulated annealing, and sampling from spherical spin glasses},
  author={Huang, Brice and Mohanty, Sidhanth and Rajaraman, Amit and Wu, David X},
  booktitle={Proceedings of the 57th Annual ACM Symposium on Theory of Computing},
  pages={915--923},
  year={2025}
}

@article{dobrushin1968description,
  title={The description of a random field by means of conditional probabilities and conditions of its regularity},
  author={Dobruschin, PL},
  journal={Theory of Probability \& Its Applications},
  volume={13},
  number={2},
  pages={197--224},
  year={1968},
  publisher={SIAM}
}

@article{parisi1981correlation,
  title={Correlation functions and computer simulations},
  author={Parisi, Giorgio},
  journal={Nuclear Physics B},
  volume={180},
  number={3},
  pages={378--384},
  year={1981},
  publisher={Elsevier}
}

@misc{mezard1988spin,
  title={Spin glass theory and beyond},
  author={M{\'e}zard, Marc and Parisi, Giorgio and Virasoro, Miguel Angel and Thouless, David J},
  year={1988},
  publisher={American Institute of Physics}
}

@article{lee2024eldan,
  title={Eldan's stochastic localization and the KLS conjecture: Isoperimetry, concentration and mixing},
  author={Lee, Yin Tat and Vempala, Santosh S},
  journal={Annals of Mathematics},
  volume={199},
  number={3},
  pages={1043--1092},
  year={2024},
  publisher={Department of Mathematics, Princeton University Princeton, New Jersey, USA}
}

@book{perko2013differential,
  title={Differential equations and dynamical systems},
  author={Perko, Lawrence},
  volume={7},
  year={2013},
  publisher={Springer Science \& Business Media}
}

@inproceedings{anari2024fast,
  title={Fast parallel sampling under isoperimetry},
  author={Anari, Nima and Chewi, Sinho and Vuong, Thuy-Duong},
  booktitle={The Thirty Seventh Annual Conference on Learning Theory},
  pages={161--185},
  year={2024},
  organization={PMLR}
}

@article{eldan2022spectral,
  title={A spectral condition for spectral gap: fast mixing in high-temperature Ising models},
  author={Eldan, Ronen and Koehler, Frederic and Zeitouni, Ofer},
  journal={Probability theory and related fields},
  volume={182},
  number={3},
  pages={1035--1051},
  year={2022},
  publisher={Springer}
}

@book{vershynin2018high,
  title={High-dimensional probability: An introduction with applications in data science},
  author={Vershynin, Roman},
  volume={47},
  year={2018},
  publisher={Cambridge university press}
}

@article{chen2025efficient,
  title={Efficient Parallel Ising Samplers via Localization Schemes},
  author={Chen, Xiaoyu and Liu, Hongyang and Yin, Yitong and Zhang, Xinyuan},
  journal={arXiv preprint arXiv:2505.05185},
  year={2025}
}

@inproceedings{fan2023improved,
  title={Improved dimension dependence of a proximal algorithm for sampling},
  author={Fan, Jiaojiao and Yuan, Bo and Chen, Yongxin},
  booktitle={The Thirty Sixth Annual Conference on Learning Theory},
  pages={1473--1521},
  year={2023},
  organization={PMLR}
}

@article{bauerschmidt2019very,
  title={A very simple proof of the LSI for high temperature spin systems},
  author={Bauerschmidt, Roland and Bodineau, Thierry},
  journal={Journal of Functional Analysis},
  volume={276},
  number={8},
  pages={2582--2588},
  year={2019},
  publisher={Elsevier}
}

@article{sambale2019modified,
  title={Modified log-Sobolev inequalities and two-level concentration},
  author={Sambale, Holger and Sinulis, Arthur},
  journal={arXiv preprint arXiv:1905.06137},
  year={2019}
}

@article{lee2023parallelising,
  title={Parallelising glauber dynamics},
  author={Lee, Holden},
  journal={arXiv preprint arXiv:2307.07131},
  year={2023}
}

@inproceedings{el2022sampling-sk,
  title={Sampling from the Sherrington-Kirkpatrick Gibbs measure via algorithmic stochastic localization},
  author={El Alaoui, Ahmed and Montanari, Andrea and Sellke, Mark},
  booktitle={2022 IEEE 63rd Annual Symposium on Foundations of Computer Science (FOCS)},
  pages={323--334},
  year={2022},
  organization={IEEE}
}

@inproceedings{anari2024trickle,
  title={Trickle-down in localization schemes and applications},
  author={Anari, Nima and Koehler, Frederic and Vuong, Thuy-Duong},
  booktitle={Proceedings of the 56th Annual ACM Symposium on Theory of Computing},
  pages={1094--1105},
  year={2024}
}

@article{el2023algorithmic-spin,
  title={Algorithmic stochastic localization for sampling from the p-spin model},
  author={El Alaoui, Ahmed and Montanari, Andrea and Sellke, Mark},
  journal={arXiv preprint arXiv:2307.07130},
  year={2023}
}

@inproceedings{anari2024universality,
  title={Universality of spectral independence with applications to fast mixing in spin glasses},
  author={Anari, Nima and Jain, Vishesh and Koehler, Frederic and Pham, Huy Tuan and Vuong, Thuy-Duong},
  booktitle={Proceedings of the 2024 Annual ACM-SIAM Symposium on Discrete Algorithms (SODA)},
  pages={5029--5056},
  year={2024},
  organization={SIAM}
}
\end{document}